%% file: cr-oopsla.tex
\PassOptionsToPackage{table}{xcolor}
\newif\iffull\fullfalse
\fulltrue
\documentclass[acmsmall]{acmart} 

\title[]{\safetree: Expressive Tree Policies for Microservices}
\usepackage[normalem]{ulem} % Enables \sout
\usepackage[safe]{tipa}
\usepackage{forest}
\usepackage{booktabs}   %% For formal tables:
\usepackage{subcaption} %% For complex figures with subfigures/subcaptions
\newcommand{\node}[0]{\textsf{node}}
\usepackage{tcolorbox}
\usepackage{syntax}
\usepackage{listings}
\usepackage{comment}
\makeatletter
\@addtoreset{equation}{enumi}
\makeatother
\usepackage{mathpartir}

\makeatletter
\let\old@lstKV@SwitchCases\lstKV@SwitchCases
\def\lstKV@SwitchCases#1#2#3{}
\makeatother
\usepackage{lstlinebgrd}
\makeatletter
\let\lstKV@SwitchCases\old@lstKV@SwitchCases

\lst@Key{numbers}{none}{%
    \def\lst@PlaceNumber{\lst@linebgrd}%
    \lstKV@SwitchCases{#1}%
    {none:\\%
     left:\def\lst@PlaceNumber{\llap{\normalfont
                \lst@numberstyle{\thelstnumber}\kern\lst@numbersep}\lst@linebgrd}\\%
     right:\def\lst@PlaceNumber{\rlap{\normalfont
                \kern\linewidth \kern\lst@numbersep
                \lst@numberstyle{\thelstnumber}}\lst@linebgrd}%
    }{\PackageError{Listings}{Numbers #1 unknown}\@ehc}}
\makeatother
\usepackage{lstlinebgrd}
\usepackage{hyperref}
\usepackage{cleveref}
\Crefname{figure}{Fig.}{Figs.}

\newcommand{\kupdates}[1]{}
\newcommand{\kcomments}[1]{{}}

\usepackage{titletoc}
\usepackage{multirow}
\usepackage{makecell}
\makeatletter
\newcommand{\bigplus}{%
  \DOTSB\mathop{\mathpalette\mattos@bigplus\relax}\slimits@
}
\newcommand\mattos@bigplus[2]{%
  \vcenter{\hbox{%
    \sbox\z@{$#1\sum$}%
    \resizebox{!}{0.9\dimexpr\ht\z@+\dp\z@}{\raisebox{\depth}{$\m@th#1+$}}%
  }}%
  \vphantom{\sum}%
}

\usepackage{custommacro}
\setcopyright{cc}
\setcctype{by}
\acmDOI{10.1145/3763127}
\acmYear{2025}
\acmJournal{PACMPL}
\acmVolume{9}
\acmNumber{OOPSLA2}
\acmArticle{349}
\acmMonth{10}
\received{2025-03-26}
\received[accepted]{2025-08-12}
\begin{document}
\title{SafeTree: Expressive Tree Policies for Microservices}
\iffull

\else
\titlenote{This is the conference version of the paper. We defer technical
details and proofs to the full version of the
paper.}
\fi
\author{Karuna Grewal}
\orcid{0009-0008-1280-6892}
\affiliation{%
  \institution{Cornell University}
  \city{Ithaca}
  \country{USA}
}
\email{kgrewal@cs.cornell.edu}
\author{Brighten Godfrey}
\orcid{0009-0003-2930-1982}
\affiliation{%
  \institution{University of Illinois at Urbana-Champaign}
  \city{Urbana-Champaign}
  \country{USA}
}
\email{pbg@illinois.edu}

\author{Justin Hsu}
\orcid{0000-0002-8953-7060}
\affiliation{%
  \institution{Cornell University}
  \city{Ithaca}
  \country{USA}
}
\email{justin@cs.cornell.edu}

\begin{abstract}

A microservice-based application is composed of multiple self-contained components called \textit{microservices}, and controlling inter-service communication is important for enforcing safety properties. Presently, inter-service communication is configured using microservice deployment tools. However, such tools only support a limited class of  \textit{single-hop} policies, which can be overly permissive because they ignore the rich \textit{service tree} structure of microservice calls. Policies that can express the service tree structure can offer development and security teams more fine-grained control over communication patterns. 

To this end, we design an expressive policy language to specify service tree structures, and we develop a \textit{visibly pushdown automata}-based dynamic enforcement mechanism to enforce \textit{service tree} policies. Our technique is non-invasive: it does not require any changes to service implementations, and does not require access to microservice code. To realize our method, we build a runtime monitor on top of a \textit{service mesh}, an emerging network infrastructure layer that can control inter-service communication during deployment. In particular, we employ the programmable network traffic filtering capabilities of Istio, a popular service mesh implementation, to implement an online and distributed monitor.  Our experiments show that our monitor can enforce rich safety properties while adding minimal latency overhead on the order of milliseconds.

\end{abstract}

\vspace{-0.1in}
\maketitle

\section{Introduction} \label{sec:introduction}
Large-scale cloud-based applications are often implemented using the
\textit{microservice} design paradigm, where the application is decomposed into
multiple microservices---that is, services that are loosely coupled and individually have narrowly-defined roles. This design offers separation of concern between the services: individual
services can be developed by independent teams, exposing their service's
functionality over well-defined API interfaces. Each service runs in isolation
in its own runtime-environment called a \textit{container}, listening for incoming traffic on a dedicated port and IP address. Inter-service communication happens over a communication protocol like HTTP or gRPC.

Controlling inter-service communications is important for enforcing safety properties in microservice applications.  For example, a service deployment team may want to split the flow of requests to two versions of their service for A/B testing. In a security-critical setting, a team might want to restrict communication between certain services to enforce security guidelines of their company. Similarly, an audit or a data-compliance team may want to enforce data-protection regulations, like GDPR \cite{gdpr} and  HIPAA \cite{hipaa}, by controlling inter-service exchange of information.

\subsection*{Challenges}
\paragraph*{Limited expressiveness of existing policies}
Today, developers can control inter-service communications with two primary methods: coarse-grained control, and fine-grained control. The first method is exemplified by Kubernetes, the most widely used container orchestration framework, which uses a container network interface (CNI) to control which containers (i.e., services) can communicate with each other. While effective, this kind of policy offers very limited expressivity: communication between services is either fully unrestricted, or entirely prohibited.

The need to offer more fine-grained control of inter-service communications has motivated the use of \textit{service meshes} (e.g., Istio \cite{istio}), which handle request-level communications on behalf of services. The data plane of the service mesh (often realized as the Envoy~\cite{envoydoc} proxy, with one instance of Envoy paired with each service instance) essentially forms a layer between the application and transport protocols, handling incoming and outgoing requests to provide API-level access control, encryption, visibility, load balancing, etc.  Service mesh can be used for monitoring inter-service communication at the API call granularity and enforcing policies, for instance, allowing or rejecting a request based on some HTTP source header value.

Both the above methods only support \emph{single-hop} policies, specifying if a pair of endpoints can communicate (e.g., if service $A$ is allows to call a certain API of another service $B$). 
However, in a microservice application,  a single initial request results in a \emph{service tree} of requests---services make API calls to multiple other services, which in turn call other services. 
Expressing this service tree structure can offer even more fine-grained control
over the communication patterns.

For example, consider a hospital management application where a request for
medical test involves interactions between three services: (1) \test, ~which
receives the request; (2) \obfuscate, which de-identifies patient information;
and (3) \lab, which sends the patient records to an external lab. HIPAA's
data-protection regulations mandate that to preserve the privacy of a  patient,
unnecessary personal health information must be de-identified \citep{hhsdeid}.
Therefore, a compliance team might require that the test functionality calls an
\obfuscate service before \lab.
% This policy can be specified by requiring that \test~invokes \obfuscate~and then invokes \lab.
However, this property cannot be specified as a single-hop policy between the
\test~ and the \lab~ services---it requires reasoning about intermediate service
interactions between these services. In particular, we need \obfuscate~ and
\lab~ services to be invoked by \test, and in that order. Such applications
necessitate policies that can specify the structure of a service tree.

\vspace{-0.08in}
\paragraph*{Enforcement requirements}
Not only are there challenges in expressing rich policies,
enforcing policies in this setting is also challenging. For example,
safety properties are often specified and maintained by teams, like compliance or
deployment, that do not have access to the service code. Therefore, services in a microservice
application will often appear as blackboxes to teams, and a policy enforcement
mechanism should be \textit{non-invasive}, \textit{i.e.,} not require code changes, and \textit{blackbox}, \textit{i.e.,} not require access to code.
% These requirements allow the enforcement mechanism to work within the scope of existing infrastructure components, specifically service meshes.
Furthermore, the inter-service communication patterns of an
application can change due to dynamic updates of service code or due to to
elastic scaling of the application, where service containers can spin up or down
in response to the load on the application. Thus, the enforcement mechanism
should be able to cope with microservice applications that change dynamically,
rather than being fixed from the outset.

\subsection*{Our approach}
In this work, we consider the question: \textit{``How can we specify and enforce
policies over the rich service tree structure at runtime without invasive
changes to the blackbox service implementation?''} 

Our solution consists of three parts: a policy language, an automata-based enforcement mechanism, and a distributed runtime monitor. The black-box and non-invasive aspects of our solution are crucial for usability. For instance, our policies can be fully decoupled from the application code, enabling them to be written and maintained by teams that do not have access to the service code (e.g., deployment or compliance teams). Accordingly, our solution supports polyglot applications. 

\paragraph*{Expressive policy language for service trees}
First, we design a policy language called \safetree for specifying allowed service tree
structure. Our language offers constructs to specify constraints on the
children, siblings, and subtrees of a service in the tree. \safetree also allows
multi-hop policy over a linear sequence of API calls, without
any reference to the tree structure.

\paragraph*{An automata-based enforcement mechanism}
Second, we design an automata-based distributed runtime monitor for \safetree
policies.  The key idea is that a service tree can be represented as a
\emph{nested word} \cite{nestedwords}, so that policies correspond to sets of
nested words.  Accordingly, we define a compilation procedure from a policy into
a \textit{visibly pushdown automaton} (VPA) \cite{vpa} that accepts the set of
valid nested-words described by the policy. Our VPA-based monitor can be
implemented in a fully distributed manner with no need for a centralized
authority, by carrying the VPA state in a custom configuration header along with
the requests and simulating the VPA transitions locally at the services.

\paragraph*{A distributed runtime monitor}
Lastly, to implement our policy checking mechanism, we develop a prototype
implementation of our monitor on top of the Istio service mesh. In Istio, each
service container is paired with a \textit{sidecar} container running an
Envoy proxy that can be programmed to specify custom traffic
filtering logic involving operations on HTTP headers. We exploit this
customizability to implement a runtime monitor locally at services, simulating
the VPA in a distributed fashion.

% which helps reduce the overhead of carrying a stack.
% JH: this previous phrase doesn't relate to distributed fashion.

\kcomments{To the best of our knowledge, there is no tool that compiles a specification in a specification language to an output automaton that is independent of the application code to be monitored. The closest work is PAL \citep{pal} which instruments a C program with a monitor that follows the semantics of an NWA, but it does not generate an NWA output independent of the program.}

\paragraph*{Outline}
We first motivate a policy language for service trees
(\cref{sec:motivation}) and introduce background on nested words
(\cref{sec:refresher}). We then introduce our primary technical contributions: 
\begin{enumerate}
    \item a policy language with constructs to model the service tree structure and a nested-word semantics of the policies (\cref{sec:policy});
    \item a translation of our policies into VPA and a translation of a VPA into
      a distributed monitor, along with a proof of soundness (\cref{sec:semantics});
    \item a broad range of case studies demonstrating real-world policies that can be
      expressed in our policy language (\cref{sec:case-studies});
    \item and an implementation and evaluation of an efficient and distributed
      runtime monitor over service mesh
      (\cref{sec:implementation,sec:evaluation}).
\end{enumerate}
Finally, we survey related work (\cref{sec:related}) and conclude with future directions (\cref{sec:conclusion}).

\section{A Tour of Service Tree Policies}\label{sec:motivation}
This section motivates service tree policies in the context of a hospital management application.

\subsection{Running example: hospital management application}
Consider a microservice-based hospital management application offering functionalities to request (a) a payment of bills, and (b) a medical test from an external lab. These functionalities are implemented using the following services: \textbf{\front}, \symfront: requests for an appointment or a medical test; \textbf{\test}, \symtest: requests for a lab test; \textbf{\obfuscate}, \symobfuscate: de-identifies personal health information; \textbf{\lab}, \symlab: sends patient records to an external lab; \textbf{\payment}, \sympayment: charges a patient for a service; \textbf{\database}, \symdb: writes to a payment database; \textbf{\eventlogging}, \symlogging~: logs database access events.

In a microservice setting, services communicate using API calls over some
communication protocol, like HTTP or gRPC, etc. In this paper, we focus on
\emph{synchronous} HTTP-based APIs, which follow a request/response pattern.
An API request initiates calls to a fleet of services, which may themselves call further services.  The runtime execution trace of the application when serving a single request to a service  can be viewed as an ordered \emph{service tree} rooted at a node corresponding to the requested API. Each API invoked in a single execution trace is represented as a node in the tree; an edge from a parent to a child node implies the parent endpoint called the child endpoint; the children of a node are ordered according to the order in which they are called by the parent.
Since this runtime behavior can depend on arguments to the call,
responses from child calls, local state at the microservice, etc., a call to a
given API might lead to multiple possible service trees.

To illustrate, let us consider possible service trees for two \front~ functionalities: requesting a medical test, or requesting a payment.
\begin{enumerate}
\item \textit{\textbf{Medical test functionality}}: the execution trace when
  \front sends a request to  \test is represented as a tree in
  \cref{fig:svc-test}. First, the \front~ invokes \test: this is represented as
  the edge \circleblack{1} between a node labeled with \front and \test. Further
  API calls invoked by \test while serving this incoming request are represented
  as outgoing edges from this \test node. In this case, \test first calls
  \circleblack{2} \obfuscate~and then \test calls \circleblack{3} \lab. The
  left-to-right ordering of the children of \test in this tree reflects the
  order of children calls.

\item \textit{\textbf{Payment functionality}}: the service tree in
  \cref{fig:svc-payment} describes the execution trace when \frontend~ requests
  \payment to charge for two bills: first, \circleblack{1} \front~  invokes
  \payment. This service then invokes \circleblack{2} \database~, which then
  invokes \circleblack{3} \eventlogging~. Then \payment~ invokes \circleblack{4}
  \database~ for the second time, before \database~ invokes \circleblack{5} \eventlogging~. 
\end{enumerate}

\begin{figure}
\begin{subfigure}{0.3 \textwidth}
\centering
\begin{tikzpicture}[font=\scriptsize,
			baseline=1ex,shorten >=.1pt,node distance=15mm,on grid,
			semithick,auto,
			every state/.style={fill=white,draw=black,circular
					drop shadow,inner sep=.10mm,text=black,minimum size=0.5cm},
			accepting/.style={fill=gray,text=white}]
		\node (F) {$F$};
		\node (T) [below = of F] {$T$};
        \node (O) [below left = of T] {$D$};
		\node (L) [below right = of T] {$L$};
		\path [-stealth, thick]
		(F) edge [left] node {\circleblack{1}}(T)
		(T) edge [above left] node {\circleblack{2}}(O)
		(T) edge [] node {\circleblack{3}}(L)
		;
	\end{tikzpicture}
 \caption{}
 \label{fig:svc-test}
\end{subfigure}
\hfill
\begin{subfigure}{0.3\textwidth}  
\centering
\begin{tikzpicture}[font=\scriptsize,
			baseline=1ex,shorten >=.1pt,node distance=15mm,on grid,
			semithick,auto,
			every state/.style={fill=white,draw=black,circular
					drop shadow,inner sep=.10mm,text=black,minimum size=0.5cm},
			accepting/.style={fill=gray,text=white}]
		\node (F) [] {$F$};
        \node (A) [below = of F] {$P$};
		\node (P1) [below left = of A] {$D$};
		\node (E1) [below = of P1] {$E$};
  		\node (P2) [below right = of A] {$D$};
		\node (E2) [below = of P2] {$E$};
		\path [-stealth, thick]
        (F) edge [left] node{\circleblack{1}}(A)
		(A) edge [above left] node {\circleblack{2}}(P1)
		(P1) edge [left] node {\circleblack{3}}(E1)
  		(A) edge [] node {\circleblack{4}}(P2)
		(P2) edge [] node {\circleblack{5}}(E2)
		;
	\end{tikzpicture}
 \caption{}
  \label{fig:svc-payment}
\end{subfigure}
\hfill
\begin{subfigure}{0.3\textwidth}
\centering
\begin{tikzpicture}[font=\scriptsize,
			baseline=1ex,shorten >=.1pt,node distance=15mm,on grid,
			semithick,auto,
			every state/.style={fill=white,draw=black,circular
					drop shadow,inner sep=.10mm,text=black,minimum size=0.5cm},
			accepting/.style={fill=gray,text=white}]
		\node (F) [] {$F$};
        \node (A) [below = of F] {$P$};
		\node (P1) [below left = of A] {$D$};
		\node (E1) [below = of P1] {$E$};
  		\node (P2) [below right = of A] {$D$};
		\node (E2) [below = of P2] {$E$};
		\node (Any) [below = of A] {\parbox{5pt}{\center{$E$\\ \circleblack{3}}}};
		\path [-stealth, thick]
        (F) edge [] node{}(A)
		(A) edge [] node {}(P1)
		(P1) edge [] node {\circleblack{1}}(E1)
  		(A) edge [] node {}(P2)
		(P2) edge [] node {\circleblack{2}}(E2)
		(A) edge [] node {}(Any)
		;
		\begin{scope}[on background layer]
    \node[fit=(A) (P1) (E1) (P2) (E2), draw, rectangle, rounded corners, inner sep=10pt, fill=gray!20] {};
    \node[fit=(P1) (E1) , draw, rectangle, rounded corners, inner sep=10pt, fill=blue!20] {};
    \node[fit=(P2) (E2) , draw, rectangle, rounded corners, inner sep=10pt, fill=blue!20] {};
        \node[fit=(Any), draw, rectangle, rounded corners, inner sep=5pt, fill=blue!20] {};
  \end{scope}
	\end{tikzpicture}
    \caption{}
    \label{fig:path}
\end{subfigure}
\caption{Service tree for: (a) ``lab test'' request; (b) ``appointment'' request; (c) tree with annotated paths.}
\end{figure}
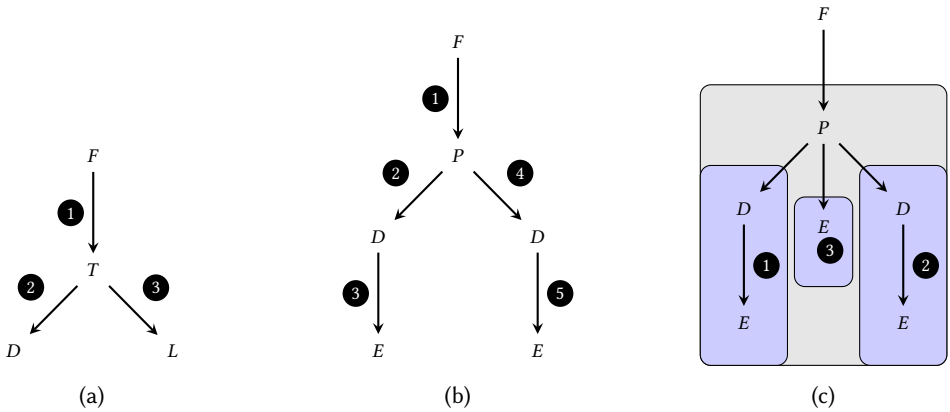
\subsection{Safety Properties}
% Safety properties in a microservice setting can be enforced by controlling
% inter-service communications.
In our example application, we can imagine several safety policies motivated by
regulatory, business, and security concerns.
\begin{enumerate}
\item \textit{\textbf{Deployment team: A/B testing.}}
  Suppose there are two versions of \payment~and \eventlogging~services, \textsf{v1} and
  \textsf{v2}. For A/B testing, a deployment team may want \payment~ requests from beta
  testers to be served by \textsf{v2} of the \payment ~service and any direct or indirect \eventlogging~ invoked by \textsf{v2} of \payment~should also be of version \textsf{v2}. The labeling of beta users is specified in the \front service, while the services undergoing A/B testing can be multiple hops deeper in the call tree.

\item  \label{list:payment} \textit{\textbf{Security team: Log database access.}}
  If \payment~service accesses the \database ~(as shown in \cref{fig:svc-payment}) to update patient's payment details, this sensitive update should be logged by \database ~ calling \eventlogging~.
\item \label{list:compliance} \textit{\textbf{Data compliance team: HIPAA compliance.}} 
In compliance with the HIPAA guidelines to protect patient privacy from an external lab, the hospital's data-compliance team might want to ensure that a
request to \test~ service (as shown in \cref{fig:svc-test}) should be processed by first calling \obfuscate and then calling the external \lab~.
\end{enumerate}

These safety properties are often specified and maintained by teams that do not
have direct access to the service code.  Therefore, services in the application
will often appear as blackboxes to teams, who may only have visibility into an
application's execution by observing inter-service communication patterns. Since
examining application code is outside the scope of service meshes, we can only
specify policies at the granularity of inter-service communication.  For
example, the data compliance policy above may not be able to enforce the
specific data that is passed from De-identify to Lab.

Existing microservice frameworks allow control over inter-service communication
via \emph{single-hop} constraints, which control calls between single pairs of
services. However, the above policies cannot be directly expressed in terms of
such policies since they involve constraints on the service tree structure,
e.g., descendants, subtrees, etc.
Thus, we need a policy language and enforcement mechanism that is more
expressive than current CNI and service mesh policies.

\subsection{Service tree policies}
To address the above issue, we design \safetree---a policy language for specifying service trees to increase the expressiveness of the policies that can be enforced  strictly using the inter-service communication information available at the service mesh layer. To get a flavor of our language, let us specify the payment database logging policy
(\ref{list:payment}), which requires that any requests from \payment~ to \database~
must call \eventlogging~. In
our language, this policy can be specified as: 
\[\underbrace{\redbox{\match ~\payment}}_{\textbf{1. root}} \allpath \underbrace{(\redbox{\database~~ \eventlogging~~ \kleenestar) ~+~(\notservice{\database}~)^*}}_{\textbf{2. path regex}} \label{lbl:spec}\]
For brevity, we defer some syntactic details to \cref{sec:policy}. For now, we will step through the above specification to understand how different parts of the policy have been expressed:

\begin{itemize}
\item Since the property is described with respect to a   subtree rooted at \payment~ (as marked in the tree in \cref{fig:path}), we need to express the root of the subtree. This is specified as the expression  ``\match ~\payment''  on the left.
\item The condition that \database~invoked by \payment~should always invoke
  \eventlogging~ is described in terms of a regular expression over the
  application endpoints. This is specified as the right-hand side expression
  $(\database~~ \eventlogging~~ \kleenestar) ~+~(\notservice{\database}~)^*$.
  Here, $\kleenestar$ denotes any sequence of the application endpoints and
  $\notservice{\textsf{Database}}$ denotes any endpoint besides \database~. This regular expression describes the valid sequence of API calls in a path  from \payment's children  to any leaf node. For instance, in the tree in \cref{fig:path}, this regular expression  matches the paths labeled with \circleblack{1}, \circleblack{2}, and \circleblack{3}.
\end{itemize}

\subsection{Policy Enforcement}
 The blackbox treatment of services calls for a non-invasive enforcement mechanism that does not require code modification. In our work, the service tree policies are enforced using a distributed runtime monitor based
on \textit{visibly  pushdown automata} (VPA) \citep{vpa}. 

%Our enforcement mechanism does not
%require code changes to the services (as long as they implement context propagation;  \S\ref{sec:implementation}).
\paragraph*{\textbf{Overview: visibly pushdown automata}}
A VPA  is a restricted type of pushdown automaton where the stack operation is
determined by the symbol being read. For instance, the VPA in
\cref{fig:VPA-payment} has transitions with two kinds of labels: (a) ``\call~
\service/ \push'' and (b) ``\ret~ \service, \pop'', where \service~ is the name
of the service endpoint and \push~ and \pop~ are values to be pushed and popped
on the stack respectively. A \service~ symbol tagged with ``\call'' corresponds
to an HTTP request to \service, and similarly ``\ret'' tags HTTP responses from
\service. Intuitively, the VPA reads a string over \call/\ret ~ tagged symbols,
say ``$\call~\sympayment; \call~\symdb;\ldots;\ret~ \sympayment$'', from left to
right. On a ``\call'' symbol, the VPA pushes a value on the stack wand transitioning to the next state; on a \ret~ symbol, the VPA pops the top of the stack and moves to the next state. For instance, in \cref{fig:VPA-payment}, the edge labeled ``\call~\sympayment / $\qp$'' between \init ~and $\qp$ represents that upon reading $\call~ \sympayment$ at \init~state, $\qp$ is pushed on stack and the VPA transitions to state $\qp$. 
Similarly, the transition at state $\qp$ on symbol ``\ret~ \sympayment''  in \cref{fig:VPA-payment} indicates that the next state will be \init~if top of the stack  is $\qp$ while making that transition.
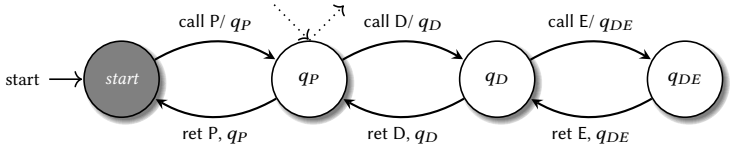
\begin{figure}
\centering
\begin{tikzpicture}[font=\scriptsize,
			baseline=1ex,shorten >=.4pt,node distance=25mm,on grid,
			semithick,auto,
			every state/.style={fill=white,draw=black,circular
					drop shadow,inner sep=.15mm,text=black,minimum size=1cm},
			accepting/.style={fill=gray,text=white}]
		\node (q0) [state, initial, accepting] {$\init$};
		\node (q1) [state, right = of q0] {$\qp$};
        \node (q2) [state, right = of q1] {$\qd$};
		\node (q3) [state, right = of q2] {$\qdl$};
		\path [-stealth, thick]
		(q0) edge [bend left] node {\call~ \sympayment/ $\qp$}(q1)
		(q1) edge [bend left] node {\call~ \symdb/ $\qd$}(q2)
		(q2) edge [bend left] node {\call ~\symlogging~/ $\qdl$}(q3)
        (q3) edge [bend left] node {\ret~\symlogging~, $\qdl$}(q2)
        (q2) edge [bend left] node {\ret~\symdb, $\qd$}(q1)
        (q1) edge [bend left] node {\ret~ \sympayment, $\qp$}(q0)
		;
		\draw[dotted, ->, inner sep=0pt, bend left] (q1.north) +(45:2em) node (tmp) {}
  (q1.north) -- (tmp);
  		\draw[dotted, ->, inner sep=0pt, bend left] (q1.north) ++(135:2em) node (tmp2) {}
  (tmp2) -- (q1.north);
	\end{tikzpicture}
 \caption{Visibly pushdown automaton, which accepts a tree  where \payment~ (\sympayment) invokes \database~(\symdb) as its child, and this \symdb~ invokes ~\eventlogging~(\symlogging~) as its children. Omitted transitions are marked in dotted. The initial and final state is $\init$.\vspace{-0.1in}}
 \label{fig:VPA-payment}
\end{figure}

\paragraph*{\textbf{Service tree---a nested-word view}}
By viewing each node in the service tree as a call and return to/from the
labeled endpoint, the service tree can be modeled as a sequence of nested calls
and returns forming \emph{nested word} \cite{nestedwords}.  For instance, each
node in \cref{fig:tree-payment} corresponds to a \call~and \ret~symbol in its
nested-word view  described in \cref{fig:nw-payment}. The \call ~and \ret~symbol
of a child call is nested between its parent's \call~ and \ret~symbols because a
synchronous API waits for the response from any API that it invokes before
resuming its execution. For instance, in \cref{fig:nw-payment}, the \call~and
\ret~ of the two \symdb~ in \cref{fig:tree-payment} are nested between the
\call~and \ret ~symbols of their parent \sympayment.

By viewing service trees as nested words, we can view policies as sets of
allowed nested words. By carefully designing our policy language, we can ensure
that our policies correspond to languages accepted by VPA, which can check if a
given service tree is permitted. For instance, the service tree in
\cref{fig:tree-payment} is accepted by the VPA in \cref{fig:VPA-payment} because
the VPA accepts the nested word view (described in \cref{fig:nw-payment}) of the
service tree. After starting the VPA run at the initial state, \init, the VPA
arrives at the final state \init, which is an accepting state.
\begin{figure}[h!]
\centering
\begin{subfigure}{0.3\textwidth}
\centering
\begin{tikzpicture}[font=\scriptsize,
			baseline=1ex,shorten >=.1pt,node distance=15mm,on grid,
			semithick,auto,
			every state/.style={circle, fill=white,draw=black,circular
					drop shadow,inner sep=.10mm,text=black,minimum size=0.5cm},
			accepting/.style={fill=gray,text=white}]
    \node (A) at (2, 0) [] {P};
    \node (P1) at (1, -1) [] {D};
    \node  (E1) at (1, -2) [] {E};
    \node  (P2) at (3, -1) [] {D};   
    \node  (E2) at (3, -2) [] {E};   
            
		\path [-stealth, thick]
		(A) edge [] node {}(P1)
		(P1) edge [] node {}(E1)
  		(A) edge [] node {}(P2)
		(P2) edge [] node {}(E2)
		;
	\end{tikzpicture}
 \caption{}
  \label{fig:tree-payment}
\end{subfigure}
\hfill
\begin{subfigure}{0.6\textwidth}
    \centering
    \begin{tikzpicture}[
        node/.style={circle, fill=black, minimum size=6pt, inner sep=0pt, font=\small, text=black}
    ]

    % Nodes with labels on top
    \node[node] (1) at (0, 0) [label=left:call P] {};
    \node[node] (2) at (0, -1) [label=left:call D] {};
    \node[node] (3) at (0, -2) [label=below:call E] {};
    \node[node] (4) at (2, -2) [label=below:ret E] {};
    \node[node] (5) at (2, -1) [label=above:ret D] {};
    \node[node] (6) at (4, -1) [label=above:call D] {};
    \node[node] (7) at (4, -2) [label=below:call E] {};
    \node[node] (8) at (6, -2) [label=below:ret E] {};
    \node[node] (9) at (6, -1) [label=right:ret D] {};
    \node[node] (10) at (6, 0) [label=right:ret P] {};
    % Edges
    \draw[->] (1) -- (2);
    \draw[->] (2) -- (3);
    \draw[->] (3) -- (4);
    \draw[->] (4) -- (5);
    \draw[->] (5) -- (6);
    \draw[->] (6) -- (7);
    \draw[->] (8) -- (9);
    \draw[->] (9) -- (10);

    % Dashed edges
    \draw[->, dotted] (2) -- (5);
    \draw[->, dotted] (6) -- (9);
    \draw[->, dotted] (1) -- (10);
    \draw[->] (7) -- (8);

    % Additional arrows
    \draw[->] (8) -- (9);
    
    \path [-stealth]
    (3) edge [->, bend right, dotted] node {}(4);
    \path [-stealth]
    (7) edge [->, bend right, dotted] node {}(8);
    \end{tikzpicture}
    \caption{}
    \label{fig:nw-payment}
\end{subfigure}
\caption{Service tree for an appointment request: (a) tree view; (b) nested-word view.}
\end{figure}
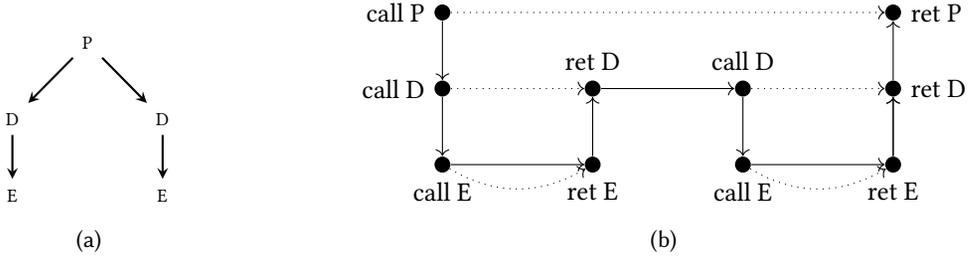
 
%Such a nested sequence of \call-\return augmented symbols capturing the nesting of API calls described in the service tree is then checked for the VPA acceptance.  
\paragraph*{\textbf{Runtime policy checking}}
To check service tree policies,
we develop a compiler to translate \safetree policies into a VPA that accepts
the valid set of nested words, and then we extract a distributed runtime monitor from this VPA. The motivation behind a  distributed design   is the lack of a global view of service trees at services. \Cref{sec:semantics} and \cref{sec:implementation} will detail the design principle behind the distributed monitoring, but the key is that instead of a centralized  monitor, we employ local sub-monitors at services to  simulate VPA transitions on their respective HTTP requests/responses. To facilitate local VPA simulation while giving the illusion of having a centralized VPA, the current VPA configuration is carried in a custom HTTP header. 

To implement this distributed monitor without invasive changes to the service
implementation, we build on top of Istio, a popular service mesh
implementation \cite{istio}. Service mesh is a network infrastructure layer that
abstracts away inter-service communication logic from the service implementation
and offloads it into a sidecar container. Istio pairs each service with a sidecar to manage the service's traffic using the Envoy proxy \citep{envoydoc} running in the container. The proxy intercepts all incoming/outgoing HTTP traffic into the service container. The Envoy proxy can be programmed with custom logic to perform reads and writes to HTTP headers and control access to the service. \Cref{sec:implementation} will detail how we use this programmability to simulate VPA transitions in Envoy.

\section{Background: Service Tree as a Nested-Word}\label{sec:refresher}
We give a refresher on the standard definition of nested words \citep{nestedwords} and introduce some custom notations that will be useful in later sections for encoding service trees as these nested words.

\subsection*{Nested Words Refresher}

We start by recalling basic concepts and notation for nested words.

\begin{definition}[Call-Return Augmented Alphabet]
$\Sigma = \Sigma_c \cup \Sigma_r$ is a \textit{call-return augmented alphabet} for some $\tilde{\Sigma}$ if $\Sigma_c = \{\cs ~\vert~s \in \tilde{\Sigma}\}$, and $\Sigma_r = \{ \rs ~\vert~s \in \tilde{\Sigma}\}$.
\end{definition}

Intuitively, $\cs$ and $\rs$ correspond to an HTTP request or call  to an endpoint $s$ and an HTTP response from an endpoint $s$ respectively.

\begin{definition}[Indexed Symbols]
An \textit{indexed symbol} $a = (e, i)$ is a pair of some symbol $e \in \Sigma$
and an index $i = \idx(a) \in \mathbb{N}^+$, where $\idx$ is the index
projection. Here, $\Sigma$ is a call-return augmented alphabet for some
$\tilde{\Sigma}$. The symbol $a$ is said to be a \textit{call symbol} if $e \in
\Sigma_c$, and a \textit{return symbol} if $e \in \Sigma_r$. We write $C(a)$
when $a$ is a call symbol and $R(a)$ when $a$ is a return symbol. 
\end{definition}

\begin{definition}[Nested Word]\label[definition]{def:nw}
A nested word $n = (w, \nu)$ over $\Sigma = \Sigma_c \cup \Sigma_r$ (a call-return augmented alphabet for some $\tilde{\Sigma} \ni s$) is a pair of: a word $w = a_1 \ldots a_l = (a_i)_{1 \le i \le l}$ such that $\idx (a_1) = k$ for some $k \in \mathbb{N}^+$ and $\idx (a_i) +1 = \idx(a_{i+1})$ for any $1 \le i < l$; and a \textit{matching relation} $\nu \subseteq \{-\infty, k, \ldots, k+l-1\} \times \{k, \ldots, k+l-1, +\infty\}$ that associates call symbols with corresponding return symbols and satisfies:

% {1, 2, ...l} {1:k+0, 2:k+1, ..., l: k+l-1}
\begin{enumerate}
\item \textit{Each symbol occurs in only one pair.} For any $a_c = (\cs, i) \in w$, there exists a unique $j \in \{ k, \ldots, k+l-1, +\infty\}$ such that $\mrel{i}{j}$ and if $j \neq +\infty$ then there exists a return symbol $a_r \in w$ with index $\iota(a_r) = j$. For any $a_r = (\rs, j) \in w$, there is a  unique $i \in \{-\infty, k, \ldots, k+l-1\}$ such that $\mrel{i}{j}$ and if $i \neq -\infty$ then there exists a call symbol $a_c \in w$ with index $\iota(a_c) = i$. 
\item \textit{Edges go forward.} If $\mrel{i}{j}$ then $i < j$. 

\item \textit{Edges do not cross.} If $\mrel{i}{j}$,  $\mrel{i'}{j'}$ and $i < i'$ then either the two edges are well-nested (\textit{i.e., } $j' < j$) or disjoint (\textit{i.e.,} $j < i'$).
\end{enumerate}
Let $c \in \Sigma_c$ and $r \in \Sigma_r$. A nested word is \textit{well-matched} if for any $a_c = (c, i) \in w$, there exists an $a_r = (r, j) \in w$ such that $\mrel{i}{j}$; and for any $a_r = (r, j) \in w$, there exists an $a_c = (c, i) \in w$ such that $\mrel{i}{j}$. Here, $a_i$ is a call whose \textit{matched return} is $a_j$.
\end{definition}
\begin{example}
\Cref{fig:formal-nw-example} formally describes the well-matched nested word
representation  for the service tree in \cref{fig:tree-payment}. Here,
$\tilde{\Sigma} = \{\sympayment, \symdb, \ldots\}$ is the set of all endpoint
names, and the call-return augmented alphabet for $\tilde{\Sigma}$ is  $\Sigma =
\{\langle\sympayment, ~\langle\symdb, ~\ldots\} \cup \{\sympayment\rangle,
\symdb\rangle, ~\ldots\}$. In the nested word representation, each API call
in the service tree  will be unfolded into a matched call and return pair. In
this case, the nested word is $n = (a_1 a_2 a_3 a_4 a_6 a_7 a_8 a_9 a_{10}, \nu)$, where $a_1, \ldots, a_{10}$ are indexed symbols; we will sometimes
call the first symbol $a_1$ the \emph{root} of the nested word, since it is root
of the service tree.

The \relmatch $\nu$ is described using the dotted lines. For instance, the first
call to the \payment service or $a_2$ is matched with the return symbol $a_5$.
Note that the first and the last symbol in the nested word, \textit{i.e.,} $a_1$
and $a_{10}$ respectively form a matched call-return pair; due to the inherent
request/response protocol, nested words corresponding to service trees in a
microservice application will always be of this form.
\end{example}
A nested word $n$ can also be sliced into smaller nested words corresponding to
various subtrees of the service tree.

\begin{definition}[A sub-nested word]\label[definition]{def:subnested}
Given $n = (a_1 \ldots a_l,  \nu)$ and $ \idx(a_1) \le i < i' \le \idx(a_l)$,
let $a_j, a_{j'} \in n$ be such that $\idx(a_j) = i$, and $\idx(a_{j'}) = i'$.
We define a sub-nested word $n[i, i'] = (a_j \ldots a_{j'}, \nu[i, i'])$. Here,
$\nu[i, i']$ is the restricted matching sub-relation:
\begin{align}
    \nu[i, i'] &= \{\mrel{p}{q}~\vert~i \le p, q \le i'\}  \nonumber \cup \{(p, +\infty) ~|~ \mrel{p}{q}, ~i \le p \le i', ~ ~q > i'\} \nonumber\\
      &\;\; \cup \{(-\infty, q) ~|~ \mrel{p}{q}, ~i \le q \le i', ~p < i\} \nonumber
\end{align}
\end{definition}

\begin{example}
Consider the service tree in \cref{fig:tree-payment}. Its first subtree rooted
at $\symdb$ consists of the first $\symdb$ and first $\symlogging~$. In
terms of the nested word (in \cref{fig:formal-nw-example}) for this service
tree, the subtree under consideration is given by the sub-nested word $n[2, 5] =
(a_2 a_3 a_4 a_5, \nu')$, where $\nu' = \{(2, 5), ~(3, 4)\}$. 
\end{example}

\begin{figure}[h!]
    \centering
    \begin{tikzpicture}[
        node/.style={circle, fill=black, minimum size=3pt, inner sep=0pt, font=\small, text=black}
    ]
    % Nodes with labels on top
    \node at (3, 0.5) {\textbf{\small $n = (a_1 a_2 a_3 a_4 a_6 a_7 a_8 a_9 a_{10}, \nu)$}};
    \node[node] (1) at (0, 0) [label=left:{\small $a_1 = \labelcallnw{\sympayment}{1}$}, font=\small] {};
    \node[node] (2) at (0, -1) [label=left:{\small $a_2 = \labelcallnw{\symdb}{2}$}, font=\small] {};
    \node[node] (3) at (0, -2) [label=below:{\small $a_3 = \labelcallnw{\symlogging~}{3}$}]{};
    \node[node] (4) at (2, -2) [label=below:{\small $a_4 = \labelretnw{\symlogging~}{4}$}] {};
    \node[node] (5) at (2, -1) [label=above:{\small $a_5 = \labelretnw{\symdb}{5}$}] {};
    \node[node] (6) at (4, -1) [label=above:{\small $a_6 = \labelcallnw{\symdb}{6}$}] {};
    \node[node] (7) at (4, -2) [label=below:{\small $a_7 = \labelcallnw{\symlogging~}{7}$}] {};
    \node[node] (8) at (6, -2) [label=below:{\small $a_8 = \labelretnw{\symlogging~}{8}$}] {};
    \node[node] (9) at (6, -1) [label=right:{\small $a_9 = \labelretnw{\symdb}{9}$}] {};
    \node[node] (10) at (6, 0) [label=right:{\small $a_{10} = \labelretnw{\sympayment}{10}$}] {};
    % Edges
    \draw[->] (1) -- (2);
    \draw[->] (2) -- (3);
    \draw[->] (3) -- (4);
    \draw[->] (4) -- (5);
    \draw[->] (5) -- (6);
    \draw[->] (6) -- (7);
    \draw[->] (8) -- (9);
    \draw[->] (9) -- (10);

    % Dashed edges
    \draw[->, dotted] (2) -- (5);
    \draw[->, dotted] (6) -- (9);
    \draw[->, dotted] (1) -- (10);
    \draw[->] (7) -- (8);

    % Additional arrows
    \draw[->] (8) -- (9);
    
    \path [-stealth]
    (3) edge [->, bend left, dotted] node {}(4);
    \path [-stealth]
    (7) edge [->, bend left, dotted] node {}(8);
    \end{tikzpicture}
\caption{Well-matched nested word for the service tree in \cref{fig:tree-payment}. The dotted line from a node $a_i$ to $a_j$ represents the match relation $\mrel{\idx(a_i)}{\idx(a_j)}$. The projection on the second element of any indexed symbol $a_i$ is $\idx(a_i)$.}
    \label{fig:formal-nw-example}
\end{figure}
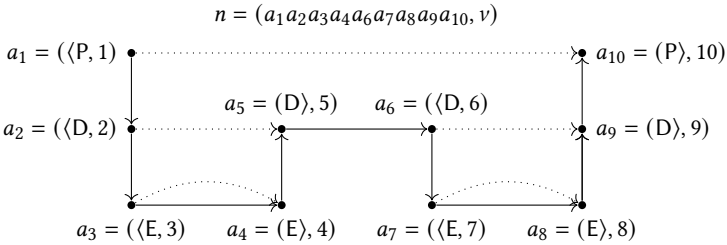

\subsection*{Service Tree Concepts}\label{sec:treeconcepts}
To define our policies, we first define nested-word
equivalents for different parts of a service tree. \kcomments{how can i emphasize that these are notations for our notion of service trees. }

First, we define a set of words that correspond to all paths from the root of the service tree representation of a nested word. Well-matched nested words are said to be \textit{rooted} if the first and the last symbol are matched. 
\begin{definition}[Set of Paths]\label[definition]{def:paths}
    For a \rmnw $n = (a_1 \ldots a_l,\nu)$ over some $\Sigma$, the set of all \textit{call sequences} $Seq(n)$ from the root is given by:
%For a nested word $n = (a_1 \ldots a_l, \nu)$ rooted at $a_1$, the set of paths $Path(n)$ is given by:
\[
  Seq(n) = \Set {\pi (n, a_m)  | C(a_m), ~a_m \in n},
\]
where $\pi  (n, a_m)$ is a word $w = (a_i)_{i \in I}$ and
\[
  I = \{i ~\vert~ a_i \in n,~\idx(a_i) \le \idx(a_m),~C(a_{i}),~ \mrel{\idx(a_i)}{j},~  j > \idx(a_m)\}.
\]
The word $w$ is a \textit{path} from the root to some subsequent call $a_m \in n$.
\end{definition}

\begin{example}
In the running example in \cref{fig:formal-nw-example}, the set of all paths
in the nested word $n$ is given by 
$  Seq(n) = \{ a_1 a_2, ~ a_1 a_2 a_3, ~a_1 a_6 a_7, ~a_1 a_6\}$. A leaf node is some call symbol that is immediately followed by a return symbol. The set of all path to leaves in $n$ is given by $SeqLeaf(n) =  \{a_1 a_2 a_3,~ a_1 a_6 a_7\}$.
\end{example}

To reason about the children of a node in the tree, we define a set of symbols
corresponding to children of a nested word's root:
\begin{definition}[Child of a root]\label[definition]{def:child}
For a \rmnw  $n = (a_1 \ldots a_l, \nu)$, a call $a_m \in n$ is a \textit{child} of root $a_1$ if $a_1 a_m \in Seq(n)$. The set of children $Child(n)$ is given by:
\[
  Child(n) = \Set{ a_m | \begin{array}{l}
    a_1 a_m \in Seq(n)
  \end{array}}
\]

\end{definition}

In the nested word $n$ in \cref{fig:formal-nw-example}, symbols $a_2$ and $a_6$
are the children of the node $a_1$. Similarly, $a_3$ is the child of $a_2$.  The
sub-nested words $n[2, 5]$ and $n[6, 9]$ form a set of sub-trees rooted at
$a_1$'s children. We formally define the set of \textit{child subtrees of the
root} of a nested word.

\begin{definition}[Child subtree of a root]\label[definition]{def:subtree}
Given a \rmnw $n = (a_1 \ldots a_l, \nu)$ with $a_m$ as one of the children of $a_1$ and some $x \in \mathbb{N}^+$ such that $\mrel{\idx(a_m)}{x}$, a \textit{subtree} rooted at $a_m$ is a \rmnw 
$n' = ((a_j)_{j \in I}, \nu_m)$, where 
$I = \{\idx(a_m), \ldots, x\}$ and $\nu_m = \nu[\idx(a_m), x]$.
The set of subtrees is given by:
\[
  Subtree(n) = \Set{ ((a_j)_{j \in I}, \nu_m) \ | \begin{array}{l}
    a_m \in Child(n), \\
    \exists~ x \in \mathbb{N}^+~\text{such that}~  \mrel{\idx(a_m)}{x}
  \end{array}}
\]
\end{definition}

\section{SafeTree: A Service Tree Policy Language} \label{sec:policy}
Now that we have introduced notation for nested words, we describe our policy
language \safetree, and define an interpretation as sets of nested words.

\subsection{Syntax}
The policy syntax is presented in \cref{fig:syntax}. \safetree policies make extensive use of regular expressions ($reg$) defined over the set of endpoints $\basealpha$. 

Top-level \safetree policies are of the form \hpolicy{\start~S: inner}, where $S \subseteq \basealpha$, and $inner$ can be two kinds of sub-policies: $p$ over the hierarchical structure or $seq$ over the linear structure of the tree. 
Intuitively, top-level policies specifies that ``the $inner$ sub-policy is satisfied on any subtree rooted at an endpoint in $S$ that has no ancestor in $S$''. For instance, if $S = \{\symappt\}$,  $inner$ should be satisfied on all subtrees whenever \symappt is first encountered. We write $\start \star : inner$ as convenient notation for $\start \basealpha : p$; this policy specifies that $inner$ should be satisfied at the root of the entire tree. 

\begin{figure}[h!]
\begin{subfigure}{0.4\textwidth}
    \begin{flalign*}
& \textbf{Regular Expressions} &\\
& a \in \basealpha &\\
& S \subseteq \basealpha &\\
& reg ::= a ~\vert~ \epsilon ~\vert~ \emptyset ~\vert~ \regexone + \regextwo ~\vert~\regexone \regextwo ~\vert~ reg^* & \\
& \\
&\textbf{Service Tree policy} &\\
& policy ::= \start S: p ~\vert~ \start ~S: seq & \\
& &\\
\end{flalign*}
\end{subfigure}
\hfill
\begin{subfigure}{0.4\textwidth}
\begin{flalign*}
&\textbf{Hierarchical Policy} &\\
&p_{l}  ::=  \match~\regexone \allpath \regextwo &\\
&p_{a}  ::= \match~ \regex \allchildren p & \\
&p_{e} ::=  \match ~\regex \existschild p_1 ~\textsf{\textbf{then}}\ldots \textsf{\textbf{then}}~p_k &\\
&p ::=  p_{l} ~\vert~ p_a ~\vert ~ p_e &\\
& &\\
&\textbf{Linear Sequence Policy} &\\
&seq ::= \callseq ~ \regex &
\end{flalign*}
\end{subfigure}
\caption{\safetree syntax}
\label{fig:syntax}
\end{figure}
\paragraph{Inner: Hierarchical Policies}
After specifying the starting symbol of a policy, a \safetree policy needs to
specify a sub-policy to express constraints on hierarchical or the linear
structure of a service tree. All hierarchical policies shown in
\cref{fig:syntax} have ``\match ~\regex'' expression to the left of the
annotated $\implies$ symbol. This expression specifies that there exists a path
from the root of the service tree to some descendant $a_i$ such that the path
matches $\regex$. (Note that $\regex$ in the \match expression should not match
the empty string.) We will step through the syntax of hierarchical policies to
understand the purpose of different right side expressions in these policies in
terms of the service tree. 

The simplest policy is $\match~\regexone \allpath \regextwo$, where $\regexone, \regextwo$ are regular expressions. It expresses the existence of a path from the root to some descendant $a_i$ that matches $reg_1$. Also, all the paths from any of $a_i$'s children to the leaves in the tree match $reg_2$. The path between the tree's root and $a_i$ should be the \textit{shortest} match of $reg_1$, \textit{i.e.,} there should be no prefix of the path that matches $\regexone$. 

\paragraph*{Example} Consider the service tree previously defined in \cref{fig:tree-payment}. Suppose $\kleenestar$ is a syntactic sugar for the regular expression notation for $\basealpha^*$. This tree satisfies  the policy $\match ~\sympayment \symdb^* \allpath \kleenestar$, which requires existence of a path from the root of the tree, \textit{i.e.,} \sympayment that matches exactly one $\sympayment$ and any number of  $\symdb$. If we consider $\sympayment$ as both the root and the descendant $a_i$, we get a path of length one that satisfies the requirement. The two paths $\symdb \symlogging~$ from the children of \sympayment satisfy the regular expression $\kleenestar$ on the right. Therefore, the service tree satisfies the policy. Although there are more paths, like $\sympayment\symdb$ in the tree that could have matched the regular expression on the left, we care about the path which does not have a prefix in the language of the regular expression $\sympayment \symdb^*$.

To specify that all subtrees of some node in a tree should satisfy a policy $p$, we can use the policy $\match~\regex \allchildren p$. This policy specifies that there exists a node $a_i$ such that the path from the root of the tree to $a_i$ matches $\regex$ and all subtrees rooted at children of $a_i$ satisfy policy $p$. This policy specifies a universal condition on a node's subtrees. 

To specify existential conditions on the subtrees, the policy $\match~ \regex \existschild p_1 ~\textsf{\textbf{then}}\ldots \textsf{\textbf{then}}~p_k$ can be used. It specifies that there exists a node $a_i$ such that the path from the root of the tree to $a_i$ matches $\regex$. Also, $a_i$ has a subtree that satisfies $p_1$ followed by another subtree somewhere after it that satisfies $p_2$, and so on till $p_k$. Suppose $c_{i}$ and $c_{i+1}$ are the root node of the subtrees matching $p_i$ and $p_{i+1}$ respectively, where $1 \le i < k$. Let $C$ be the parent of $c_1$ and $c_2$. The node $C$ can have children older than $c_{i}$ and younger than $c_{i+1}$.

\paragraph{Inner: Linear Sequence Policy}
% A policy on the linear structure disregards the tree structure of API calls and specifies constraints on the linear sequence of API calls in the tree (or formally, the sequence of API calls in the depth-first traversal of the tree). 
\safetree offers a $\callseq$ construct to specify the desired sequence of API calls as a regular expression, while disregarding the hierarchical tree structure. For instance, in a tree starting at \sympayment, say, we want to express that  $\symdb$ happens after $\sympayment$, without any specific details about the subtrees or paths in the tree where \symdb should occur. This will be specified as $\callseq~ \sympayment\kleenestar\symdb\kleenestar$, where  $\kleenestar$ in the regular expression denotes any sequence of symbols. 
This policy specifies that the depth-first traversal of the tree should match the given regular expression.
\begin{figure}[h]
\begin{align*}
& \hpolicy{policy = \start S: seq} \\
 & \nw{policy} =
  \Set{n = (a_1 \dots a_l, \nu)\ | \begin{array}{l}
    \forall ~\pi(n, a_i) = a_1 \ldots a_i \in Seq(n), \\
    \text{if}~
    \pi(n, a_i) \in FirstMatch(n, (\tilde{\Sigma})^*S)~ \text{then}\\
    \exists ~x~ \text{such that} \\
    \mrel{\idx(a_i)}{x} ~\text{and} ~n[\idx(a_i), x] \in \Seq{seq}
  \end{array}}\\
  & \hpolicy{policy = \start S: p} \\
   & \nw{policy} =
  \Set{n = (a_1 \dots a_l, \nu)\ | \begin{array}{l}
    \forall ~\pi(n, a_i) = a_1 \ldots a_i \in Seq(n), \\
    \text{if}~
    \pi(n, a_i) \in FirstMatch(n, (\tilde{\Sigma})^*S)~ \text{then}\\
    \exists ~x~ \text{such that} ~\mrel{\idx(a_i)}{x}~\text{and}~n[\idx(a_i), x] \in \tree{p}
  \end{array}}\\
 & \hpolicy{p_l = \match~~reg_1 \allpath reg_{2}}\\
 & \tree{p_l} = \Set{n = (a_1 \dots a_l, \nu) \ | \begin{array}{l}
          \exists ~\pi(n, a_i \in n) = a_1 \ldots a_i \in Seq(n)~
    \text{such that} ~\\
    \pi(n, a_i) \in FirstMatch(n, reg_1),  \\
    \exists ~x ~\text{such that}~\mrel{\idx(a_i)}{x}~\text{and}\\
    \quad \forall n_s \in Subtree(n[\idx(a_i), x]),\\
    \quad \quad \forall \pi' \in SeqLeaf(n_s), ~Calls ( \pi' ) \in \mathcal{L}(reg_2)
  \end{array}}\\
  & \hpolicy{p_a = \match~ \regex \allchildren p}\\
  & \tree{p_a} = \Set{n = (a_1 \dots a_l, \nu)\ | \begin{array}{l}
    \exists ~\pi(n, a_i \in n) = a_1 \ldots a_i \in Seq(n)~
    \text{such that} ~\\
    \pi(n, a_i) \in FirstMatch(n, \regex),  \\
 \exists ~x ~\text{such that}~\mrel{\idx(a_i)}{x}~\text{and}~Subtree(n[\idx(a_i), x]) \subseteq \tree{p}
  \end{array}} \\
   & \hpolicy{p_e = \match~ \regex \existschild p_1 ~\textsf{\textbf{then}}\ldots \textsf{\textbf{then}}~p_k}\\
  &\tree{p_e} = 
  \Set{n = (a_1 \dots a_l, \nu)\ | \begin{array}{l}
    \exists ~\pi(n, a_i \in n) = a_1 \ldots a_i \in Seq(n)~
    \text{such that} \\
    \pi(n, a_i) \in FirstMatch(n, \regex),  \\
    \exists ~x ~\text{such that}~\mrel{\idx(a_i)}{x}~\text{and}~\\
    \quad \exists ~\{t_1, \dots, t_k\} \subseteq Subtree(n[\idx(a_i), x])~
    \text{such that} \\
     \quad \quad \text{for any}~ 1 \le j < k, ~\idx(a^j_1) < \idx(a^{j+1}_1), \text{and}\\
    \quad \quad \text{for any}~ 1 \le j \le k,~ t_j \in \tree{p_j},\\
    \text{where}~\forall 1 \le j \le k, t_y = (a^{j}_1 \ldots, \nu^j)
  \end{array}} \\
  & \hpolicy{seq = \callseq ~\regex} \\
    &\Seq{seq} = 
  \{n = (a_1 \dots a_l, \nu) \mid
 Calls((a_{x})_{x \in I} ) \in \mathcal{L}(reg), ~
 \text{where}~I = \{\idx(a_j) \mid C(a_j),~a_j \in n\}\}
\end{align*}
\caption{\safetree semantics (where $Subtree$ is defined in \cref{def:subtree})}
\label{fig:intp}
\end{figure}
\vspace{-0.5em}
\subsection{Nested word interpretation for service trees}
Formally, we interpret a \safetree policy as a set of \rmnw. 
Since matching paths with regular expressions is common across \safetree policies, we define a nested word variant for a set of shortest paths matching some regular expression. 

\begin{definition}[First or Shortest Match Path]
Given a nested word $n = (a_1 \ldots a_l, \nu)$ over $\Sigma$, which is a call-return augmented alphabet for some $\tilde{\Sigma}$, a path $\pi(n, a_m) \in Seq(n)$ is the \textit{first match} for some regular expression $reg$ over $\tilde{\Sigma}$  (or $ \pi(n, a_m) \in FirstMatch(n, reg)$) if
$w = Calls (\pi(n, a_m))$ is matched by $reg$, and there is no smaller prefix of $w$ that matches $reg$. 
\[FirstMatch(n, reg) = \Set{\pi (n, a_m) \in  Seq(n) \ | \begin{array}{l}
     \text{for all} ~ 1 \le j < m, ~\text{the following holds} \\
     s_1 \ldots s_j \notin \mathcal{L}(reg), ~
     w \in \mathcal{L}(reg), \\
     \text{where}~ w = Calls( \pi(n, a_m)) = s_1 \ldots s_m
  \end{array}}
\]
\end{definition}
Here, $Calls$ of a path rewrites every augmented call symbol with its counterpart in $\basealpha$. 
\begin{definition}[Calls Projection of a Path]\label[definition]{def:apx-call}
Given a nested word $n = (a_1 \ldots a_l, \nu)$ over base alphabet $\basealpha$ and a path $\pi(n, a_m) = a_1 \ldots a_m$, we define 
\[Calls (a_1 \ldots a_m) = s_1 \ldots s_m, ~\text{where}~s_i = \textsf{proj}(a_i)~\text{for any}~1 \le i \le m.\]
Here, $\textsf{proj}(a)$ projects the base symbol of an indexed symbol, \textit{i.e.,} for any $a = (\langle s, i)$ or $a = (s \rangle, i)$ (with some index $i$), we define $\textsf{proj}(a)=s$.
\end{definition} 
\vspace{-0.1in}
The nested word interpretation $\nw{p} \subseteq {NestedWords(\Sigma)}$ for a service tree policiy $p$ is  defined in \cref{fig:intp}. Here, $NestedWords(\Sigma)$ is the set of all nested words over $\Sigma$. 
For a nested word $n = (a_1 \ldots a_l, \nu)$ to be accepted by a service tree policy $\start~S: ~inner$, the $inner$ policy  should be satisfied on every sub-nested word of $n$ that is rooted at some symbol  $a_i \in S$ and the path from the root of the nested word to the symbol $a_i$ should contains no symbol in $S$ besides $a_i$. In the $\nw{p}$ definition in \cref{fig:intp}, this condition is formally expressed by requiring the path between the root and $a_i$ to be a first match of regular expression $(\basealpha)^*S$. 

%
%Given a hierarchical policy $p$, the set of nested words accepted by $p$ is  $\tree{p} \subseteq NestedWords(\Sigma)$, which is defined in \cref{fig:intp}. 
The nested word interpretation $\tree{p} \subseteq NestedWords(\Sigma)$ in
\cref{fig:intp} for a hierarchical policy first defines an existential
constraint on a path from the root to some node, followed by specific constraints
on either path until the leaves, or on sub-nested words (or subtrees).   
As defined in \cref{fig:intp}, a linear policy accepts a nested word $n \subseteq \Seq{p}$ if the sequence of call symbols in $n$ match the given regular expression in $p$.

\section{Enforcement} \label{sec:semantics}

The SafeTree policies are enforced by a visibly pushdown automaton (VPA)-based monitor. \kcomments{While VPAs \cite{vpa} have been well-studied, to the best of our knowledge, there is no tool that compiles a specification in a language to an output VPA automaton that is independent of the application code to be monitored at runtime.} For this, we define a compiler from our policy language to a VPA. The VPA is then used to check if a service tree is valid. Before we detail our VPA-based enforcement mechanism, we first give a refresher on the standard VPA model:
\begin{definition}[Visibly pushdown automaton]
A (deterministic) \emph{visibly pushdown automaton} (VPA) is defined as $\mathcal{M} = (Q, q_{init}, F, \Sigma, \Gamma, \bot, \delta_c, \delta_r)$, where:
\begin{itemize}
\item $Q$ is the set of all states, $q_{init} \in Q$ is the initial state, and $F \subseteq Q$ is the set of final states,
\item $\Sigma$ is the alphabet, where $\Sigma = \Sigma_c  \cup \Sigma_r$ consists of call symbols $\Sigma_c$ and return symbols $\Sigma_r$,
\item $\Gamma$ is the set of stack symbols with a special bottom of stack symbol $\bot \in \Gamma$,
\item $\delta^p_c: \calltype$ is the call transition function,
\item $\delta^p_r: \rettype$ is the return transition function.
\end{itemize}
\end{definition}
Note that we do not need the conventional  \textit{internal symbols} of a VPA
to monitor our policies; extending our policies to use internal symbols is an
interesting avenue for future work.

\begin{example}
    \Cref{fig:ex-VPA} shows a two state VPA, which  is defined  over $\Sigma = \{\vpacall{\symappt}, \vparet{\symappt}\}$. In this case, $Q = \{q_0, q_1\}$, $q_{init} = q_0$,  $F= \{q_0\}$, and $\Gamma = \{\bot, ~q_0\}$. 
\end{example}
A VPA configuration $(q, \theta)$ is a pair of its current  state $q$ and the current stack $\theta$. The stack $\theta$ is a sequence of stack symbols with only one occurrence of $\bot$ at the starting. For instance, a possible configuration for the VPA in \cref{fig:ex-VPA} can be $(q_1, \bot q_0)$. This configuration denotes that the VPA is at state $q_1$ and its stack has only $q_0$.
\begin{figure}[h!]
\begin{subfigure}[c]{0.2 \textwidth}
\centering
    \begin{tikzpicture}[font=\scriptsize,
			baseline=1ex,shorten >=.4pt,node distance=25mm,on grid,
			semithick,auto,
			every state/.style={fill=white,draw=black,circular
					drop shadow,inner sep=.15mm,text=black,minimum size=1cm},
			accepting/.style={fill=gray,text=white}]
		\node (q0) [state, initial, accepting] {$q_0$};
		\node (q1) [state, right = of q0] {$q_1$};
		\path [-stealth, thick]
		(q0) edge [bend left] node {$\vpacall{\symappt}$/ $q_0$}(q1)
        (q1) edge [bend left] node {$\vparet{\symappt}$, $q_0$}(q0)
		;
	\end{tikzpicture}
\end{subfigure}
\hfill
\begin{subfigure}[c]{0.5 \textwidth}
\centering
\begin{flalign*}
&\text{Nested word:}~    &&n = (a_1 a_2, \nu), ~\text{where}&\\
&    &&a_1 = (\vpacall{\symappt}, 1),~ a_2 = (\vparet{\symappt}, 2),~\text{and}&\\
&     && \nu = \{(1, 2)\} &
\end{flalign*}
\begin{flalign*}
&\text{Run:}~  (q_0, \bot) \xrightarrow{\vpacall{\symappt}}(q_1, \bot q_0) \xrightarrow{\vparet{\symappt}} (q_0, \bot) &
\end{flalign*}
\end{subfigure}
 \caption{(a) A two state VPA over $\Sigma = \{\vpacall{\symappt}, \vparet{\symappt}\}$ with $Q = \{q_0,~ q_1\}$, $q_{init} = q_0$,  $F= \{q_0\}$, and $\Gamma = \{\bot, q_0\}$; (b) Run of the VPA on $n = (a_1 a_2, \nu)$.}
 \label{fig:ex-VPA}
\end{figure}

\vspace{-0.2in}
\subsection{Semantics of a VPA}
The run $\rho(w) = (q_1, \theta_1) \ldots (q_k, \theta_k)$ of a VPA $\mathcal{M}$ on some nested word $w$ is a  sequence of configurations. Nested word $w$ is accepted by $\mathcal{M}$ if $q_n$ is a final state of $\mathcal{M}$. The nested word $w \in \mathcal{L}(\mathcal{M})$ is in the language of $\mathcal{M}$ if $w$ is accepted by $\mathcal{M}$. 
Intuitively, when a VPA $\mathcal{M}$ transitions from one configuration to another on a call symbol, say $a \in \Sigma_c$, it pushes a value on the stack and moves to the next state. While transitioning on a return symbol $a \in \Sigma_r$, the top of the stack is popped and the configuration state is updated.

For example in \cref{fig:ex-VPA}, the  transition $\vpacall{\symappt}/ q_0$ is a call transition that pushes $q_0$ on the stack upon reading the symbol $\vpacall{\symappt}$ . The transition labeled with $\vparet{\symappt}, q_0$ is a return transition on $\vparet{\symappt}$ and it pops the top of the stack $q_0$.

The valid transitions allowed by a VPA can be defined as follows:
\begin{definition}
Let $\mathcal{M} = (Q, q_{init}, F, \Sigma, \Gamma, \bot, \delta_c, \delta_r)$ be a deterministic VPA with states $Q$, $\Sigma = \Sigma_c \cup \Sigma_r$. Let $\mu$ be the set of all stacks and $a$ be some symbol in $\Sigma$ then $\xrightarrow{a}: Q \times \mu \to Q \times \mu$ is defined as follows:
\begin{enumerate}
\item if $a \in \Sigma_c$ then $(q', \theta') \xrightarrow{a} (q, \theta)$ if there exists $(q', a, q, s) \in \delta_c$, where $s \in \Gamma$ and $\theta = \theta' s$,
\item if $a \in \Sigma_r$ then $(q', \theta') \xrightarrow{a} (q, \theta)$ if there exists $(q', s, a, q) \in \delta_r$, where $s \in \Gamma$ and $\theta s = \theta'$.
\end{enumerate}
\end{definition}

\begin{example}
    Let us look at the run of the VPA in \cref{fig:ex-VPA} on the nested word $n = ((\vpacall{\symappt}, 1) (\vparet{\symappt}, 2), \nu)$, where $\nu = \{(1, 2)\}$. As shown in the figure, on the first symbol, the VPA goes to state $q_1$ and pushes $q_0$ on stack. On the next symbol, the top of the stack is $q_0$, so the VPA goes from the state $q_1$ to $q_0$ and the stack value $q_0$ is popped. 
\end{example}
\vspace{-0.1in}
\subsection{Compilation Sketch}\label{sec:compilation}
We define a VPA interpretation $\semvpa{.}: Policy \to \setvpa$ for \safetree
policies. Here, $Policy$ is the set of all policies and $\setvpa$ is the set of
all VPAs.  Here, we provide a sketch of  the compilation; the detailed rules and
further discussion can be found in \iffull 
Appendix \S\ref{app:compiler}. 
\else
the full paper.
\fi
Below, we write $\mathcal{M}$ for the target policy's VPA. We first consider sequential
policies.

\paragraph*{\textbf{VPA of} $\callseq ~ \regex$} $\mathcal{M}$ simulates $\regex$'s DFA $\mathcal{A}$  on all call symbols and ignores the return symbols. The nested word is accepted if $\mathcal{A}$ accepts the sequence of calls.

We now turn to hierarchical policies, which are of the form ``$\match~\regex \placeholderpolicy inner$''. 
For such policies, $\mathcal{M}$ first simulates $\regex$'s DFA $\mathcal{A}$
(on call symbols) to find a path that matches $\regex$. During this phase,
$\mathcal{M}$ maintains a stack of $\mathcal{A}$'s run on the path. 
Suppose the above path ends with some symbol $\cs{}$, $\mathcal{M}$ checks if
$inner$ is satisfied on the subtree rooted at  $\cs{}$. If $inner$ is not
satisfied, the policy can still be satisfied if there exists another path from
the root that matches $\regex$ leading to a subtree satisfying $inner$. To
implement this behavior, the following \textit{retry steps} are taken: on return
symbols, $\mathcal{M}$ backtracks $\mathcal{A}$'s run using the stack history of
the run; $\mathcal{M}$ searches for another path by simulating $\mathcal{A}$ to
go forward on call symbols; and then running the checks for $inner$.

Now, we detail the other steps of policy checking.

\paragraph*{\textbf{VPA of }$\match~\regex \allpath \regex'$}   While reading the paths in the subtree rooted at $\cs{}$, $\mathcal{M}$ simulates the DFA $\mathcal{A}'$ of  $\regex'$ on the call symbols; and on return symbols, backtracks $\mathcal{A}$'s run using the stack history. If $\mathcal{A}'$ accepts each path, $\mathcal{M}$ accepts the nested word; otherwise it searches for another path matching $\regex$.

\paragraph*{\textbf{VPA of} $\match~ \regex \allchildren p$} $\mathcal{M}$ simulates $p$'s VPA $\mathcal{M}_p$ on all the subtrees rooted at $\cs{}$'s children. If $\mathcal{M}_p$ accepts each of  $\cs{}$'s child subtree, $\mathcal{M}$ accepts the nested word; otherwise it searches for another path matching $\regex$.

\paragraph*{\textbf{VPA of} $\match ~\regex \existschild p_1 ~\textsf{\textbf{then}}\ldots \textsf{\textbf{then}}~p_k$} Let any policy $p_i$'s VPA be $\mathcal{M}_{p_i}$. $\mathcal{M}$ simulates $\mathcal{M}_{p_1}$  on $\cs{}$'s first child subtree; if the subtree is not accepted by $\mathcal{M}_{p_1}$, this simulation is repeated on the next child subtree of $\cs{}$. When $\mathcal{M}_{p_1}$ accepts such a subtree, $\mathcal{M}$  simulates $\mathcal{M}_{p_2}$ on the next child subtree, and so on. $\mathcal{M}$ accepts the word if after repeating these steps, $\mathcal{M}$ finds a subtree accepted by $\mathcal{M}_{p_k}$; otherwise $\mathcal{M}$ retries starting from the search for a path matching $\regex$.

\paragraph*{\textbf{VPA of} $\start~S: inner$} First, $\mathcal{M}$ simulates
the DFA $\mathcal{A}$ (on the call symbols) that accepts paths from the root to
some symbol in $S$ that does not have any ancestor in $S$. Similar to the
hierarchical policies, $\mathcal{M}$ then simulates $inner$'s VPA
$\mathcal{M}_{in}$ on the subtree rooted at $\cs{}$.  $\mathcal{M}$ continues to
apply retry  steps (similar to the hierarchical policies) to the search for
other paths matching $\mathcal{A}$ and checking if $\mathcal{M}_{in}$ accepts
the relevant subtrees for each of those paths.
\kcomments{is it clear that it will reject otherwise}

As expected, our compilation is sound with respect to \safetree's nested word
semantics.
\begin{theorem}[Soundness]\label{thm:soundness}
Let $p$ be a policy and  $\mathcal{L}(\vp{p})$ be the set of \rmnws accepted by its visibly pushdown automaton $\vp{p}$. Then,
\begin{enumerate}
    \item $\nw{p} = \mathcal{L}(\vp{p})$ if  $p$ is a service tree policy,
    \item $\tree{p} = \mathcal{L}(\vp{p})$ if $p$ is a hierarchical policy, and
    \item $\Seq{p} = \mathcal{L}(\vp{p})$ if $p$ is a linear sequence policy.   
\end{enumerate}
\end{theorem}
All three equivalences in the above theorem are proved by induction on the structure of the policies. The proof is given in 
\iffull
Appendix \S\ref{app:soundproof}.
\else 
the full paper.
\fi

Finally, we consider complexity. We can bound the size of the compiled VPA in
terms of the size of the policy as follows:
\begin{theorem}
Suppose the DFA of every regular expression in a policy $p$ has at most $R$ states; each sub-policy has at most $k$ immediate sub-expressions, \textit{i.e.}, fan-out at most $k$; and (nesting) depth $d$. Then the VPA $\mathcal{M}(p)$ such that $\nw{p} = \mathcal{L}(\mathcal{M})$ has $\mathcal{O}((k+1)^d R)$ states. 
\end{theorem}
\vspace{-0.1in}
The proof goes by structural induction on the policy $p$. The proof and
definitions of depth and fan-out are given in 
\iffull
Appendix \S\ref{app:complexity}.
\else
the full paper.
\fi
Regarding the quantity $R$, note that the worst case size complexity of a DFA is
known to be exponential in the length of the regular expression
\citep{dfatheory}, but in the common cases this quantity is often much smaller.

\subsection{Runtime monitor for a VPA}
We will now formalize the monitor implementation for a VPA. This section first
describes the centralized interpretation, which follows immediately from the
standard VPA semantics, and then introduces a new distributed interpretation.
All the compiled VPAs are deterministic and complete, so we treat the VPA
transition relations as functions below.

\begin{definition}
A VPA $\mathcal{M} = (Q, q_{init}, F, \Sigma, \Gamma, \bot, \delta_c, \delta_r)$ can be interpreted as $\intcentral{\mathcal{M}}: \VPAConfig \to \setNW{\Sigma} \to \VPAConfig$, which takes: (a) an initial VPA configuration $(q, \theta)$---a pair of a state and some stack; (b) and a nested word $w \in \setNW{\Sigma}$, and returns a final VPA configuration. We define $\intcentral{\mathcal{M}}$ inductively on the length of the nested word $w$:
\begin{enumerate}
\item Case $w = (\epsilon, \phi)$:
\[\intcentral{\mathcal{M}}~ (q', \theta')~ w = (q', \theta')\] 
\item Case $w= (a_1 \ldots a_k a_{k+1}, \nu)$, where $a_{k+1} = (e, i)$ for some $e \in \Sigma$:
\[\intcentral{\mathcal{M}}~ (q', \theta')~ w = (q, \theta),~ \text{where}~ \intcentral{\mathcal{M}}~ (q', \theta')~ (a_1 \ldots a_k, \nu[\idx(a_1), \idx(a_k)]) \xrightarrow{e} (q, \theta)\] 
\end{enumerate}
\end{definition}
% High-level def

Intuitively, a distributed monitor is a set of call transition functions of the
type  $\shortcall: Q \to (Q \times \Gamma)$ and return transition functions of
the type $\shortret: (Q \times \Gamma) \to Q$. Here, $Q$ is a set of states and $\Gamma$ is a set of stack elements. In the following definition of the distributed monitor, $\setCall$ and $\setRet$ represent the set of all call and return functions.

\begin{definition}
A distributed monitor $\mathcal{D}: \Sigma \to (\monitorcall \times
~\monitorret)$ maps symbols in some alphabet $\Sigma$ to a pair of call and
return transition functions, where $\setCall$ and $\setRet$ represent the set of
all call and return functions of the type $\shortcall: Q \to (Q \times \Gamma)$
and $\shortret: (Q \times \Gamma) \to Q$. Here, $Q$ is the set of states and $\Gamma$ is a set of stack elements.
We write $\setdistmonitor$ for the set of distributed monitors.
\end{definition}

Now, we present a translation function $\extractmonitor: \setvpa \to
\setdistmonitor$ to convert a VPA $\mathcal{M}$ to a distributed monitor
$\mathcal{D} \in \setdistmonitor$.

\begin{definition} \label[definition]{def:distmon}
The distributed monitor for some VPA $\mathcal{M} = (Q, q_{init}, F, \Sigma, \Gamma, \bot, \delta_c, \delta_r)$, where $\basealpha$ is the base alphabet that gets call-return augmented into $\Sigma$, as:

\[\extractmonitor(\mathcal{M}) \triangleq \Set{(a, (\shortcall, \shortret))\ | \begin{array}{l}
   a \in \basealpha,\\
    \text{if}~\delta_c(q, \langle a) = (q', s)~\text{then}~\shortcall(q) = (q', s)~\text{and} \\ \text{if}~\delta_r(q', s, a \rangle) = q ~\text{then}~\shortret(q', s) = q
  \end{array}} \]

For any $a \in \basealpha$, the mappings in the call and return transition functions in $\mathcal{D}(a) = (\shortcall, ~\shortret)$ can be viewed as the transition rules in $\delta_c,~ \delta_r$ on the symbols $\vpacall{a}$ and $\vparet{a}$.
\end{definition}

Operationally, a distributed monitor takes a pair of state and stack as an initial configuration, a nested word  and returns a final configuration. Before defining the operational semantics for a distributed monitor, we introduce a single step transition operator $\xrightarrow{x}_{dist}$, where $x$ is a symbol in some alphabet $\Sigma$.

\begin{definition}[Single Step Transition]\label[definition]{def:singlestep}
Let $\mathcal{D}$ be a distributed monitor defined over a set of base alphabet $\basealpha$ such that for any $e \in \basealpha$,~ the pair of call-return transition mapped to $e$ is $\mathcal{D}(e) = (\shortcall, \shortret)$. Let $\Sigma$ be the call-return augmented alphabet of $\basealpha$ and $x$ be some symbol in $\Sigma$. Let $Q$ be the set of states, $\mu$ be the set of stacks, and the stack $\theta \in \mu$. The single step function $\xrightarrow{x}_{dist}: Q \times \mu \to Q \times \mu$ is defined as follows:
\begin{enumerate}
\item if $x = \vpacall{e}$ for some $e \in \basealpha$ then $(q, \theta) \xrightarrow{x}_{dist} (q', \theta')$, where $\shortcall(q) =  (q', s)$ and $\theta' = \theta s$,
\item if $x = \vparet{e}$ for some $e \in \basealpha$ then $(q, \theta) \xrightarrow{x}_{dist} (q', \theta')$ where $\shortret(q, s) = q'$ and $\theta' s = \theta$.
\end{enumerate}
\end{definition}

When a distributed monitor $\mathcal{D}$ reads a symbol $a_i = (x, \idx(a_i))$
from a nested word $w = (a_1 \ldots a_k, \nu)$, it selects a pair of call-return
transition functions $\mathcal{D}(e)$ such that $x = \vpacall{e}$ or $x =
\vparet{e}$. The pair $\mathcal{D}(e)$ is referred as a sub-monitor in the
distributed monitor $\mathcal{D}$. Based on the tag of the symbol $x$, the call
or return transition of $\mathcal{D}(e)$ is applied to the input configuration
of $\mathcal{D}$. Formally:

\begin{definition}[Operational Semantics]
A VPA $\mathcal{M}$'s distributed monitor $\mathcal{D} = \extractmonitor (\mathcal{M})$ can be interpreted as $\llbracket \mathcal{D} \rrbracket_{dist}: (Q \times \mu) \to \setNW{\Sigma} \to (Q \times \mu)$, where $\setNW{\Sigma}$ is the set of nested words on $\Sigma$. The run of a distributed monitor $\mathcal{D}$ on some nested word $w$ starting at some initial state $q$ and a distributed stack $\mu$ can be inductively defined on the length of the nested word $w$:

\begin{enumerate}
\item Case $w = (\epsilon, \phi)$:
\[\llbracket \mathcal{D} \rrbracket_{dist}~ (q, \theta)~ w = (q, \theta)
\]

\item Case $w = (a_1 \ldots a_{k} a_{k+1}, \nu)$, where $a_{k+1} = (e, i)$ for some $e \in \Sigma$:
\[\llbracket \mathcal{D} \rrbracket_{dist}~ (q, \theta)~ w = (q', \theta'),~ \text{such that} ~\llbracket \mathcal{D} \rrbracket_{dist} (q, \theta)~(a_1 \ldots a_k, \nu[\idx(a_1), \idx(a_k)]) \xrightarrow{e}_{dist} (q', \theta')\]
\end{enumerate}
\end{definition}

Finally, we can show that running the distributed and the centralized variant of a VPA are equivalent, i.e., they accepted the same nested words:
\begin{theorem} \label{thm:distmon}
Given a nested word $w$ over some call-return augmented alphabet $\Sigma$ and a VPA $\mathcal{M}$ whose run on $w$  is  $\rho(w) = (q_1, \theta_1) \ldots (q_n, \theta_n)$ then:
\begin{enumerate}
\item the run of $\mathcal{M}$'s centralized monitor is $\llbracket \mathcal{M}\rrbracket_{central} (q_1, \theta_1)~ w =  (q_n, \theta_n)$, and 
\item the run of $\mathcal{M}$'s distributed monitor $\llbracket \mathcal{D}\rrbracket_{dist} (q_1, \theta_1)~w =  (q_n, \theta_n)$.
\end{enumerate}
\end{theorem}
The proof goes by induction on the length of the nested word $w$.

In \cref{sec:implementation}, we will see how to use these distributed monitors
to enforce policies in an implementation. But first, we consider some example
policies that can be expressed in our policy language.

\section{Case Studies}\label{sec:case-studies}

This section motivates real-world service tree policies that are relevant to
three teams: data-compliance or audit teams, security teams, and deployment
teams. All case studies are presented in the context of the running example of
the hospital management application from \cref{sec:motivation}.

\paragraph*{Notation} We write \textsf{Endpoint} as shorthand for the set
$\{\textsf{Endpoint}\}$, and $S=\{a_1, \ldots, a_k\}$ as shorthand for the
regular expression $a_1 + \cdots + a_k$. We write $\any = \basealpha$ and
$\notservice{\textsf{Endpoint}}$ for the set $\basealpha -
\{\textsf{Endpoint}\}$. 

\subsection{Deployment team policies}
Thorough testing is a key step in the development process of a microservice application. Therefore, we describe case-studies about testing scenarios of varying complexities. 
\paragraph*{\textbf{Case study 1: A/B testing}} 
Say a deployment team wants to test the interaction of a small subset of beta testers, labeled as \textsf{Beta}, with a new version \textsf{v2} of the \database~ service. For example, the label \textsf{Beta} might have been assigned to a random subset of users by the frontend, or it could be assigned to internal users which will test the service before it is released publicly. So the deployment team requires all traffic coming from \textsf{Beta} to be served by \database~-\textsf{v2} instead of \database~-\textsf{v1}.

Let us model requests labeled as \textsf{Beta} to be requests from some endpoint labeled as $\textsf{Beta}$. The above A/B testing policy is specified as:
\[
  \hpolicy{\start \textsf{Beta}: \callseq~ \textsf{Beta} (\notservice{ \setdatabase~\textsf{-v1})^*}}.
\] 
Since this policy specifies a constraint on subtrees starting at \textsf{Beta},
the policy is of the form  $\start ~\textsf{Beta}: ~inner$. The $inner$ policy
needs to be matched on a subtree starting at \textsf{Beta}. The policy $inner$
has to match the sequence of API calls in the subtree rooted at \textsf{Beta}
with a regular expression that prevents calls to \textsf{v1} of \database~.

\paragraph*{\textbf{Case study 2: Factorial testing}} 
Deployment teams often want to test the interactions between all recently
updated services. For instance, suppose a deployment team wants each request to
either use the old \textsf{v1} version of all services, or the latest
\textsf{v2} version of all services; this is a \textit{factorial testing}
scenario. More concretely, consider that the deployment team wants to prevent
\test-\textsf{v2} from invoking \textsf{v1} versions of the \obfuscate~ and
\lab~services. This can be expressed as the policy:
\[
  \hpolicy{\start~\test\text{-}\textsf{v2}: \callseq ~(\any - \setobfuscate\textsf{-v1} - \setlab\textsf{-v1})^*}.
\]
Here, $(\any - \setobfuscate\textsf{-v1} - \setlab\textsf{-v1})$ is the set of
all endpoints besides \textsf{v1} of \obfuscate~and \lab.

\paragraph*{\textbf{Case study 3: Regional access control}}
Sometimes certain services need to have restricted regional access. For
instance, suppose the hospital application wants to prevent EU users from accessing the
main \database ~ service to avoid inadvertent violation of GDPR guidelines.
Supposing EU users are labeled at the frontend as \front-{EU}, we can express
this policy as:
\[
  \hpolicy{\start~\setfront\text{-EU}: \callseq ~\setfront\text{-EU}(\notservice{\database}~)^*}.
\]

\paragraph*{\textbf{Case study 4: External requirement}}
A deployment team might want to specify some business logic involving services
with global effects, like a write service that updates the database or an
account creation service that creates a new user. For example, new appointments
should be saved in the appointment database by invoking some \database~service.
This policy can be expressed as:
\[\hpolicy{\start~\setappt: \callseq~ (\setappt~ \kleenestar~\setdatabase~ \kleenestar~)}.\]
The pattern $\kleenestar$ matches any sequence of API calls. The \database ~service can be replaced by, say, a  \logging~ service to log admin access to some resource.

\subsection{Security team policies}
\paragraph*{\textbf{Case study 5: Payment logging}}
Suppose a security team wants all \payment ~requests to  call payment \database
~ at least once, and \database ~ to send requests only to \eventlogging~. The
team can state this requirement as ``\payment ~ invokes at least one \database ~
and this \database~ invokes \eventlogging~ as all its children,'' and specify
this policy as:
\[
  \hpolicy{\start \payment: \match~ (\payment~\setdatabase~)~\allpath ~(\eventlogging~) \kleenestar}.
\]
The $\match~(\payment~\setdatabase~)$ matches a subtree rooted at \payment~
that invokes ~\setdatabase~as its child.  To specify that all children of this
\setdatabase~are \eventlogging{}, the right side sub-expression of $\allpath$
matches all outgoing paths from \setdatabase~with the regular expression
(\eventlogging~) \kleenestar{}, where the pattern $\kleenestar$ matches any
sequence of API calls. 

\paragraph*{\textbf{Case study 6: Data Vault---no outgoing calls, a constraint on the leaves}}
Consider that the hospital application has a data \vault~service to store
patient records, and a security team wants to restrict \vault~ from invoking any
endpoints to prevent it from sharing confidential data with any third-party
services. This property can be expressed by requiring the \vault~service to have
no children in the service tree:
\[
  \hpolicy{\start~ \any:~\match~ \any \allpath (\notservice{\setvault})^* (\setvault + \epsilon)}.
\]

This can also be expressed as a singleton sequence policy: $\start~\setvault:\callseq~\setvault$.

\paragraph*{\textbf{Case study 7: Resource pricing}}
Suppose that patients can invoke \test service to request a batch of medical
tests. For each medical test, the \test service sends a child request. To ensure
correct billing, a security team might require every child call of \test to
directly or indirectly invoke \payment. This is policy can be specified as: 
\[
  \hpolicy{\start ~ \settest: \match~ \settest~ \allchildren (\match~\kleenestar~\setpayment~ \allpaths \kleenestar~)}.
\]

\subsection{Compliance team policies}
\paragraph*{\textbf{Case study 8: Data compliance}} 
Data protection laws, like HIPAA and GDPR, mandate compliance teams to
responsibly  handle customer data. For instance, personal health information,
like a patient's name, should be de-identified for privacy. Since compliance
teams do not have direct insight into the implementation of the application,
they would like to check compliance by monitoring inter-service communications.
To encode this kind of requirement as a service tree policy, they can require
that a service sending a request to external parties should have previously run
the de-identification service. This policy can be specified as:
\[\hpolicy{\start \settest: ~\match~ \settest \existschild p_1~\textsf{then}~p_2}, \]
where $p_1 = \match~ \setobfuscate \allpath \epsilon$ and $p_2 = \match~ \setlab \allpath \epsilon$. 

The right sub-expression of $\existschild$ needs to specify the existence of two subtrees satisfying  sub-policies $p_1$ and $p_2$. Here, $p_1, p_2$ will specify the existence of a call to \obfuscate and \lab respectively. 
The $\match$ sub-expression in $p_1$ and $p_2$ are   ``$\match~\obfuscate$''
and ``$\match~\lab$'' respectively because they need to match a subtree rooted
at \obfuscate and \lab.  In the above policy, the right side expression in
$p_1$ and $p_2$ is an empty string because the team wants \obfuscate~and \lab
to not invoke any APIs. This regular expression prevents the subtrees rooted at
\obfuscate and \lab from having any calls to children.  Note the regular
expression on the right of policy $p_1$ and $p_2$ could be replaced by
$\kleenestar$ if the compliance team wanted to allow \obfuscate and \lab to
invoke other services.

\textit{Add-on policy.} Suppose the compliance team wants each \test~request to send one \lab request. Additionally, it wants \test to de-identify the patient records before invoking \lab. This property specifies constraints on the sequence of API calls in  subtrees starting at \test, so the policy is: \[\hpolicy{\start \settest:~\callseq~(\notservice{\setlab})^*\setobfuscate(\notservice{\setlab})^* \setlab (\notservice{\setlab})^*}.\]
The above regular expression specifies that \obfuscate is called before  \lab.
%prohibit access to data store
\paragraph*{\textbf{Case study 9: Data Proxy or Middleware}}
%different structural property if datascrubber filters data. more like a proxy or firewall
Sometimes the access to a service needs to be managed by a proxy, such as a
firewall to secure a data source, a load balancer, or an authenticator. Consider
the compliance team wants \test to invoke \lab via an authentication service \authenticate ~to avoid creating unauthorized lab requests. More formally, we need to express that the subtree rooted at \test should have an \authenticate~ descendant, and \lab should be directly or indirectly invoked by \authenticate~, meaning \lab is \authenticate~'s descendant. 
This requirement can be specified as the following policy:
\[\hpolicy{\start~ \test:~\match~\settest \underbrace{\existschild}_{(1)} (\match~(\notservice{\setlab})^* \setauthenticate~ ~\underbrace{\existschild}_{(2)} (\match~\kleenestar\setlab \underbrace{\allpath}_{(3)} \kleenestar))}.\]
First, we specify \test as the start endpoint. Subtrees starting with \test
certainly need the root to be \test. This is described in the regular expression
associated with \match condition of the policy labeled with (1). The inner
policy to the right of (1) checks for the existence of a subtree that has a path to \authenticate~. The regular expression $(\notservice{\lab})^*\textsf{Auth}$ on the left of (2) permits \authenticate~, but no \lab in the path. After matching this path, the inner policy of (2) checks for the existence of $\lab$ as a descendant, which is similarly expressed as the expression $\match ~\kleenestar \lab$ on the left of (3). Since \lab can be followed by any endpoint, the regular expression on the right of the label (3) is   $\kleenestar$.

This policy can be strengthened to require that \authenticate~  is called before
the only call to \lab, as follows:
\[\hpolicy{\start \settest :~\callseq~(\notservice{\setlab})^*\setauthenticate ~(\notservice{\setlab})^* \setlab (\notservice{\setlab})^*}.\]
\section{Implementation} \label{sec:implementation}

To demonstrate our design, we developed a prototype implementation in $\sim 2k$
lines of Java. Our tool compiles a policy into a VPA, and then extracts a
runtime monitor that runs on top of the Istio service mesh. At its core, the
runtime monitor enforces the policy by simulating a VPA in a distributed manner,
using the distributed monitors in \cref{def:distmon}. Here, we discuss aspects
of our monitor's design that are particular to the Istio \cite{istio} service
mesh framework.

\paragraph*{Istio-based implementation} 
In the service mesh framework, each microservice container instance is paired with a
\emph{sidecar} container implemented as an Envoy proxy~\cite{envoydoc}. 
The sidecar can intercept all incoming and outgoing HTTP traffic for its
corresponding service container, and the Envoy proxy running at the sidecar can perform a variety of useful functions
orthogonal to our work, such as load balancing, service discovery, failover, etc.  In
addition, Envoy can be programmed with custom logic to inspect HTTP headers and
perform actions, like adding/removing/updating the headers, or allowing/denying HTTP traffic,
etc. Envoy proxies maintain an in-memory state for each request/response pair
for the duration of the request's lifetime. Therefore, any metadata saved in the
in-memory state during the request processing can be retrieved during its
response's processing. While not a core feature of the framework, this
capability turns out to enable some important optimizations for our enforcement
method, which we will discuss below.

\paragraph*{Local monitors as Envoy filters}
We implement our local monitors as Envoy filters, which are custom traffic
filtering Lua scripts that simulate VPA transitions on the service (symbol) from the current VPA configuration carried in the HTTP header.
Since in our setting, the service trees are \rmnw, the top of the stack symbol
read by a response message is the same as the value pushed by its matched
request. Therefore, instead of carrying the stack in the HTTP header, we save
the stack symbol locally in the proxy's in-memory state. This reduces the memory
overhead of propagating stack information along with requests, which can be
significant. With our design, only the VPA state is carried in the HTTP header. 

\paragraph*{Extracting local monitors}
%The Envoy filter at a given service simulates VPA transitions on the request/responses to this service. 
Given a VPA $\mathcal{M}$ and its distributed monitor $\mathcal{D} =
\extractmonitor{(\mathcal{M})}$ as defined in \cref{def:distmon}, our compiler
extracts the filter for a service $s \in \basealpha$  from the $\shortcall$ and
$\shortret$ function mapped to $s$, \textit{i.e.,} $(\shortcall, \shortret) =
\mathcal{D}(s)$. The \shortcall~ and \shortret~ functions are essentially the
VPA's call and return transitions on the service $s$.  The filter comprises two
callback functions:  \textsf{OnRequest} implements \shortcall, while \textsf{OnReponse} implements \shortret.

For example, the two callbacks for \payment's (\textsf{P}) local monitor  for the VPA in \cref{fig:VPA-payment} are given in \cref{fig:monitor-snippets}. The \textsf{OnRequest} function implements the  call transition on $\textsf{P}$ as a  conditional block that updates the \textsf{state} header and the custom in-memory metadata, $\textsf{local\_stack}$ to $\qp$ if the current \textsf{state} is $start$. Similarly, \textsf{OnResponse} implements the return transition on \textsf{P} that updates \textsf{state} to $start$ if the current \textsf{state} header and $\textsf{local\_stack}$ are $\qp$. 

\paragraph*{Local monitor execution}
When a request arrives at a service, the service's co-located proxy executes the filter's \textsf{OnRequest} callback to run a call transition. For instance, if the \textsf{state} header of an incoming request at \textsf{P} is set to $start$,  \textsf{OnRequest} updates \textsf{state} to $\qp$ and writes $\qp$ to the (custom) \textsf{local\_stack} metadata in \textsf{P}'s proxy's in-memory state (corresponding to this request's session). 

Likewise, the \textsf{OnResponse} callback is executed on intercepting a
response.  For instance, if the current \textsf{state} header of an outgoing
response from \textsf{P} is  $\qp$ and the in-memory \textsf{local\_stack} saved
during its corresponding request's processing was $\qp$,  \textsf{OnResponse}
updates the \textsf{state} header  to $start$. Thus, the callback implements our
distributed monitor's single-step transition (as defined in
\cref{def:singlestep}).

\kcomments{\begin{enumerate} 
\item We will draw the parallel between the distributed interpretation and the execution of the monitor in the servicemesh
\item First, we compile sub-monitors for each service. It is implementing the sub-monitors constructed in that definition. 
How it executes, well the semantics is that onrequest runs when receiving a request and response when sending out a response. For instance,
\item We encode the transition as conditional branch, it checks the source state and updates the state header and updates a variable for the stack. Istio's implementation of these callbacks is such that it keeps the session persistent, so we can set and retrieve the stack value. Since we deal with well-matched words here, it is sufficient to save it. That's what we do. 
\item When a request comes and this executes, it is like picking the service's sub-monitor and running it. For example,...
\item To keep the monitor efficient, we avoided sending the request and response to a centralized monitor that kept track of a service tree's run. This was to avoid sending extra traffic and also heavy-weight instantiation of the centralized monitor for each concurrent request. Instead, the VPA configuration is carried in the HTTP header and we run (a sub-) monitor co-located with the service. This is where Istio's service mesh features stand out. They offer us an out-of-the box solution suitable for in-band monitoring. Here's what istio does.
\item The distributed monitoring formalism captures that this execution is equivalent to a centralized monitor with a global view of the service tree. 
\item Give an example of the distributed sub-monitors.
\item The distributed monitor's implementation involves Lua filters as follows. This execution maps to a single step transition of the distributed sub-monitor definition in ...
\end{enumerate}}

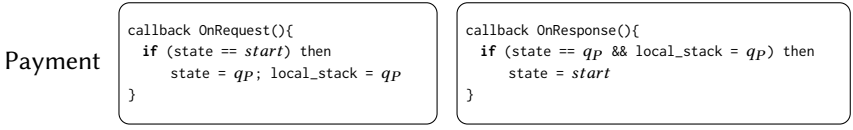
\begin{figure}
\begin{tikzpicture}[
    block/.style={draw, rectangle, rounded corners, minimum width=0.1cm, minimum height=1.5cm, font=\small}
  ]

  % Define 3 blocks
  \node (bp) at (-3,0) {
    \payment
    };
  \node[block] (b1) at (0,0) {
    \begin{minipage}{4cm}
    \lstset{style=codeblock}
    \begin{lstlisting}
callback OnRequest(){
  if (state == (*@$start$@*)) then 
      state = (*@$\qp$@*); local_stack = (*@$\qp$@*) 
}
    \end{lstlisting}
    \end{minipage}
  };
    \node[block] (b2) at (5,0) {
    \begin{minipage}{5cm}
    \lstset{style=codeblock}
    \begin{lstlisting}
callback OnResponse(){
  if (state == (*@$\qp$@*) && local_stack = (*@$\qp$@*)) then
      state = (*@$start$@*)
}
    \end{lstlisting}
    \end{minipage}
  };
\end{tikzpicture}
\caption{\payment's (\textsf{P}) local monitor  for the VPA in \cref{fig:VPA-payment}: \textsf{OnRequest} and \textsf{OnResponse} simulate VPA transition on requests and responses respectively. The top of the stack symbol  is locally saved as \textsf{local\_stack} metadata in the in-memory state of $\textsf{P}$'s proxy and \textsf{state} is carried in the HTTP header. Here, $qp,~start$ are VPA states.}
\label{fig:monitor-snippets}
\vspace{-0.1in}
\end{figure}
\paragraph*{Context propagation}
We require that the \safetree header carrying the VPA state is propagated from
an \emph{incoming} request to any resulting \emph{outgoing} request. Although
Envoy can see requests entering and leaving the service, the service itself is a
black box from its perspective, so Envoy does not directly know which incoming
(parent) requests produced which outgoing (child) requests. We therefore require
that the application copy the header from incoming to outgoing requests. This
functionality is known in the microservice community as \emph{context
propagation}; it is anyway required for other purposes---in particular,
distributed tracing to monitor application behavior and track the cause of
request failures and performance issues---and libraries exist to help implement
it~\cite{jaeger,zipkin}. \safetree does not assume anything more about
microservice applications beyond this standard requirement.

\paragraph*{Rejecting invalid service trees}
For simplicity, our implementation logs any policy violation after the request's entire service
tree has been processed, rather than actively blocking requests. It should be
possible to extend the prototype to block a request as soon as we know the policy must be violated---depending on the policy this can happen early or later in the service tree. For example, for policies of the form $\start S: inner$, which require subtrees starting at the symbols in $S$ to satisfy $inner$, the response from the root of such subtrees can be early blocked if the $inner$ policy is violated on the subtree. For certain policies, it is possible to block a request if transitioning on it will send the VPA into a state that's sufficient for it to never accept the service tree. For instance,  for a policy of the form $\start S: \callseq ~\regex$, a request can be blocked if transitioning on it will violate $\regex$. Similarly, for policies of the form $\start S: \regex \allpath \regex$ and $\start S: \regex \allchildren p$, a request can be blocked if the start symbol in the start set $S$ is not the same as the first symbol of all the words in the language of $\regex$. For instance, $\start \textsf{A}: \textsf{B} \allpath \textsf{B}$.

\section{Evaluation}  \label{sec:evaluation}
We evaluate two aspects of the \ourmonitor: its performance overhead and its
memory footprint by considering the following research questions:
\begin{itemize}
\item \textit{RQ1: How much header space is required for the context headers?}
\item \textit{RQ2: How much latency overhead does the monitor add?}
\end{itemize}
\paragraph*{Setup}
Our experimental setup consists of two microservice applications written in Go: a hotel reservation application from the DeathStarBench \citep{deathstar} and a simple hospital application that we wrote to exhibit the call structure of the application described in \cref{sec:motivation}. The average number of nodes in the service trees of both the applications are 4.5 and 6 respectively.

While the service implementations here---especially in the hospital
application---are rather simple, the specific logic inside the application does
not affect the overhead of the \safetree monitor since SafeTree runs in the
service mesh outside the application, so its performance is not affected by
application logic. SafeTree's overhead does, however, depend on the topology of
API calls and the policies being checked, which we will study in the evaluation.

The microservice applications are deployed on a minikube cluster enabled with
Istio sidecar injection. The cluster runs locally on a machine with 16 GB of
RAM, an i7 processor, and Ubuntu 22.04 operating system.
The application's inter-service communication is managed by the Envoy proxy
running in the Istio version $1.23.2$. \iffull We run our experiments 
on the policies
listed in  \cref{tbl:policies} in the Appendix, which are the case studies
presented in \cref{sec:case-studies}.
\else
The case study policies used in our experiments are listed in the full paper.
\fi 

\subsection*{RQ1: Header Space Overhead}
Maintaining the current state of the VPA-based runtime monitor and its stack
configuration is central to our enforcement mechanism. As described in
\cref{sec:implementation}, the stack configuration is saved locally at sidecar
proxies, but the current state of the  VPA is propagated alongside requests in a
custom HTTP header.
\Cref{tbl:eval} presents the total number of call and return transitions for
each policy's VPA (in columns \#call and \#return); total number of VPA states
(in \#state column); and the number of bits needed to encode the maximum number
of VPA states (in \#bits column). The context headers for all our policies are
at most six bits long, which is minimal compared to the available space in HTTP
headers (on the order of kilobytes).

\begin{table}[h!]
\caption{Policies prefixed with ``Hotel'' are evaluated on the hotel
  application, and the remainder are evaluated on the hospital application. Main
  findings: (1) context headers can encode the VPA state in a small number of
  bits, and (2) policy checking adds minimal latency, on the order of a
millisecond to the application.\vspace{-0.05in}}

\label{tbl:eval}
\centering
\begin{tabular}{ l c ccccccc }
\toprule
 & \multicolumn{2}{c}{VPA transitions} & \multicolumn{2}{c}{VPA states} & Latency   & Policy & Nesting\\
\emph{Scenario} & \#call & \#return & \#states & \#bits
          & overhead (ms) & class & \#levels\\
\midrule
  A/B Testing & 6 & 30    & 6 & 3  & 0.700 & Linear & NA\\
  Factorial Testing & 11 & 80  & 11 & 4 & 0.420  & Linear & NA\\
  Access Control    & 12 & 100 & 12 & 4&  0.448  & Linear & NA\\
  Update & 25&  544& 25 & 5 & 0.433 & Hierarch. & 2\\
  Data-compliance & 38& 1326 & 38 & 6 & 0.958 & Hierarch. & 2 \\
  Data Proxy & 36 & 1184 & 36 & 6 & 1.117  & Hierarch.  & 3\\
  Encryption & 23& 454   & 23 & 5 &  0.460 & Hierarch. & 1 \\
  Data Vault & 20& 346 & 20 & 5 &  0.370 & Hierarch. & 1 \\
  Resource pricing & 25&  562& 25  & 5   & 0.457 & Hierarch. & 2\\
  Hotel Encryption & 23& 454   & 23 & 5 &  0.216 & Hierarch. & 1 \\
  Hotel Data Proxy & 36 & 1184 & 36 & 6 & 0.443  & Hierarch.  & 3\\
  Hotel Compliance & 38& 1326 & 38 & 6 &0.278  & Hierarch. & 2\\
\bottomrule
\end{tabular}
\vspace{-0.05in}
\end{table}

We can understand how the size of VPAs vary across different classes of policies
if we look at the \#states column and the ``Policy class'' column, which
describes if a policy is linear or hierarchical. We observe that the linear
policies compile to a VPA with fewer states than the hierarchical policies. If
we look at the ``Nesting \#levels'' column, we can further observe that among
the hierarchical policies, more deeply nested policies tend to have more VPA
states. The data-compliance policy appears to be an outlier as it has the
greatest number of states even though other policies have deeper nesting, but
this is because the policy specifies multiple existential condition on subtrees.
\textit{\textbf{To summarize, the \ourmonitor requires only a few bits of extra
HTTP header space to compactly encode contextual information about the service
tree structure.}}

\subsection*{RQ2: Latency Overhead}
Another aspect of the \ourmonitor's evaluation is to understand its effect on
the application's performance. Accordingly, we compare the latency of requests
when the application is being monitored versus when it is not, on
a workload of 200 requests for all the user-facing endpoints of the application.
% The average latency of a request is measured using a (custom) HTTP benchmarking tool.
Latency benchmarking in a microservice application is prone to variance due to several network factors, like congestion, application's concurrency, elastic scaling, etc. We ensure that the application is not overloaded by sending the requests at a sufficiently low rate so that we are measuring per-request processing latency, rather than queuing effects. To minimize the variation in our results, we average the latency over five such workloads.

\paragraph*{Overhead versus Policy}
Our experiment involves extracting Envoy filters from the VPAs compiled in the
previous experiment. Then for each policy, latency overhead is measured as the
difference between the average request latency in the above applications when
the policy checking enabled and when it is disabled.  The policies in
\cref{tbl:eval} that are prefixed with ``Hotel'' were evaluated on the hotel
reservation application, and the remainder were evaluated on the hospital
application.  The  ``Latency overhead'' column in \cref{tbl:eval} reports the
average latency overhead in milliseconds.

Observe that all values are at most $2 ms$, which implies low latency overhead
of running the monitor. We found these results to be stable across both of our
benchmark applications, which is expected since the internal details of the
services should not affect the latency overhead. \textit{\textbf{To summarize,
\ourmonitor adds minimal latency overhead---on the order of a millisecond.}}

\paragraph*{Checking multiple policies}
To check multiple policies simultaneously, instead of unioning the VPAs of
individual policies, we run each VPA independently. This prevents blowing up the
VPA size, which is essential for maintaining low memory overheads. To understand
the latency overhead of running multiple policies, we run multiple copies of the
``Hotel Compliance'' monitor from \cref{tbl:eval}. 
\begin{minipage}[t]{0.38\textwidth} % Adjust width as n
\centering
\small
\setlength{\tabcolsep}{3pt} % Reduce column spacing
\renewcommand{\arraystretch}{0.9} % Reduce row spacing
\vspace{-0.4in} % Add this line
\captionof{table}{Overhead for multiple policies.\vspace{-0.1in}}
\label{tab:overhead}
\begin{tabular}{cc}
\toprule
\textbf{\# Policies} & \textbf{Latency overhead (ms)} \\
\midrule
1 & 0.278 \\
2 & 0.510 \\
3 & 0.610 \\
4 & 0.676\\
\bottomrule
\end{tabular}
\end{minipage}
\hfill 
\begin{minipage}{0.6\textwidth} % Left section for text
The column ``\# Policies'' in \cref{tab:overhead} describes the number of simultaneous policies being checked. We can conclude from the  ``Latency overhead'' column that the latency overhead increases with the increase in the number of policies.
\textbf{The results suggest we could feasibly monitor multiple policies with a reasonable amount of overhead.}
 \end{minipage}

\vspace{-0.1in}
\begin{figure}[h!]
    \centering
    \begin{subfigure}[b]{0.32\textwidth}
        \centering
        \includegraphics[height=3cm, width=\textwidth]{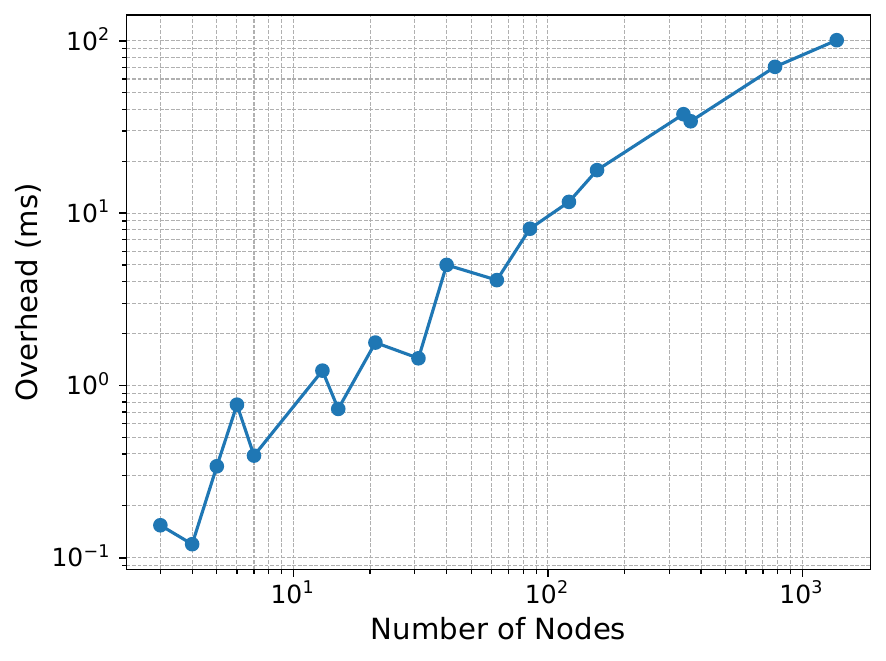}
        \caption{Overhead vs Nodes}
        \label{fig:scaleall}
    \end{subfigure}
   \begin{subfigure}[b]{0.32\textwidth}
        \centering
        \includegraphics[height=3cm, width=\textwidth]{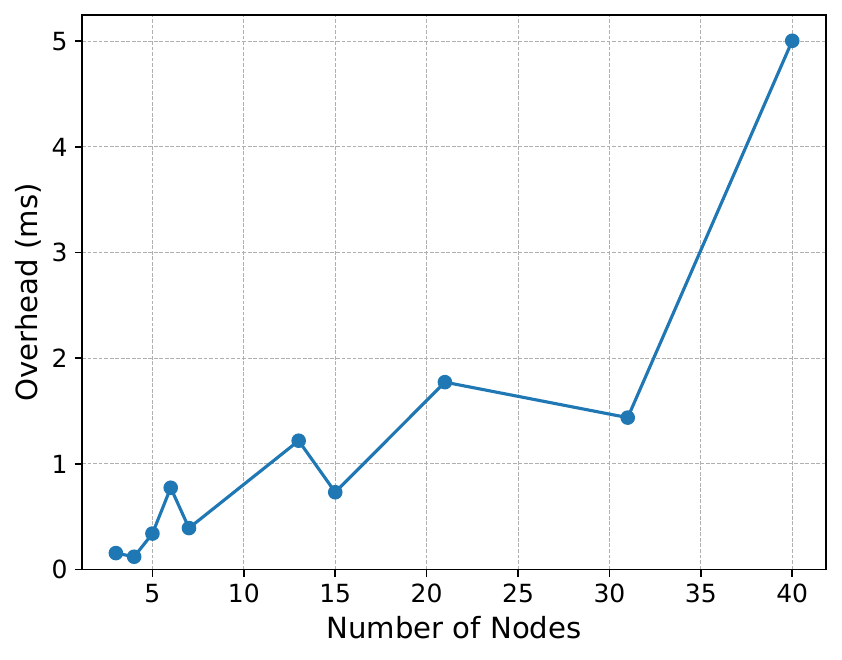}
        \caption{Overhead vs Nodes ($\le$40)}
        \label{fig:scalezoom}
    \end{subfigure}
    \begin{subfigure}[b]{0.32\textwidth}
        \centering
        \includegraphics[height=3cm, width=\textwidth]{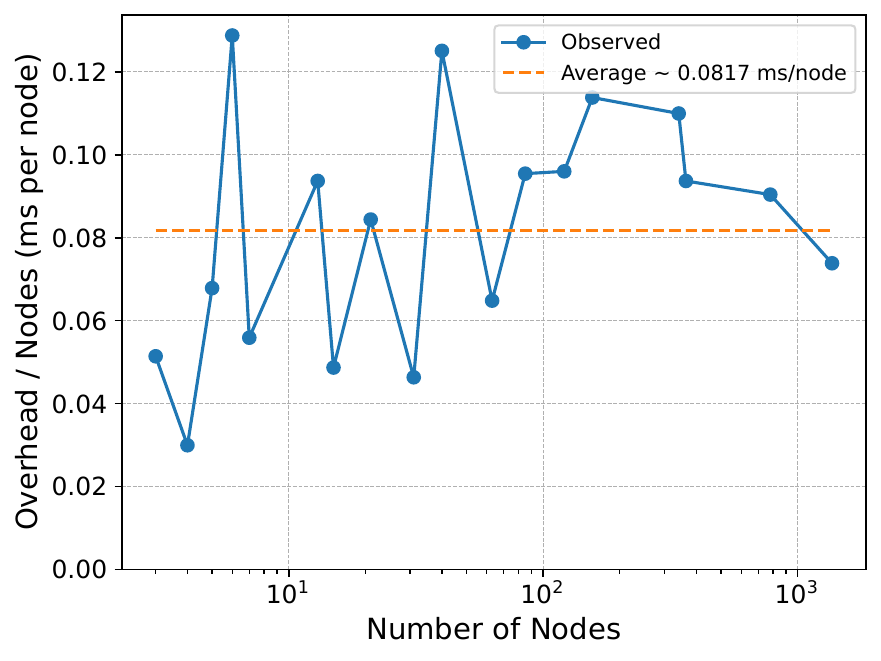}
        \caption{Overhead to Nodes Ratio}
        
    \label{fig:overheadfactor}
    \end{subfigure}
    \caption{Latency overhead vs topology scale, measured as number of nodes in the application's service tree.}
    \label{fig:scalability}
\end{figure}

\paragraph*{Overhead versus topology scale}
To evaluate SafeTree monitor's scalability, we measure the increase in latency overhead with the scale of the application's topology, which we measure as the number of nodes in the application's service tree.  For this experiment, we synthetically generate applications with different topology shapes---all combinations of depths from 2 to 5 and fan-out from 1 to 4, resulting in the total number of nodes in the service trees ranging from 3 to 1365. We measure the latency overhead for the same (``Hotel Compliance'') policy across different applications.

\Cref{fig:scaleall} shows that the latency overhead in milliseconds (on the y-axis) linearly increases with the increase in the number of nodes in the service tree (on the x-axis). 
\Cref{fig:scaleall}  uses log scale on both the x-axis and y-axis. As a reference point, note that the data collected by Alibaba \citep{alibaba} showed that in their microservice deployment, the common case depth and fan-out per service is lower than the median depth of $\sim4$ and the  median fan-out of $\sim2$. A binary tree with depth $4$ will have maximum $31$ nodes; the common case trees will have fewer nodes. 

We zoom into the latency overhead for trees with nodes between 0--40 in
\cref{fig:scalezoom}, where the latency overhead in milliseconds is plotted on
the y-axis, and the x-axis reports the number of nodes on a linear scale. Notice
that the common case overhead is at most $5$ ms and the overhead for a typical
tree with $\sim 31$ nodes is at most $2$ ms. 

We plot the per-hop overhead  (latency overhead divided by number of nodes) in
milliseconds on \cref{fig:overheadfactor}'s y-axis for trees with a range of
nodes (plotted on the log scale on the x-axis). We can see that on an average
$0.082$ ms of latency overhead is incurred per hop in the service tree.
\textit{\textbf{In summary, \safetree monitor's latency overhead is under
$\sim5$ms for a typical topology. The overhead linearly increases with the
tree's size, adding $~0.082$ms of average overhead per hop.}}
We expect that a production implementation of SafeTree could achieve
significantly better performance, e.g., by writing filters using C++/WASM
instead of Lua.

\section{Related Work} \label{sec:related}
\paragraph*{Safety in microservice and cloud applications.}
Today's microservice systems support policies that control communication between
pairs of microservices. Recent works have explored more general policies. For
instance, Trapeze~\citep{DBLP:journals/pacmpl/AlpernasFFRSSW18} is a system for
dynamic information flow control~\citep{DBLP:journals/jsac/SabelfeldM03} in
serverless computing. Trapeze can precisely specify how data at different
security levels flow around the application. In contrast, our system can specify
properties about the structure of the API call tree.  

Another interesting work in this area is Whip \citep{whip}, a higher-order
contract system for describing service-level specifications as contracts on
blackbox services. Whip policies can describe the arguments and return values of
microservices; however, these policies are focused on individual API calls, and
require a custom network adapter for monitoring. In contrast, our approach can
express policies about trees of API calls using a lightweight monitoring
approach that can be implemented in existing service mesh frameworks.

Our prior work~\citep{DBLP:conf/hotnets/GrewalGH23} proposed policies for
microservices based on a linearization of the service calls; \safetree is more
general in supporting policies that describe the tree structure of the calls,
which requires a richer automaton (VPA) for monitoring. The
Copper~\citep{copper} system also uses this idea of linearization to combine
several single hop policies. Unlike \safetree, Copper does not support tree
policies, nor policies over sequences of API calls.

\paragraph*{Execution correctness in serverless runtime}
A recent line of work aims to make it easier to correctly execute microservice
applications on serverless platforms. For example, serverless operational
semantics were formalized in the $\lambda_{\lambda}$ calculus
\cite{lambdalambda}; language primitives were introduced in $\mu2sls$
\cite{mu2sls} to write microservice code with transaction and asynchrony
abstractions without manually handling failures and execution nuances; durable
functions \cite{durablefn} introduced a programming model for ensuring
orchestration correctness of stateful workflows under retries and failures in
otherwise stateless serverless applications. \safetree assumes that
microservices are executed correctly, but instead checks that the runtime
behavior conforms to some desired specification.

\paragraph*{Programmable runtime verification in networks.}
% Security and RV in programmable networks: Poise, Hydra, DBVal, Aragog

Our work can be seen as a programmable system for runtime verification in the
service mesh layer. Runtime verification has been used at other layers of the
network stack. For instance, DBVal~\citep{DBLP:conf/sosr/KumarKPAUT21}  checks
for packet forwarding correctness.
Hydra~\citep{DBLP:conf/sigcomm/RenganathanRKVC23} allows the network operator to
enforce  their desired policy and specify custom telemetry to attach to packets
in a custom high-level language, and synthesizes monitors to enforce polices in
the dataplane. For security policies, Poise~\citep{DBLP:conf/uss/Kang00TCL20}
converts high-level security and access control policies into P4 code for
enforcement. At a higher level, Aragog~\citep{DBLP:conf/osdi/YaseenABCL20} is a
scalable system for specifying and enforcing policies about high-level events in
networked systems.

\paragraph*{Distributed and decentralized runtime verification.}
% Distributed and decentralized runtime verification

Distributed and decentralized monitoring has been well-studied in the runtime
verification community, but most works target distributed systems and are
designed to contend with monitors that may see events out of
order~\citep{DBLP:journals/tocl/BasinKZ20} or monitors that may see different
views of the global system~\citep{DBLP:conf/ipps/MostafaB15}.  We do not face
such difficulties in our setting: the state of our distributed monitors is
carried with the request, and so each monitor has the full information required
to enforce the policy. Distributed tracing frameworks \citep{jaeger, zipkin,
otel} used for observability rely on a centralized collector for recording
execution traces, which can then be analyzed later. In contrast, a key goal of
our work is to check policies and detect violations efficiently at runtime,
rather than after the fact.

\paragraph*{Automata-theoretic models for hierarchical structure.}
Our policy specification and enforcement is based on nested word languages and
visibly-pushdown automata~\citep{nestedwords}.  The literature on nested words
and VPAs is too large to survey here; the interested reader should consult
\citet{nestedwords}.  While tree automata \cite{tata, treeacceptors} also work
on hierarchical input, their inputs must be structured trees. \safetree uses VPA
because its input is a serialized tree of nodes being incrementally processed,
not the entire service tree structure. 

\paragraph*{Monitoring context-free properties.}
\safetree policies resemble to context-free properties, which can be captured in
extensions of temporal logic.  For instance, VLTL \citep{vltl} supports a more
complex policy language, although monitoring such policies requires B\"uchi and
parity automata, which seem difficult to realize in a microservice setting.
Other examples include CaRet \citep{caret} and PTCaRet \citep{ptcaret}, temporal
logic extensions with call-return matching that can be interpreted using
recursive state machines. Runtime monitoring algorithms that have been
considered for context-free properties specified in temporal logic are include
formula rewriting-based pushdown-automata for PtCaRet \citep{ptcaret}; pushdown
Mealy machine for CaReT with future fragment \cite{futurecaret}; and an LR(1)
parsing based algorithm for parametric properties \citep{javamop}.

In the realm of software verification, PAL \citep{pal} is a DSL for writing
context-sensitive monitors for C programs, where the user directly encodes the
low-level state transitions of the monitor, and the framework automatically
instruments an existing C program.  In contrast, the \safetree policy language
is high-level, and the monitoring automaton is generated automatically.

\section{Conclusion}\label{sec:conclusion}

We have presented \safetree, a policy language for specifying rich, tree-based
safety properties for microservices. By compiling policies to VPA, we  derive an efficient and performant distributed runtime monitor for
enforcing our policies without invasive code changes to
microservices.

We see several possibilities for future work. First, extending our work to the
asynchronous setting, where API calls are processed concurrently, will expand
the applications of our work. However, it is unclear how to specify safety
policies where part of the service tree may not be completed yet, and
asynchronous calls may return in a different order than they were originally
issued, leading to service trees that may not be well-matched.  Second, it can
be useful to support policies that reason about arguments of API calls. It
could also be interesting to support richer nested word languages---our design
uses just call and return symbols, but internal symbols might be useful for
modeling other aspects of microservice behavior.  Finally, our monitoring
strategy could be useful beyond microservices; for instance, for network control
planes.

%, or for higher-level communications, like remote procedure calls.

\section*{Acknowledgements}
We thank the reviewers for their constructive feedback.
This work is supported by NSF grants \#2152831 and \#2312714. This work benefited from discussions with Loris D\textquotesingle Antoni and the HTTP benchmarking tool by Talha Waheed. 
\section*{Data-Availability Statement}
Our artifact \cite{zenodo-safetree} consists of the SafeTree policy compiler \cite{safetreecompiler}, which takes in a policy,
and generates a VPA followed by emitting a distributed monitor. It also contains the source code of the policy compiler to generate a monitor; source for the benchmark applications; a list of test policies; and workload data to reproduce the experiments. 
\bibliography{header,PLDI}
\iffull
\input{appendix}
\else
\fi
\end{document}

%% file: appendix.tex
\newpage
\appendix

%\crefname{section}{Appendix}{Appendices} 
\crefname{lemma}{lemma}{lemmas}
\Crefname{lemma}{Lemma}{Lemmas}
\section{Case Study Policies}
\begin{table}[h!]
\caption{Case studies}
\label{tbl:policies}
\begin{tabular}{ l l }
\toprule
Scenario & Policy \\
\midrule
  A/B Testing & \makecell[l]{$\start \beta: \callseq~ \beta (\notservice{ \setappt\textsf{-v1})^*}$}\\   
   \\ 
  Factorial Testing & \makecell[l]{$\start~\settest$-$\textsf{v2}: \callseq ~reg$\\
                                   $reg = (\any - \setobfuscate\textsf{-v1} - \setlab\textsf{-v1})^*$}\\
   \\
  (Regional) Access Control & \makecell[l]{$\start~\setfront\text{-EU}: \callseq~ reg$\\
                                         $reg = ~\setfront\text{-EU}(\notservice{\database}~)^*$}\\
  \\
  External Requirement & \makecell[l]{$\start~\setappt: \match~ \setappt~ \existschild p$\\
                                      $p = \match~\kleenestar\setdatabase~ \allpath~ \kleenestar$} \\
  \\
  \hline 
  Data-compliance &  \makecell[l]{$\start \settest: ~\match~ \settest \existschild p_1~\textsf{then}~p_2$\\
                      $p_1 = \match~ \setobfuscate \allpath \epsilon$ \\
                      $p_2 = \match~ \setlab \allpath \epsilon$}\\
   \\
  Data Proxy &  \makecell[l]{$\start~ \settest:~\match~\settest \existschild p'$ \\
                             $p' = \match~(\notservice{\setlab})^* \setauthenticate ~\existschild p''$\\
                             $p'' = \match~\kleenestar\setlab \allpath \kleenestar$}\\
  \\
  \hline
  Encryption & \makecell[l]{$\start \setfront: \match~ (\setfront~\setappt)~\allpath ~reg$\\
                            $reg = (\setpayment~ \setencrypt ~+~ \notservice{\setpayment}) \kleenestar$}\\
  \\
  Data Vault &  \makecell[l]{$\start~ \any:~\match~ \any \allpath (\notservice{\setvault})^* (\setvault + \epsilon)$\\
                            $p' = \match~\any^*\payment~ \allpaths \any^*$}\\
  \\
  Resource pricing & \makecell[l]{$\start ~ \settest: \match \settest~ \allchildren (\match~\kleenestar~\setpayment~ \allpaths \kleenestar~)$} \\
\bottomrule
\end{tabular} 

\end{table}
\newpage

\section{Omitted Definitions (in \S\ref{sec:treeconcepts})}

\begin{definition}[Root of a nested word]\label{def:r-nw}
A \mnw $n = (a_1 \ldots a_l,\nu)$ is \textit{rooted} if $\nu(\idx (a_1), \idx (a_l))$ holds  . Here $a_1$ is called the root.
\end{definition}

\begin{definition}[Call Sequence from root to a leaf]
For a \rmnw $n = (a_1 \ldots a_l, \nu)$, a call sequence $\pi(n, a_m) = a_1 \dots a_m$ to some $a_m \in n$ is a \textit{call sequence to leaf} if $\nu(\idx(a_m)) = \idx(a_{m+1})$ and $a_{m+1} \in n$. The set of all such sequences is given by: 
\[
  SeqLeaf(n=(a_1 \ldots a_l, \nu)) =  \Set{ \pi (n, a_m)  \mid \begin{array}{l}
    a_{m+1}, a_m \in n,
    ~\nu(\idx(a_m)) = \idx(a_{m+1})
  \end{array}}
\]

\end{definition}

\section{Omitted Compilation Rules (in \S\ref{sec:compilation})}\label{app:compiler}
\textit{Note: For clarity and readability, throughout this appendix, we use relational notation for the call $\delta_c: \calltype$ and return $\delta_r: \rettype$ transition functions in a deterministic VPA. We represent VPA call transitions as a set of 4-tuple of the form $(q, s, q', g)$, where $\delta_c(q, s) = (q', g)$ and similarly return transitions as a set of 4-tuple of the form $(q, g, s, q')$, where $\delta_r(q, g, s) = q'$. Since the VPA is deterministic, each 4-tuple appears at most once ensuring that this relational presentation adheres to the functional semantics. Similarly, the transitions  $\delta: Q \times \Sigma \to Q$ of a DFA is a function, and for consistency with the VPA's notations, we will represent the DFA transitions as a set of triples of the form $(q, s, q')$, where $\delta(q, s) = q'$. 
}

We define \[\vp{.}: Policy \to \setvpa\] as follows:

\label{automata}
\begin{enumerate}
\item VPA for $\textsf{\textbf{match}}~reg_1 \allpath reg_{2}$.

Let the \DFA~ for $reg_1$ and $reg_2$ be $\mathcal{A}_1$ and $\mathcal{A}_2$ respectively, where
\begin{align*}
    &\dfa{1}, \\
    & \dfa{2}.
\end{align*}
Then $\vp{\textsf{\textbf{match}}~reg_1 \allpath reg_{2}} = \mathcal{M}^p$, where
\begin{align*}
    & \mathcal{M}^{p} = (Q,~ q_{beg},~ \{q_{end}\},~ \Sigma,~ \Gamma,~ \bot,~ \delta_{c},~  ~ \delta_{r}) \\
    & \Sigma = \{\langle e, e \rangle \mid e \in \tilde{\Sigma}\}\\
    & \tilde{Q^2} = \{\textsf{aug}(q^2) \mid q^2 \in Q^2\}\\
    & Q = \{q_{beg}, q_{end}, q_{sat}, q_{rej}^2\} \cup Q^1 \cup Q^2 \cup \tilde{Q^2}.
\end{align*}
The call transition is given by $\delta_c = \delta_{c1} \cup \delta_{c2} \cup \delta_{c3}$. We define the sets $\delta_{c1}$, $\delta_{c2}$, and $\delta_{c3}$ below. 
$\delta_{c1}$ defined below encodes the rules to simulate $\mathcal{A}_1$ for finding a path from the root that matches $reg_1$:
\begin{align*}
\delta_{c1} = ~ & \{(q_{beg}, \langle s, q, q_{beg}) \mid (q_{init}^1, s, q) \in \delta^1\}\\
& \cup \{(q_{src}, \langle s, q_{dst}, q_{src}) \mid(q_{src}, s, q_{dst}) \in \delta^1, ~q_{src} \notin F^1\}.
\end{align*} 
Next, $\delta_{c2}$ simulates $\mathcal{A}_2$ after any final state of $\mathcal{A}_1$, say $q_{final}^1$ is reached. These call transitions push an element of $\tilde{Q^2}$ on the stack. According to the second rule below, if some source state is $q_{src}$ then the VPA pushes the corresponding augmented state $\textsf{aug}(q_{src})$ on the stack. 
\begin{align*}
\delta_{c2} =~ & \{(q^1_{final}, \langle s, q_{dst}, \textsf{aug}(q^2_{init}))  \mid  q_{final}^1 \in F^1, ~(q_{init}^2, s, q_{dst}) \in \delta^2\}\\
& \cup \{(q_{src}, \langle s, q_{dst}, \textsf{aug}(q_{src}))  \mid  (q_{src}, s, q_{dst}) \in \delta^2\}\\
&\cup \{(\textsf{aug}(q_{src}), \langle s, q_{dst}, \textsf{aug}(q_{src})) \mid (q_{src}, s, q_{dst}) \in \delta^2\}
\end{align*}
Finally, $\delta_{c3}$ specifies that once the VPT finds a path matching $reg_1$ such that all the outgoing paths starting at the children of the \textit{end} node of this path also match $reg_2$ then the VPT remains in the satisfied state $q_{sat}$. The second rule describes that once in the reject state $q_{rej}^2$ , the VPT will remain in that state. The last rule ensures that the VPT does not read any symbols after transitioning to $q_{end}$, which is meant to be transitioned to only at the last symbol of a word.
\begin{align*}
\delta_{c3} = ~& \{(q_{sat}, \langle s, q_{sat}, q_{sat})\}\\
& \cup \{(q_{rej}^2, \langle s, q_{rej}^2, q_{rej}^2)\} \\
& \cup \{(q_{end}, \langle s, q_{rej}^2, q_{rej}^2)\}
\end{align*}
To define the return transition $\delta_r = \delta_{r1} \cup \delta_{r2} \cup \delta_{r3} \cup \delta_{r4}$, we first define $\delta_{r1}$ to backtrack $\mathcal{A}_1$ on return symbols:
\begin{align*}
\delta_{r1} = ~& \{(q^1, q, s \rangle, q)  \mid  q^1, q \in Q^1 - F^1\}\\
  &  \cup \{(q^1, q_{beg}, s \rangle, q^1)  \mid  q^1 \in Q^1-F^1\} \\
    & \cup \{(q_{rej}^2, q^1, s \rangle, q^1)  \mid q^1\in Q^1-F^1\}
\end{align*}
Next, $\delta_{r2}$ describes the transitions to: (a) the state $q_{sat}$ that indicates that one occurrence of a path satisfying the $\forall\textsf{-\textbf{path}}$ policy has been found; (b) the accepting state $q_{end}$. This set is defined as:
\begin{align*}
\delta_{r2} = ~& \{(q_{final}^1, q^1, s \rangle, q_{sat}) \mid q^1 \in Q^1-F^1,~q_{final}^1 \in F^1\}\\
& \cup \{(q_{sat}, q, s \rangle, q_{sat})  \mid q \in (Q^1 \setminus F^1)  \cup \tilde{Q^2} \cup \{q_{sat},  q_{rej}^2\}\}\\
& \cup \{(\textsf{aug}(q^2_{init}), q^1, s \rangle, q_{sat})  \mid q^1 \in Q^1-F^1, ~(q_1, s, q_{final}) \in \delta^1,~q_{final} \in F^1\}\\
  & \cup \{(q_{final}^1, q_{beg}, s \rangle, q_{end})  \mid q_{final}^1 \in F^1\}\\
      & \cup \{(\textsf{aug}(q_{init}^2), q_{beg}, s \rangle, q_{end})\}\\
  & \cup \{(q_{sat}, q_{beg}, s \rangle, q_{end})\}
\end{align*}

The set $\delta_{r3}$ describes the backtracking of $\mathcal{A}_2$:
\begin{align*}
  \delta_{r3} = ~& \{(q', q, s \rangle, q)  \mid  q' \in F^2, q \in \tilde{Q^2}\}\\
  & \cup \{(q^2, q^1, s \rangle, q^{1})  \mid  q^1 \in Q^1-F^1, q^2 \in Q^2\} \\
  & \cup \{(q', q, s \rangle, q)  \mid  q', q \in \tilde{Q^2}\}\\
  & \cup \{(\textsf{aug}(q^2_{init}), q^1, s \rangle, q^1)  \mid q^1 \in Q^1-F^1, ~(q^1, s, q_{final}) \notin \delta^1,~q_{final} \in F^1\} \\
  & \cup \{(q, q^1, s \rangle, q^1)  \mid q^1 \in Q^1-F^1, q \in \tilde{Q^2} - \{\textsf{aug}(q^2_{init})\}\}
  \end{align*}
  
The set $\delta_{r4}$ describes transitions that go to $q_{rej}^2$ state:
\begin{align*}
\delta_{r4} = ~& \{(q^1, q_{sat}, s \rangle, q_{rej}^2) \mid q^1 \in Q^1\} \\
  & \cup \{(q^1, q^2, s \rangle, q_{rej}^2) \mid q^1 \in Q^1, ~q^2 \in \tilde{Q^2}\} \\
  & \cup \{(q', q, s \rangle, q_{rej}^2)  \mid  q' \in Q^2 - F^2, ~q \in \tilde{Q^2}\}\\
  & \cup \{(q^2, q, s \rangle, q_{rej}^2)  \mid  q^2 \in Q^2, q \in \{q_{beg}, q_{sat, q_{rej}^2}\} \}\\
  & \cup \{(q, q^2, s \rangle, q_{rej}^2)  \mid q^2 \in \{q_{sat}, q_{rej}^2\}, q \in \tilde{Q^2}\}\\
        & \cup\{(q, q_{beg}, s \rangle, q_{rej}^2) \mid  q \in \tilde{Q^2} - \{\textsf{aug}(q_{init}^2)\}\}\\
  & \cup \{(q_{rej}^2, q^2, s \rangle, q^2_{rej})  \mid  q^2 \in \{ q^2_{rej}, q_{sat}, q_{beg}\} \cup \tilde{Q^2}\}\\
  & \cup \{(q_{beg}, q, s \rangle, q_{rej}^2)  \mid  q \in \Gamma\} \\
  & \cup \{(q_{end}, q, s \rangle, q_{rej}^2)  \mid  q \in \Gamma\}\\
       & \cup \{(q, q_{rej}^2, s \rangle, q_{rej}^2) \mid  q \in Q^1\}
\end{align*}

\item VPA for $\match~ reg \allchildren p_1$.

Let the VPA for $p_1$ be $\mathcal{M}^{p_1} = (Q^{p_1},~ q_{beg}^{p_1},~ F^{p_1},~ \Sigma,~ \Gamma,~ \bot,~ \delta_{c}^{p_1}, ~ \delta_{r}^{p_1}).$\\
Let the \DFA~ for $reg$~be  $\dfa{1}$.\\
Then $\vp{\match~ reg \allchildren p_1} = \mathcal{M}^p $, where 
\begin{align*}
    & \mathcal{M}^{p} = (Q,~ q_{beg},~ \{q_{end}\},~ \Sigma,~ \Gamma,~ \bot,~ \delta_{c},~  ~ \delta_{r}), \\
    & Q = \{q_{beg}, q_{end}, q_{sat}, q_{rej}\} \cup Q^{p_1} \cup Q^1, \\
    & \Gamma = \Gamma^{p_1} \cup \{q_{beg}, q_{sat}, q_{rej}\} \cup (Q^1 \setminus F^1).
\end{align*}

The call transition is given by $\delta_c = \delta_{c1} \cup \delta_{c2} \cup \delta_{c3}$. We define the sets $\delta_{c1}$, $\delta_{c2}$, and $\delta_{c3}$ below. 
$\delta_{c1}$ defined below encodes the rules to simulate $\mathcal{A}_1$ for finding a path from the root that matches $reg_1$:
\begin{align*}
\delta_{c1} =~&\{(q_{beg}, ~ \cs , ~q, q_{beg})  \mid (q^{1}_{init}, ~s, ~q) \in \delta^1,~ q\in Q^1 \} \\
& \cup \{(q_{src}, \cs, q_{dst}, q_{src})  \mid  (q_{src}, s, q_{dst}) \in \delta^1 ,~q_{src} \notin F^1\}
\end{align*}

After finding a \textit{first match} path for $reg$, the VPA starts to run $\mathcal{M}^{p_1}$ (captured by the first rule in the definition below). To simulate $\mathcal{M}^{p_1}$, we add all its call transition elements, but those that transition from the final state of $\mathcal{M}^{p_1}$. Instead, transitions from the final state of $\mathcal{M}^{p_1}$ simulate $\mathcal{M}^{p_1}$'s transition from its initial state. The idea is that we need to restart simulating $\mathcal{M}^{p_1}$ on all children of a given subtree to match the sub-policy $p_1$. The set $\delta_{c2}$ is defined as: 
\begin{align*}
\delta_{c2} =~& \{(q^1_{final}, \cs , q_{dst}, q_{stack})  \mid (q_{beg}^{p1}, \cs , q_{dst}, q_{stack}) \in ~ \delta^{p_1}_{c}, ~ q_{final}^1 \in F^1\}\\
& \cup (\delta_{c}^{p_1} - \{(q_{end}^{p_1}, \langle s, q, q_{stack}) \in \delta_{c}^{p_1}  \mid  q_{end}^{p_1} \in F^{p_1}\}) \\
& \cup \{(q_{end}^{p_1}, \langle s, q, q_{stack}) \mid  (q_{beg}^{p_1}, \langle s, q, q_{stack}) \in \delta_{c}^{p_1}, ~ q_{end}^{p_1} \in F^{p_1}\}
\end{align*}

The state $q_{sat}$ describes that a path has been found that satisfies $reg$ and the children subtrees of that path's endpoint satisfy $p_1$. Therefore, as described in $\delta_{c3}$ below, any further call transition keeps the VPA in $q_{sat}$. Similarly, $q_{rej}$ is the special state that denotes the policy $p_1$ has been violated and call transitions keep the VPA in $q_{rej}$. Lastly,  in this complete VPA, any transitions from $q_{end}$ or the final state take the VPA to $q_{rej}$. The idea is that $q_{end}$ is supposed to be transitioned to only at the end of the word or the matching return for the first call symbol in the word. The set $\delta_{c3}$ is defined as:
\begin{align*}
\delta_{c3} =~& \{(q_{sat}, \cs, q_{sat}, q_{sat})\} \\
& \cup \{(q_{rej}, \cs, q_{rej}, q_{rej})\} \\
& \cup \{(q_{end}, \cs, q_{rej}, q_{rej})\}
\end{align*}

The return transition function is given by $\delta_r = \delta_{r1} \cup \delta_{r2} \cup \delta_{r3} \cup \delta_{r4} \cup \delta_{r5}$. We define these sets below. Firstly, the elements of $\delta_{r1}$ describe the transitions of $\mathcal{M}^{p_1}$ with the exception of transitions from $\mathcal{M}^{p_1}$'s final state and those with stack symbol $q_{beg}^{p_1}$. This is done because we want to restrict these transitions. For all return transitions in $\mathcal{M}^{p_1}$ destined to a  non-final state in $F^{p_1}$ after starting at $\delta_{r}^{p_1}$ and with stack symbol  $q_{beg}^{p_1}$, the destination state is set to $q_{rej}$. These rules are captured in:
\begin{align*}
&\delta_{r1} =~\delta_{r1'} \cup  \{(q, q_{beg}^{p_1}, \rs, q_{rej}) \mid  (q, q_{beg}^{p_1}, \rs, q') \in \delta_{r}^{p_1}, ~q' \notin F^{p_1}\}, \text{where} \\
&\delta_{r1'} =~\delta_{r}^{p_1} \setminus \delta_{\text{no-final}} ~\text{and} \\
&\delta_{\text{no-final}} =~ \{(q_{end}^{p_1}, q, \rs, q')\in \delta_{r}^{p_1} \mid  q_{end}^{p_1} \in F^{p_1}\} \cup  \{(q, q_{beg}^{p_1}, \rs, q') \in \delta_{r}^{p_1}, ~q' \notin F^{p_1}\}. 
\end{align*}
The transitions to backtrack $\mathcal{A}_1$'s simulation on return symbols is given in:
\begin{align*}
\delta_{r2} =~&  \{(q_{end}^{p_1}, q_1, \rs, q_{1})  \mid  q_1 \in Q^{1} \setminus F^1, ~q_{end}^{p_1} \in F^{p_1}, ~(q_1, s, q_{final}^1) \notin \delta^1, ~q_{final}^1 \in F^1\}  \\
& \cup\{(q_{rej}, q^1, \rs, q^1)  \mid q^1 \in Q^1 \setminus F^1\} \\
& \cup \{(q, q_1, \rs, q_1)  \mid  q \in (Q^1 \setminus F^1) \cup (Q^{p_1} \setminus F^{p_1}), ~q_1 \in Q^1 \setminus F^1\}  \\
& \cup \{(q^1, q_{beg}, s \rangle, q^1)  \mid  q^1 \in Q^1 \setminus F^1\} 
\end{align*}
The below definition of $\delta_{r3}$ describes the transitions to mark that  a path satisfying the policy is found, meaning the nested word will be eventually accepted by $\mathcal{M}^p$. To mark this, the return transition  goes to the satisfiability marker state $q_{sat}$. The key idea of the below transitions  is that  when the VPA returns from the root of the tree being inspected for all children subtrees to satisfy $p_1$, the VPA should be at the final state of $\mathcal{M}^{p_1}$:
\begin{align*}
\delta_{r3} =~& \{(q_{end}^{p_1}, q_1, \rs, q_{sat})  \mid  q_1 \in Q^{1}-F^1, ~q_{end}^{p_1} \in F^{p_1}, ~(q_1, s, q_{final}^1) \in \delta^1, ~q_{final}^1 \in F^1\}  \\
& \cup  \{(q_{sat}, q_1, \rs, q_{sat})  \mid  q_1 \in (Q^1 \setminus F^1) \cup \Gamma^{p_1} \cup \{q_{sat}, q_{rej}\}\}  \\
& \cup \{(q_{final}^1, q^1, s \rangle, q_{sat}) \mid q^1 \in Q^1-F^1,~q_{final}^1 \in F^1\}
\end{align*}
$\mathcal{M}^p$ accepts the word under any of the following conditions: (a) the tree has a single node and the root of this nested word itself matches the $reg$; (b) this root had children that satisfied the inner policy (c) the VPA is in the $q_{sat}$ state meaning an early acceptance has been indicated. These are described below:
\begin{align*}
\delta_{r4} =~& \cup \{(q_{final}^1, q_{beg}, s \rangle, q_{end})  \mid q_{final}^1 \in F^1\} &\\
& \cup \{(q_{end}^{p_1}, q_{beg}, s \rangle, q_{end})  \mid q_{end}^{p_1} \in F^{p_1}\} &\\
& \cup \{(q_{sat}, q_{beg}, \rs, q_{end}) \}
\end{align*}

Lastly, the following transitions keep the VPA in reject state $q_{rej}$. (It can be read as rules added to make the VPA complete.)
\begin{align*}
\delta_{r5} =~&\{(q_{rej}, q^{p_1}, \rs, q_{rej})  \mid  q^{p_1} \in \Gamma^{p_1} \cup \{q_{beg}, q_{sat}\}\} &\\
& \cup \{(q_{rej}, q_{rej}, \rs, q_{rej})\} &\\
& \cup \{(q^1, q_{sat}, s \rangle, q_{rej})  \mid  q^1 \in Q^1\}&\\
& \cup \{(q^1, q^2, s \rangle, q_{rej})  \mid  q^1 \in Q^1, ~q^2 \in \Gamma^{p_1}\}&\\
& \cup \{(q^2, q, s \rangle, q_{rej})  \mid  q^2 \in Q^{p_1}, q \in \{q_{sat}, q_{rej}\} \}& \\
& \cup \{(q^2, q_{beg}, s \rangle, q_{rej})  \mid  q^2 \in Q^{p_1}-F^{p_1}\}& \\
& \cup \{(q_{beg}, q, \rs, q_{rej}) \mid  q \in \Gamma\}& \\
& \cup \{(q_{end}, q, \rs, q_{rej})  \mid  q\in \Gamma\}& \\
& \cup \{(q, q_{rej}, s \rangle, q_{rej}) \mid  q \in Q^1\} &
\end{align*}

\item VPA for $\match~ reg \existschild p_1 ~\textsf{\textbf{then}}\ldots \textsf{\textbf{then}}~p_n$.

Let the VPA for a policy $p_i$ be $\mathcal{M}^{p_i} = (Q^{p_i},~ q_{init}^{p_i},~ F^{p_i},~ \Sigma,~ \Gamma,~ \bot,~ \delta_{c}^{p_i}, ~ \delta_{r}^{p_i})$. \\
Let the \DFA~ for $reg$~be  $\dfa{1}$.\\
Then $\vp{\match~ reg \existschild p_1 ~\textsf{\textbf{then}}\ldots \textsf{\textbf{then}}~p_n} = \mathcal{M}^{p}$, where
\begin{align*}
&\mathcal{M}^{p} = (Q,~ q_{beg},~ \{q_{end}\},~ \Sigma,~ \Gamma,~ \bot,~ \delta_{c},~  ~ \delta_{r}), \\
&Q = \{q_{beg}, q_{end}, q_{rej}, q_{exists}, q_{sat}\} \cup Q^1 \cup Q_{\text{all-}p_i}, \\
&Q_{\text{all-}p_i} =  \bigcup_{1 \le i \le k} Q^{p_i}, \\
&\Gamma = \{q_{beg}, q_{rej}, q_{exists}, q_{sat}\} \cup (Q^1 \setminus F^1) \cup \Gamma_{\text{all-}p_i},\\
& \Gamma_{\text{all-}p_i} = \bigcup_{1 \le i \le k} \Gamma^{p_i}.
\end{align*}
We define the call transition in terms of $\delta_{c1}$, $\delta_{c2}$, and $\delta_{c3}$ given below. First, the transitions to simulate $\mathcal{A}_1$ followed by running the VPA of $p_1$ are  recorded in $\delta_{c1}$:
\begin{align*}
\delta_{c1}  =~& \{(q_{beg}, \cs, q, q_{beg})  \mid (q_{init}^{1}, s, q) \in \delta_{1}\} &\\
& \cup \{(q_{src}, \cs, q_{dst}, q_{src})  \mid (q_{src}, s, q_{dst}) \in \delta^1, ~q_{src} \notin F^1\} &\\
&\cup \{(q_{final}^1, \cs, q_{dst}, q_{stack})  \mid (q_{beg}^{p_1}, \cs, q_{dst}, q_{stack}) \in \delta_{c}^{p_1},~ q_1^{final} \in F^1\} &
 \end{align*}
The following set $\delta_{\text{calls-}p_{i}}$ contains the call transitions corresponding to the policy $p_i$, but all the call transition in $\mathcal{M}^{p_i}$ from its accept state are updated to give the effect that VPA is at the begin state of $\mathcal{M}^{p_{i+1}}$.
\begin{align*}
\delta_{\text{calls-}p_{i}} =~& \delta_{c}^{p_i} \setminus \{(q_{end}^{p_i}, \cs, q, q_{stack}) \in \delta_{c}^{p_1}  \mid  q_{end}^{p_i} \in F^{p_i}\}&\\
& \cup \{(q_{end}^{p_i}, \cs, q, q_{stack})  \mid   (q_{beg}^{p_{i+1}}, \cs, q, q_{stack}) \in \delta_{c}^{p_{i+1}}, ~q_{end}^{p_i} \in F^{p_i}\}&
\end{align*}
We write the union of all such $\delta_{\text{calls-}p_i}$ for $1 \le i \le k$ as $\delta_{c2} = \bigcup_{1 \le i \le k} \delta_{\text{calls-}p_i}$.

Finally, $\delta_{c3}$ contains the rules to mark if the policy is violated or will definitely be satisfied. The state $q_{exists}$ marks that all $k$ sub-policies have been satisfied.
\begin{align*}
\delta_{c3} =~&  \{(q_{end}^{p_k}, \cs, q_{exists}, q_{exists}) \mid q_{end}^{p_k} \in F^{p_k}\} & \\
& \cup \{(q_{exists}, \cs, q_{exists}, q_{exists})\} & \\
& \cup  \{(q_{sat}, \cs, q_{sat}, q_{sat})\} & \\
& \cup \{(q_{rej}, \cs, q_{rej}, q_{rej})\} &\\
& \cup \{(q_{end}, \cs, q_{rej}, q_{rej})\}
\end{align*}
Putting the above sets together, we define the call transition $\delta_c$ as:\\
$\delta_c = \delta_{c'} \cup \{(q, \cs, q_{rej}, q_{rej}) \mid  q,~ q',~ q'' \in Q,~ \cs \in \Sigma,~(q, \cs,q', q'') \notin \delta_{c'} \}, \text{where}$\\
$\delta_{c'} = \delta_{c1} \cup \delta_{c2} \cup \delta_{c3}$.

The return transition function is defined in terms of $\delta_{r1}$, $\delta_{r2}$, $\delta_{r3}$, and $\delta_{r4}$. The set $\delta_{r1}$ simulates the return transitions of the VPAs of each of the $p_i$ sub-policies. To retry checking $p_i$ on another subtree in case the current subtree did not satisfy $p_i$, we update the destination state of the transition to the beginning state of $p_i$. 
\begin{align*}
\delta_{\text{ret-}p_i} =~& \delta_{r}^{p_i} \setminus \{(q, q_{beg}^{p_i}, s, q') \in \delta^{p_i}_{r} \mid  q' \notin F^{p_i}\} &\\
& \cup \{(q, q_{beg}^{p_i}, s, q_{beg}^{p_i})\vert ~(q, q_{beg}^{p_i}, s, q')\in \delta_{r}^{p_i}, ~q' \notin F^{p_i}\} &
\end{align*}
We write the union of all such $\delta_{\text{ret-}p_i}$ for $1 \le i \le k$ as $\delta_{r1} = \bigcup_{1 \le i \le k} \delta_{\text{ret-}p_i}$. 

The transitions to backtrack $\mathcal{A}_1$ are given by:
\begin{align*}
\delta_{r2} =~&  \{(q, q^1, \rs, q^1)  \mid  q \in (\cup_{1 \le i \le k} Q^{p_i})\setminus F^{p_k}, ~q^1 \in Q^1 \setminus F^1\}&\\
& \cup  \{(q, q^1, \rs, q^1)  \mid  q \in Q^1,~q^1\in Q^1 \setminus F^1\} & \\
& \cup \{(q_{rej}, q^1, \rs, q^1)  \mid  q^1\in Q^1 \setminus F^1\} &\\
& \cup \{(q_{end}^{p_k}, q^1, \rs, q_{1})  \mid  q_{end}^{k} \in F^{p_k}, ~(q^1, s, q^1_{final}) \notin \delta^1,~q^1_{final} \in F^1\} &
\end{align*}
The set $\delta_{r3}$ contains the transition elements that determine whether the VPA has found a path that satisfies the existential condition of its subtrees satisfying the sub-policies, \textit{i.e.,} VPA goes to the $q_{sat}$ state.
\begin{align*}
  \delta_{r3} =~& \{(q_{end}^{p_k}, q^1, \rs, q_{sat})  \mid  q_{end}^{p_k} \in F^{p_k}, ~(q^1, s, q^1_{final}) \in \delta^1,~q^1_{final} \in F^1\} & \\
  & \cup \{(q_{exists}, q^1, \rs, q_{sat})  \mid  (q^1, s, q^1_{final}) \in \delta^1,~q^1_{final} \in F^1\} & \\
  & \cup \{(q_{sat}, q, \rs, q_{sat})  \mid  q \in \Gamma - \{q_{beg}\}\} & \\
  & \cup\Set{(q_{exists}, q, \rs, q_{exists}) \ |  \begin{array}{l} ~q \in \Gamma \setminus (Q^1_{\text{non-final}} \cup \{q_{beg}\}), ~\text{where}\\
   ~Q^1_{\text{non-final}} = Q^1 \setminus F^1 
  \end{array}}&
\end{align*}

The set $\delta_{r4}$ contains the transitions that send the VPA to $q_{rej}$ state, which can either imply that the $k$ sub-policies were not satisfied on the children of a subtree, or simply the special state used to make the VPA complete by sending all disallowed transitions to $q_{rej}$.
\begin{align*}
\delta_{r4} =~& \{(q^1, q_{beg}, \rs, q_{rej})~ \vert~q^1 \in Q^1\}&\\
  & \cup \{(q_{rej}, q, \rs, q_{rej}),  \mid  q \in \Gamma \setminus (Q^1 - F^1)\} &\\
  & \cup  \{(q, q_{rej}, \rs, q_{rej})  \mid q \in Q-\{q_{sat}, q_{exists}\}\} &\\
  & \cup  \{(q_{beg}, q, \rs, q_{rej})  \mid q \in \Gamma\} &\\
  & \cup  \{(q_{end}, q, \rs, q_{rej})  \mid q \in \Gamma\} &\\
  &\cup \{(q_{exists}, q^1, \rs, q_{rej})  \mid  (q^1, s, q^1_{final}) \notin \delta^1,~q^1_{final} \in F^1\} &\\
  &\cup  \{(q^{p_k}, q_{beg}, \rs, q_{rej}) \mid  q^{p_k} \in Q^{p_k} - F^{p_k}\} &\\
    &\cup \{(q^{p_i}, q_{beg}, \rs, q_{rej}) \mid  q^{p_i} \in \cup_{1 \le i < k} Q^{p_i}\} &\\
  & \cup  \delta_{\text{rej-ret-}p_i} &
\end{align*}
Here, $\delta_{\text{rej-ret-}p_i} = \bigcup_{1 \le i \le k} \{(q^{p_i}, q, \rs, q_{rej}) \mid  q^{p_i} \in Q^{p_i}, ~q \in \cup_{i \neq j} \Gamma^{p_j}\cup \{q_{beg}, q_{exists}, q_{sat}\}\}$.

The transitions to the final state of the VPA are given by:
\begin{align*}
\delta_{r5} =~& \{(q_{sat}, q_{beg}, \rs, q_{end})\} &\\
  & \cup \{(q^{p_k}_{end}, q_{beg}, \rs, q_{end}) \mid  q^{p_k}_{end} \in F^{p_k}\}&\\
    & \cup \{(q_{exists}, q_{beg}, \rs, q_{end})\}&
\end{align*}

Putting these sets together, the return transition function is given by:\\
$\delta_r = \delta_{r'} \cup \{(q, q', \rs, q_{rej}) \mid  q,~ q',~ q'' \in Q,~\rs \in \Sigma,~(q, q', \rs, q'') \notin \delta_{r'} \}$, where\\
$\delta_{r'} = \delta_{r1} \cup \delta_{r2} \cup \delta_{r3} \cup \delta_{r4} \cup \delta_{r5}$.

\item VPA for $\callseq ~reg$.

Let the \DFA for $reg$ be $\dfa{1}$.

Then $\vp{\callseq ~reg} = \mathcal{M}_{s}$, where
\begin{align*}
    & \mathcal{M}_{s} = (Q,~ q_{beg},~ \{q_{end}\},~ \Sigma,~ \Gamma,~ \bot,~ \delta_{c},~  ~ \delta_{r}),\\
    & Q = \{q_{beg}, q_{end}, q_{rej}\} \cup Q^1,  \\
    & \Gamma = Q.
\end{align*}
We write the transition function $\delta_c$ and $\delta_r$ as a union of a set of transition elements. We first define $\delta_{c'}$ and $\delta_{r'}$ to encode the simulation of $\mathcal{A}_1$ on call and return symbols, respectively. The simulation of $\mathcal{A}_1$ on call symbols is encoded as:
\begin{align*}
\delta_{c'} \triangleq & \{(q_{beg}, \cs, q, q_{beg})  \mid  (q_{init}^1, s, q) \in \delta^1\}\\
 & \cup \{(q_{src}, \cs, q_{dst}, q_{src})  \mid  (q_{src}, s, q_{dst}) \in \delta^1\}.
\end{align*}
Similarly, the simulation of $\mathcal{A}_1$ on return symbols is encoded as: 
\begin{align*}
\delta_{r'} \triangleq & \{(q^1_{final}, q_{beg}, \rs, q_{end}) \mid  q^1_{final} \in F^1, ~\rs \in \Sigma\}\\
& \cup \{(q', q, \rs, q') \mid  q', q \in Q^1, ~\rs \in \Sigma\}.
\end{align*}
The return transition elements in the above set disregard the symbol and remain at the same state unless it is the last symbol of the trace, which is identified by the stack symbol being $q_{beg}$. Upon reaching the last symbol, the VPA accepts the trace if the string so far is accepted by $\mathcal{A}_1$, which is indicated by being in the source state $q_{final}^1$. 

To make the VPA complete or turn the  sets $\delta_{c'}$ and $\delta_{r'}$ into a function, $\delta_{c'}$ and $\delta_{r'}$ are extended with appropriate transitions to the reject state $q_{rej}$ for the missing elements in $\delta_{c'}$ and $\delta_{r'}$. Finally, the VPA's transition functions are given by:
\begin{align*}
&\delta_c = \delta_{c'} \cup \{(q, \cs, q_{rej}, q_{rej}) \mid  q,~ q',~ q'' \in Q,~ \cs \in \Sigma,~(q, \cs,q', q'') \notin \delta_{c'} \},   \\
& \delta_r = \delta_{r'} \cup \{(q, q', \rs, q_{rej}) \mid  q,~ q',~ q'' \in Q,~\rs \in \Sigma,~(q, q', \rs, q'') \notin \delta_{r'} \}.
\end{align*}

\item VPA for $\start ~S: inner$.

Let the \DFA ~ for the regular expression $\basealpha^*S$ be $\dfa{1}$. \\
Let the VPA for $inner$ be $\mathcal{M}^{in} = (Q^{in},~ q_{beg}^{in},~ F^{in},~ \Sigma,~ \Gamma,~ \bot,~ \delta_{c}^{in}, ~ \delta_{r}^{in})$. \\
Then $\vp{\start ~S: inner} = \mathcal{M}^{p}$, where
\begin{align*}
&\mathcal{M}^{p} = (Q,~ q_{beg},~ \{q_{end}\},~ \Sigma,~ \Gamma,~ \bot,~ \delta_{c}, ~ \delta_{r}), \\
&Q = \{q_{beg}, q_{end}, q_{rej}\} \cup Q^1_{\text{non-final}} \cup Q^{in}_{\text{non-final+beg}}~, \\
&Q^1_{\text{non-final}} = Q^1 \setminus F^1, \\
&Q^{in}_{\text{non-final+beg}} = Q^{in} \setminus (\{q^{in}_{beg}\} \cup F^{in}),\\
&\Gamma = Q \setminus \{q_{end}\}.
\end{align*}
To define the call transition function, we first define the sets $\delta_{c1}$ and $\delta_{c2}$ below. 
$\delta_{c1}$ defined below contains the rules to simulate $\mathcal{A}_1$ for finding a path from the root to a symbol in set $S$ with no ancestors in $S$: 
\begin{align*}
\delta_{c1} \triangleq & \{(q_{beg}, \cs, q, q_{beg})  \mid  (q_{init}^1, s, q) \in \delta^1, ~q \notin F^1\}  \\
& \cup \{(q_{src}, \cs, q_{dst}, q_{src}) \mid  (q_{src}, s, q_{dst}) \in \delta^1, ~q_{src}, q_{dst} \notin F^1\} \\
& \cup \Set{(q_{beg}, \cs, q^{in}, q_{beg})\ | \begin{array}{l}
(q_{init}^1, \cs, q_{final}^1) \in \delta^1, \\
q^1_{final} \in F^1,\\
(q_{beg}^{in}, \cs, q^{in}, q_{beg}^{in}) \in \delta_{c}^{in}
\end{array}}\\
& \cup \Set{(q_{src}, \cs, q^{in}, q_{src})\ | \begin{array}{l}
(q_{src}, \cs, q_{final}^1) \in \delta^1, \\
q_{src} \notin F^1, ~q^1_{final} \in F^1, \\
(q_{beg}^{in}, \cs, q^{in}, q_{beg}^{in}) \in \delta_{c}^{in}\end{array}}
\end{align*}
The last two elements describe a transition to the simulation of $\mathcal{M}_{in}$ after reaching the final state of $\mathcal{A}_1$. Note that we have two special cases: (a) when the source state is $q_{beg}$ meaning the VPA is reading the first symbol of the word and it happens to be a symbol in $S$, and (b) when the path is longer than one symbol.

We define $\delta_{c2}$ to simulate the $\mathcal{M}_{in}$ after the path has been matched by $\mathcal{A}_1$:
\begin{align*}
&\delta_{c2} = \delta_{c}^{in} \setminus \delta_{\text{no-beg+end}}, &\\
& \delta_{\text{no-beg+end}} = \{(q_{beg}^{in}, \cs, q', q_{beg}^{in}) \in \delta^{in}_c\} \cup \{(q_{end}^{in}, \cs, q', q'') \in \delta_{c}^{in} \mid q_{end}^{in} \in F^{in}\}. &
\end{align*}

In the above definition, we have first added the transitions of $\delta_{c}^{in}$. To ensure that after a subtree is accepted by $\mathcal{M}_{in}$, the VPA $\mathcal{M}_p$ restarts simulating $\mathcal{A}_1$ to search for other paths at which $inner$ should be satisfied, we remove all transitions starting at the accept state of $\mathcal{M}_{in}$. Also, we remove the transition from the initial state of $\mathcal{M}_{in}$ because $\mathcal{M}^{p}$ will start simulating $\mathcal{M}_{in}$ after it reaches the final state of $\mathcal{A}_1$. Therefore, the transitions from the initial state of $\mathcal{M}_{in}$ are redundant.

The call transition function for the complete VPA is defined as: \\
$\delta_c = \delta_{c'} \cup \{(q, \cs, q_{rej}, q_{rej}) \mid  q,~ q',~ q'' \in Q,~ \cs \in \Sigma,~(q, \cs,q', q'') \notin \delta_{c'}$, \text{where}\\
$\delta_{c'} \triangleq \delta_{c1} \cup \delta_{c2} \cup \{(q_{rej}, \cs, q_{rej}, q_{rej}) \mid \cs \in \Sigma\}$.\\
Here, all missing transitions of $\delta_{c'}$ get sent to the destination state $q_{rej}$.

To define the return transition function, we first define the sets $\delta_{r1}, ~\delta_{r2}, ~\delta_{r3},~\delta_{r4}$. 
The acceptance conditions are defined as follows:
\begin{align*}
\delta_{r1} = &\{(q^1, q_{beg}, \rs, q_{end}) \mid (q_{init}^1, s, q^1) \in \delta^1\}\\
& \cup \Set{(q^{in}, q_{beg}, \rs, q_{end}) \ | \begin{array}{l}
 (q^{in}, q_{beg}^{in}, \rs, ~q^{in}_{end}) \in \delta_{r}^{in},\\
 ~q_{end}^{in} \in F^{in},\\
 (q_{init}^1, s, q^1_{final}) \in \delta^1\end{array}}.
\end{align*}

The backward simulation of $\mathcal{A}_1$ is defined in:
\begin{align*}
 \delta_{r2} = & \{(q, q^1, \rs, q^1)  \mid   q, q^1 \in Q^1-F^1\} \\
& \cup \Set{(q^{in}, q^1, \rs, q^1) \ | \begin{array}{l} (q^{in}, q_{beg}^{in}, \rs, ~q^{in}_{end}) \in \delta_{r}^{in},\\
q_{end}^{in} \in F^{in}, \\
(q^1, s, q^1_{final}) \in \delta^1, ~q^1 \neq q_{init}^1\end{array}}.
\end{align*}

The set $\delta_{r3}$ contains the transitions of $\mathcal{M}_{in}$: \\
$\delta_{r3} = \delta_{r}^{in} \setminus \delta_{no-end+beg}$, where \\
$\delta_{no-end+beg} = \{(q^{in}_{end}, q', \rs, q'')  \in \delta^{in}_{r} \mid q^{in}_{end} \in F^{in}\} \cup \{(q_{beg}^{in}, q', \rs, q'') \in \delta_{r}^{in}\}$.

Finally, the rejection criterion is captured by $\delta_{r4}$. A word is rejected if a subtree where $inner$ was supposed to be satisfied, does not end in the final state of $\mathcal{M}_{in}$:
\begin{align*}
\delta_{r4} =& \{(q^{in}, q^1, \rs, q_{rej})  \mid  (q^{in}, q_{beg}^{in}, \rs, ~q^{in}_{end}) \notin \delta_{r}^{in},~q_{end}^{in} \in F^{in}, ~(q^1, s, q^1_{final}) \in \delta^1\}\\
 & \cup \{(q_{rej}, q, \rs, q_{rej})  \mid  q \in \Gamma\} 
\end{align*}
Piecing together these sets, the return transition function for the VPA is defined as:\\
$\delta_r = \delta_r' \cup \{(q, q', \rs, q_{rej}) \mid  q,~ q',~ q'' \in Q,~\rs \in \Sigma,~(q, q', \rs, q'') \notin \delta_{r'} \}$, where \\
$\delta_{r'} = \delta_{r1} \cup \delta_{r2} \cup \delta_{r3} \cup \delta_{r4}$.
\end{enumerate}

\section{Soundness Proof}\label{app:soundproof}
\begin{theorem}[Soundness]\label{thm:soundness}
Let $p$ be a policy and  $\mathcal{L}(\vp{p})$ be the set of \rmnws accepted by its visibly pushdown automaton $\vp{p}$. Then,
\begin{enumerate}
    \item $\nw{p} = \mathcal{L}(\vp{p})$ if  $p$ is a service tree policy,
    \item $\tree{p} = \mathcal{L}(\vp{p})$ if $p$ is a hierarchical policy, and
    \item $\Seq{p} = \mathcal{L}(\vp{p})$ if $p$ is a linear sequence policy.   
\end{enumerate}
\end{theorem}

\begin{proof}
   The proof follows directly from (1) \cref{thm:tree}, (2) \cref{thm:seq}, and (3)  \cref{thm:start}.
\end{proof}

\begin{theorem}\label{thm:tree}
Let $p$ be a hierarchical tree policy and $\mathcal{L}(\vp{p})$ be the set of \rmnws accepted by its visibly pushdown automaton $\vp{p}$.  Then, the set of \rmnws accepted by the policy, \textit{i.e.,} $\tree{p}$ is equivalent to those accepted by the visibly pushdown automaton $\vp{p}$, \textit{i.e.,} $\tree{p} = \mathcal{L}(\vp{p})$.
\end{theorem}
\begin{proof}
The proof goes by \cref{lem:tree-dir1} and \cref{lem:tree-dir2}.
\end{proof}

\begin{lemma}\label[lemma]{lem:tree-dir1}
  Let $p$ be a tree policy and  $n$ be a  \rmnw, if $n \in \mathcal{L}(\vp{p})$ then  $n \in \tree{p}$, \textit{i.e.,} $\mathcal{L}(\vp{p}) \subseteq \tree{p}$. Here, $\mathcal{L}(\vp{p})$ is the set of \rmnws accepted by the visibly pushdown automaton $\vp{p}$.
  \end{lemma}
  \begin{proof}
  Given a \rmnw ~$n = (a_1 \dots a_l, \nu) \in \mathcal{L}(\vp{p})$, let the accepting run of 
  \[\vp{p} = \mathcal{M}_{p} = (Q,~ q_{beg},~ \{q_{end}\},~ \Sigma,~ \Gamma,~ \bot,~ \delta_{c}, ~ \delta_{r})\] 
  be 
  \[\rho(\mathcal{M}_p, n) = (q_{beg}, \theta_0 = \bot) \dots (q_l=q_{end}, \theta_l=\bot).\]
  
  An auxiliary result that will be used in this proof is that in the run of a \rmnw, the stack of each configuration has the values pushed by the call symbols with pending returns. Furthermore, the top of the stack symbol at any return symbol is the stack symbol pushed by its matching call.
  
  The proof for the current lemma goes by induction on the structure of the tree policies with the following cases:
  \begin{description}
  \item \textbf{Case} $p = \textsf{\textbf{match}}~reg_1 \allpath reg_{2}$:\\
  By definition of $\vp{p}$, $\mathcal{A}_1$ is the automaton for $reg_1$, where $Q^1$ is the set of all its states, $q_{init}^1$ is its initial state, and $F^1$ is the set of final states. Similarly, $\mathcal{A}_2$ is the automaton for $reg_2$, where  $Q^2$ is the set of all its states, $q_{init}^2$ is its initial state, and $F^2$ is the set of final states.\\
  The run $\rho(\mathcal{M}_{p}, n)$ implies that there exists a transition $(q_{l-1}, \theta_{l-1}) \xrightarrow{a_{l}} (q_{end}, \bot)$. By definition of \rmnw, $a_l$ is a return symbol and $(q_{l-1}, q_{stack}, a_{l}, q_{end}) \in \delta_r$ for some $q_{stack}$. There are three cases:  
  \begin{enumerate}
  \item  \label{same}  Case $(q_{l-1}, q_{stack}, a_{l}, q_{end}) = (q_{final}^1, q_{beg}, s \rangle, q_{end}), ~\text{where}~q_{final}^1 \in F^1:$\\
  In this case $q_{l-1}$ is some final state of $reg_1$ DFA and $q_{stack} = q_{beg}$. Now, $q_{beg}$ is pushed on the stack only by the rule $(q_{beg}, \langle s, q, q_{beg})$, where $q$ satisfies $(q_{init}^1, s, q) \in \delta^1$. 
  The stack entry $q_{beg}$ will be pushed by the matched call for the return symbol $a_l$, \textit{i.e.,} is $a_1$.\\
  The final state of $reg_1$ DFA can only be a destination state upon reading a call symbol. Therefore, $a_{l-1}$ will be a call symbol. By definition of \rmnw, $a_{l-1}$ will be the matched call for $a_{l}$. Since we know $a_1$ is the matched call for $a_l$, the nested word will be $n = (a_1 a_2, \nu)$. 
This word has only one path, and that path $a_1 = \cs$ satisfies the first match property because $q_1 =q_{final}^1$ implies there exists a transition  $(q_{init}^1, \cs, q_{final}^1) \in \delta^1$. Since the grammar of our language disallows $\epsilon \in \mathcal{L}(reg_1)$, $FirstMatch(reg_1)$ will not contain empty string. Lastly, this nested word $n = (a_1 a_2, \nu)$ doesn't have any subtrees for the nested word rooted at $a_1$, so vacuously all paths satisfy $reg_2$. 
  \item Case $(q_{l-1}, q_{stack}, a_{l}, q_{end}) = (q_{sat}, q_{beg}, s \rangle, q_{end})$:\\
  Based on the rules, there will be some return symbol $a_j \in n$ upon reading which the destination state $q_j = q_{sat}$ shows up for the first time in \VPA run. 
    Being in $q_{sat}$ is sufficient to eventually accept the nested word. It marks the existence of a path matching $reg_1$ that satisfies the \textbf{\textsf{all-paths}} constraints. The transition on $a_{j}$ could have observed the application of one of the following rules, where $q_{s}$ refers to some stack value:
  \begin{enumerate}
  \item Case $(q_{j-1}, q_{s}, \rs, q_{sat}) = (\textsf{aug}(q^2_{init}), q^1, s \rangle, q_{sat})~\text{such that}~(q^1, s, q_{final}^1) \in \delta^1$:\\
  Here, $q^1 \in Q^1-F^1,~q_{final}^1 \in F^1$. 
  In this case, the top of the stack is some non-accepting state of $reg_1$ DFA that transitions to the final state of $reg_1$ DFA after reading $s$. Let the matched call for $a_j$ be the symbol $a_i$. This top of the stack symbol $q^1$ would have been pushed while transitioning on $a_i$. Since $(q^1, s, q_{final}^1) \in \delta^1$, we have $q_{i} = q_{final}^1$. Consider $\alpha = \pi(n, a_i) \in Seq(n)$. By induction on the length of prefixes of the word $Calls(\alpha)$, we can prove that the final state in the run of the  $reg_1$ \DFA~ on $Calls(\alpha)$ is same as $q^i$. Therefore, $Calls(\alpha)$ is accepted by $reg_1$ because $q_{i} = q_{final}^1$.\\
  Now, we will show that in the nested word $n' = n[\idx(a_i), \idx(a_j)]$ with $m$ subtrees, any path in any subtree in  $Subtree(n')$ is matched by $reg_2$.
  Note that $a_{j-1}$ must be a return symbol because no call transition has destination $\textsf{aug}(q_{init}^2)$, which is the stack value that's pushed when the source state of a transition is the initial state of $reg_2$ DFA. 
  In particular, this return transition on $a_{j-1}$ could have been $(q, q', s \rangle, q')$, where $q \in F^2$ and $q' \in \tilde{Q^2}$ or $(q, q', \rs, q')$, where $q, q' \in \tilde{Q^2}$. In both cases, $q' = q_{j-1} = aug(q_{init}^2)$. Therefore the top of stack $\theta_{j-2}$ will be $\textsf{aug}(q_{init}^2)$. 
  Let $a_{j'}$ be the matched call for $a_{j-1}$, \textit{i.e.,} $\nu(\idx(a_{j'})) = \idx(a_{j-1})$.
  Then top of stack $\theta_{j'}$ is $\textsf{aug}(q_{init}^2)$. Now, the \rmnw ~$n[\idx(a_{j'}), \idx(a_{j-1})]$ is a   subtree in $Subtree(n')$. If the subtrees in $Subtree(n')$ are ordered by the index of the first symbol of the corresponding nested word then $n[\idx(a_{j'}), \idx(a_{j-1})]$ is the last subtree. This can be proved by case analysis on the transition rule used on the call symbol $a_{j'}$. \\
  Next, given a subtree in  $Subtree(n')$, and some path $\beta = \pi(n_s, a_{k}) \in Seq(n')$, it can be shown by induction on the length of prefixes of the word $Calls(\beta)$ that the final state in the run of the  $reg_2$ \DFA~ on $Calls(\beta)$ is same as the destination state in the VPA run after reading $a_k$.\\
  Finally, we will show that in the nested word $n'$ with  $m$ subtrees, any path in any subtree in  $Subtree(n')$ is matched by $reg_2$. Since $q_j = q_{sat}$, the state $q_{rej}^2$ corresponding to a path not being matched by $reg_2$ could not have been pushed on the stack while reading symbols between $a_i \ldots a_j$. This implies that the rule $(q', q, s \rangle, q_{rej}^2)$, where  $q' \in Q^2 - F^2, ~q \in \tilde{Q^2}$ would not have been used while reading symbols between $a_i \ldots a_j$. Therefore, there exists no leaf symbol or $a_k \in a_i \ldots a_j$ such that $\nu(\idx(a_k)) + 1 = \idx(a_{k+1})$, bit $q_k$ is not a final state of $reg_2$ DFA. Therefore, if we consider any $\pi(n_i, a_k) \in SeqLeaf(n_i)$ then $q_k \in F^2$, which means $Calls(\pi(n', a_k))$ is accepted by $reg_2$. We conclude that all path in any subtree of $n_i$ are matched by $reg_2$.
  
  \item Case $(q_{j-1}, q_{s}, \rs, q_{sat}) = (q_{final}^1, q^1, s \rangle, q_{sat})$:\\
  Here, $q^1 \in Q^1-F^1,~q_{final}^1 \in F^1$. If the transition on $a_{j}$ used the rule $(q_{final}^1, q^1, s \rangle, q_{sat})$ then $a_{j-1}$ would have been a call symbol. Further, $a_{j-1}$ is the matching call for $a_j$. Similar to the above case, we can show that the call sequence or path $\pi(n, a_{j-1})$ satisfies $Calls(\pi(n, a_{j-1})) \in \mathcal{L}(reg_1)$. The nested word $n[\idx(a_{j-1}), \idx(a_{j})]$ has no subtrees. It is vacuously true that all paths satisfy $reg_2$.
  \end{enumerate}
  \item Case $(q_{l-1}, q_{stack}, a_{l}, q_{end}) = (\textsf{aug}(q^2_{init}), q_{beg}, s \rangle, q_{end})$:\\
  The reasoning here is similar to the above case.
  \end{enumerate}
  
  \item \textbf{Case} $p = \match~ reg \allchildren p'$:\\
    By definition of $\vp{p}$, $\mathcal{A}_1$ is the automaton for $reg$ and  $Q^1$ is the set of all its states, $q_{init}^1$ is its initial state, and $F^1$ is the set of final states. $F^{p_1}$ is the set of final states for $\vp{p'}$ and $q_{beg}^{p_1}$ is its initial state. \\
  Similar to the case for $p = \textsf{\textbf{match}}~reg_1 \allpath reg_{2}$, the two transition rules that could have been applied on reading $a_l$ can be:
  \begin{itemize}
  \item Case $(q_{l-1}, q_{stack}, a_{l}, q_{end}) = (q_{final}^1, q_{beg}, s \rangle, q_{end}), ~\text{where}~q_{final}^1 \in F^1:$\\
  The reasoning for this case is same as that in \cref{same}.
  
  \item Case $(q_{l-1}, q_{stack}, a_{l}, q_{end}) = (q_{sat}, q_{beg}, s \rangle, q_{end})$:\\
  State $q_{sat}$ corresponds to the policy being satisfied by existence of a path matching $reg_1$, where all subtrees satisfy $p'$. Based on the rules, there is some $a_j \in n$ upon reading which the destination state $q_j = q_{sat}$ shows up for the first time in \VPA run. The transition on $a_{j}$ can use one of the two rules, where $q_{s}$ refers to some stack value:
  \begin{enumerate}
  \item Case $(q_{j-1}, q_{s}, \rs, q_{sat}) = (q_{end}^{p'}, q_1, \rs, q_{sat})$:\\
  Here, $q_1 \in Q^1-F^1, ~q_{end}^{p'} \in F^{p'}, ~(q_1, s, q_{final}^1) \in \delta^1, ~q_{final}^1 \in F^1$. This implies that the $q_{j-1}$ is the final state of $\vp{p'}$. The symbol $a_{j-1}$ is a return symbol as only return symbols can transition to the final state of $\vp{p'}$. If we consider $a_{j'}$ to  be the matched call for $a_{j-1}$ then top of the stack $\theta_{j'}$ will be the initial state of $\vp{p'}$ or $q_{beg}^{p'}$. This can be shown by case analysis on the policy $p'$ and inspecting the return transition with destination being an the accepting state of $\vp{p'}$. \\
  Suppose $a_i$ is the matched call for $a_j$ and $n_i = n[\idx(a_i), \idx(a_j)]$. The \rmnw ~$n[\idx(a_{j'}), \idx(a_{j-1})]$ is a   subtree in $Subtree(n')$. If the subtrees in $Subtree(n')$ are ordered by the index of the first symbol of the corresponding nested word then $n[\idx(a_{j'}), \idx(a_{j-1})]$ is the last subtree. This can be proved by case analysis on the transition rule used on the call symbol $a_{j'}$. \\
  Now we focus on the state $q_{j'}$ and the state $q_{j'-1}$.  
  The following two transition rules can be applied on $a_{j'}$ to get $q_{beg}^{p'}$ pushed on the stack:
  \begin{itemize}
  \item Case $(q_{j'-1}, \cs, q_{dst}, q_{beg}^{p'}) = (q^1_{final}, \cs , q_{dst}, q_{beg}^{p'})$:\\
  Here, $(q_{beg}^{p1}, \cs , q_{dst}, q_{beg}^{p'}) \in ~ \delta^{p'}_{c}, ~ q_{final}^1 \in F^1$.
  \rules this rule can be applied only if $a_{j'} = a_{i+1}$ because the final state of the the DFA of $reg$ can only occur as a destination state if $a_{j'-1}$ is a call symbol. 
  This implies that $n_i$ has only one subtree $n_s = n_i[\idx(a_{j'}), \idx(a_{j-1})]$.\\
  The run of \VPA  for the policy $p'$ on this \rmnw ~ $n_s$ is accepted because the execution ends at $q_{end}^{p'}$ and also it started at the state $q_{final}^1$, which is equivalent to $q_{beg}^{p'}$ in terms of call transitions. By induction hypothesis of \cref{lem:tree-dir1}, if $n_s$ is accepted by $\vp{p'}$ then $n_s \in \nw{p'}$.\\
  Consider $\alpha = \pi(n, a_i) \in Seq(n)$, it can be shown by induction on the length of prefixes of the word $Calls(x)$ that the final state in the run of the  $reg_1$ \DFA~ on $Call(\alpha)$ is same as $q_i$. Therefore, $Calls(\alpha)$ is accepted by $reg$ because $q_{i} = q_{final}^1$.   
  
  \item Case $(q_{j'-1}, \cs, q_{dst}, q_{beg}^{p'}) = (q_{end}^{p'}, \langle s, q, q_{stack})$:\\
  Here, $(q_{beg}^{p'}, \langle s, q, q_{stack}) \in \delta_{c}^{p'}, ~ q_{end}^{p'} \in F^{p'}$. 
  It is easy to show that  $a_{j'-1}$ is a return symbol. Again, it can be inductively shown that this is a subtree of $n_i$ as defined in the above case.
  Similar to the above case, we can show that all $m$ subtrees in $Subtree(n_i)$ will be in $\nw{p'}$ and that $Calls(\pi(n, a_i))$ is accepted by DFA of $reg$.
      
  \end{itemize}
  
  \item Case $(q_{j-1}, q_{s}, \rs, q_{sat}) = (q_{final}^1, q^1, s \rangle, q_{sat})$:\\
  Here, $q^1 \in Q^1-F^1,~q_{final}^1 \in F^1$. If the transition on $a_{j}$ used the rule $(q_{final}^1, q^1, s \rangle, q_{sat})$ then $a_{j-1}$ would have been a call symbol. Further, $a_{j-1}$ is the matching call for $a_j$. Similar to the above case, we can show that the call sequence or path $\pi(n, a_{j-1})$ satisfies $Calls(\pi(n, a_{j-1})) \in \mathcal{L}(reg_1)$. The nested word $n[\idx(a_{j-1}), \idx(a_{j})]$ has no subtrees. It is vacuously true that all paths satisfy $reg_2$.
  \end{enumerate}
  \end{itemize}
  % HERE
  \item \textbf{Case} $p = \match~ reg \existschild p_1 ~\textsf{\textbf{then}}\ldots \textsf{\textbf{then}}~p_k$:\\
    By definition of $\vp{p}$, $\mathcal{A}_1$ is the automaton for $reg$ and  $Q^1$ is the set of all its states, $q_{init}^1$ is its initial state, and $F^1$ is the set of final states. $F^{p_i}$ is the set of final states for $\vp{p_i}$ and $q_{beg}^{p_i}$ is its initial state. \\
  Similar to the cases for the above policies, the two transition rules that could have been applied on reading $a_l$ can be:
  \begin{itemize}
  \item Case $(q_{l-1}, q_{stack}, a_{l}, q_{end}) = (q_{end}^{p_k}, q_{beg}, s \rangle, q_{end}), ~\text{where}~q_{end}^{p_k} \in F^{p_k}:$\\
  As per this rule, $q_{l-1} = q_{end}^{p_k}$ is some final state of $\vp{p_k}$. We know that $a_{l-1}$ is a return symbol because it cannot be the matched call symbol for $a_{l}$. This is because there is no  transition from source $q_{beg}$ to any state of $\vp{p_k}$.\\
  Suppose $a_{j'}$ is the matched call for $a_{l-1}$ then $n_{j'} = n[\idx(a_{j'}), \idx(a_{l-1}))]$ is a subtree of $n$. 
  Based on the rules, $q_{beg}^{p_k}$ could have been pushed on reading $a_{j'}$ if $q_{j'-1} \in \{q_{final}^1, q_{end}^{p_{k-1}}, q_{beg}^{p_k}\}$, meaning it is the final state of the DFA for $reg$ or final state of the $\vp{p_{k-1}}$ or initial state of $\vp{p_k}$. The run of $\vp{p_k}$ on $n_{j'}$ is same as the run of $\vp{p}$ on $n$ between the configuration upon reading $a_{j'} \ldots a_{l-1}$. By inductive hypothesis of \cref{lem:tree-dir1}, if $n_{j'}$ is accepted by $\vp{p_k}$ then $n_{j'} \in \tree{p_k}$. \\
  Next, we can show by induction on the position $k > 1$ that existence of a subtree $n_s \in Subtree(n)$ such that $n_s \in \vp{p_k}$ implies the existence of a set of $k-1$ subtrees preceding it (in the set of subtrees ordered by the index of the symbols in the subtree) such that $m^th$ subtree is accepted by $\vp{p_{m}}$,

  \item Case $(q_{l-1}, q_{stack}, a_{l}, q_{end}) = (q_{exists}, q_{beg}, s \rangle, q_{end})$\\
  Proof is similar to the above case with additional detail that $q_{l-1} =q_{exists}$ implies the existence of a subtree $n[\idx(a_{j'}),  \nu(\idx(a_{j'}))]$ that satisfies $p_k$.

  \item Case $(q_{l-1}, q_{stack}, a_{l}, q_{end}) = (q_{sat}, q_{beg}, s \rangle, q_{end})$:\\
  Based on the rules, there will be some return symbol $a_j \in n$ upon reading which the destination state $q_j = q_{sat}$ shows up for the first time in \VPA run. The transition on $a_{j}$ can use one of the two rules, where $q_{s}$ refers to some stack value:
  \begin{enumerate}
  \item Case $(q_{j-1}, q_{s}, \rs, q_{sat}) = (q_{end}^{p_k}, q_1, \rs, q_{sat})$:\\
  Here, $q_{end}^{p_k} \in F^{p_k}, ~(q_1, s, q_{final}^1) \in \delta^1, ~q_{final}^1 \in F^1$. This implies that the $q_{j-1}=q_{end}^{p_k}$. Also, $a_{j-1}$ is a return symbol because a transition to the final state of $\vp{p_k}$ is only possible on return symbols. If we consider $a_{j'}$ to  be the matched call for $a_{j-1}$ then top of the stack $\theta_{j'}$ is $q_{beg}^{p_k}$ or the initial state of $\vp{p_k}$.\\
  Suppose $a_i$ is the matched call for $a_j$ and $n' = n[\idx(a_i), \idx(a_j)]$. Let $a_{j'}$ be the matched call for $a_{j-1}$. As shown in the above cases, the \rmnw ~$n[\idx(a_{j'}), \idx(a_{j-1})]$ is the  subtree in $Subtree(n')$. Similar to the above case, we can show that there exists a set of $k$ subtrees in $Subtree(n')$ which when ordered by the index of their first elements satisfy $p_1, \ldots p_k$ in the order of increasing indices of the subtree. Also, $q_{i}$ would have been the final state $q_{final}^1$  of $reg$ DFA. As shown previously. the path $\pi(n, a_i)$ will satisfy $Calls(\pi(n, a_i)) \in \mathcal{L}(reg)$.

  \item Case $(q_{j-1}, q_{s}, \rs, q_{sat}) = (q_{exists}, q^1, s \rangle, q_{sat})$:\\
  Here, $(q_1, s, q_{final}^1) \in \delta^1, ~q_{final}^1 \in F^1$. Suppose $a_i$ is the matched call for $a_j$ and $n' = n[\idx(a_i), \idx(a_j)]$.  
  The proof is similar to the above case with the additional detail that here $q_{j-1} = q_{exists}$ implies the existence of $k$ subtrees in $Subtree(n')$ which when ordered by the index of their first elements satisfy $p_1, \ldots p_k$ in the order of increasing indices of the subtree. 
  \end{enumerate}
  \end{itemize}
  \end{description}
  \end{proof}
  
\begin{lemma}\label[lemma]{lem:tree-dir2}
  Given a hierarchical tree policy $p$ and a \rmnw $n$, if  $n \in \tree{p}$ then $n \in \mathcal{L}(\vp{p})$, \textit{i.e.,} $\tree{p}  \subseteq  \mathcal{L}(\vp{p})$. Here, $\mathcal{L}(\vp{p})$ is the set of \rmnws accepted by the visibly pushdown automaton $\vp{p}$.
  \end{lemma}
  \begin{proof}
  Given a \rmnw ~$n = (a_1 \dots a_l  \nu) \in \mathcal{L}(\vp{p})$, let the accepting run of  
  \[\vp{p} = \mathcal{M}_{p} = (Q,~ q_{beg},~ \{q_{end}\},~ \Sigma,~ \Gamma,~ \bot,~ \delta_{c}, ~ \delta_{r})\] 
  be 
  \[\rho(\mathcal{M}_p, n) = (q_{beg}, \theta_0 = \bot) \dots (q_l=q_{end}, \theta_l=\bot).\]
  
  The proof goes by induction on the structure of the tree policies. The cases in the inductive proof will be:
  \begin{enumerate}
  \item Case $p = \textsf{\textbf{match}}~reg_1 \allpath reg_{2}$:\\
 Suppose $\mathcal{A}_1$ is the automaton for $reg_1$, where $Q^1$ is the set of all its states, $q_{init}^1$ is its initial state, and $F^1$ is the set of final states. Similarly, $\mathcal{A}_2$ is the automaton for $reg_2$, where $Q^2$ is the set of all its states, $q_{init}^2$ is its initial state, and $F^2$ is the set of final states.\\
  Existence of the path $\pi(n, a_i) = a_1 \ldots a_i$ implies that the configuration state $q_i$ after reading $a_i$ will be a final state of $reg_1$ \DFA. 
  Suppose $\pi(n, a_i)$ is the first path such that $Calls(\pi(n, a_i)) \in \mathcal{L}(reg_1)$ and also any path in any subtree in $n[\idx(a_i), \nu(\idx(a_i))]$ matches $reg_2$. We need to show that the above conditions are sufficient for the $\vp{p}$ to accept $n$.\\
  Let us consider possible cases for the index of symbol $a_i$ and if it has subtrees or not:
  \begin{enumerate}
  \item \label{pt1} $\idx(a_i) = 1$ and $a_i$ has no subtrees, \textit{i.e.,} $n = (a_1 a_2  \nu)$, where $a_i = a_1$:\\
  In this case, $\pi(n, a_i) = a_1$ because there is just one call symbol in $n$. Since $Calls(\pi(n, a_i))$ is accepted by $reg_1$ DFA, the final state in the run of the DFA of $reg_1$ on $Calls(\pi(n, a_i))$ will be a final state $q_{final}^1 \in F^1$. Just like in the proof of theorem \cref{lem:tree-dir1}, we can show that $q_1 = q_{final}^1$.  Based on the transition function, the run of $\vp{p}$ on $n$ will be an accepting run, $(q_{beg}, \bot) (q_{final}^1, \bot q_{beg}) (q_{end}, \bot)$.
  
  \item  $\idx(a_i) = 1$ and $a_i$ has a subtree, \textit{i.e.,} $n = (a_1 a_2 \ldots a_{l-1} a_l,  \idx, \nu)$:\\
  Similar to the above case $q_1 = q_{final}^1 \in F^1$. Consider there are $m$ subtrees in $Subtree(n)$. For any subtree $n[\idx(a_j), \nu(\idx(a_j))] \in Subtree(n)$, $q_{j-1} \in \{aug(q_{init}^2)\} \cup F^1$. This can be shown by induction on the position of the subtrees (if ordered by the index of their first symbols). Finally, $q_{l-1} = aug(q_{init}^2)$. Since top of $\theta_{l-1}$ is $q_{beg}$, the transition on $a_l$ will yield $q_l = q_{end}$. Therefore, this is accepted by $\vp{p}$.
  
  \item \label{pt3} $\idx(a_i) > 1$, \textit{i.e.,} $n = (a_1 \ldots a_i \ldots a_{\nu(\idx(a_i))} \ldots a_l,  \idx, \nu)$:\\
  \rules having a configuration with state $q_{sat}$ is sufficient to be accepted by the $\vp{p}$. Requiring $Calls(\pi(n, a_i)) \in FirstMatch(n, reg_1)$ is  necessary for being accepted by $\vp{p}$. Therefore, $q_{sat}$ can not occur in any configurations prior to $q_i$.   As $Calls(\pi(n, a_i))$ is accepted by \DFA~ for $reg_1$, the final state will be some $q_{final}^1 \in F^1$. This implies $q_{i} = q_{final}^1$.\\
  By induction on the number of subtrees $m \ge 0$ between $a_i$ and $a_{\nu(\idx(a_i))}$, if all paths in any subtree are matched by $reg_2$ then the state $q_{\nu(\idx(a_i))} = q_{sat}$.
  \end{enumerate}
  
  Note that we have skipped the details about why $q_{sat}$ in any configuration is sufficient to be accepted by the VPA.
  \item Case $p = \match~ reg \allchildren p'$:
Suppose $\mathcal{A}_1$ is the automaton for $reg$ and  $Q^1$ is the set of all its states, $q_{init}^1$ is its initial state, and $F^1$ is the set of final states. $F^{p_1}$ is the set of final states for $\vp{p'}$ and $q_{beg}^{p_1}$ is its initial state.  \\
  Existence of the path $\pi(n, a_i) = a_1 \ldots a_i$ implies that the configuration state $q_i$ after reading $a_i$ will be a final state of $reg_1$ \DFA. \\
  Suppose $\pi(n, a_i)$ is the first path such that $Calls(\pi(n, a_i)) \in FirstMatch(n, reg)$ and also any subtree in $Subtree(n[\idx(a_i), \nu(\idx(a_i))])$ satisfies $p'$ or is in $\tree{p'}$. We need to show that the above conditions are sufficient for the $\vp{p}$ to accept $n$.
  
  Let's consider possible cases for the index of symbol $a_i$ and if it has subtrees or not:
  \begin{enumerate}
  \item  $\idx(a_i) = 1$ and $a_i$ has no subtrees, \textit{i.e.,} $n = (a_1 a_2 , \nu)$, where $a_i = a_1$:\\
  Proof is similar to the above case \ref{pt1}.
  
  \item  $\idx(a_i) = 1$ and $a_i$ has a subtree, \textit{i.e.,} $n = (a_1 a_2 \ldots a_{l-1} a_l,  \idx, \nu)$:\\
  Similar to the above case $q_1 = q_{final}^1 \in F^1$. Consider there are $m$ subtrees in $Subtree(n)$. Given any subtree in $n_s \in Subtree(n[\idx(a_i), \nu(\idx(a_i))])$ satisfies $n_s \in \tree{p'}$, by induction $n_s \in \vp{p'}$. In any subtree $n_s = n[\idx(a_j), \nu(\idx(a_j))]$ the start state before parsing the first symbol of the subtree is  either the initial state of $reg$ DFA or the final state of the $\vp{p_1}$. This can be shown by induction on the position of the subtrees (if ordered by the index of their first symbols). The transition function of $\vp{p}$ is such that the call transition initial state of $reg$ DFA or the final state of the $\vp{p'}$ is equivalent to transitions being made from the initial state of $\vp{p'}$. Moreover, the state after reading the last symbol of the subtree is $q_{\nu(\idx(a_j))} \in F^{p_1}$. Since $n_s \in \mathcal{L}(\vp{p'})$ if $n_s$ is parsed starting the initial state of $p'$ or an equivalent state, the last state will be the final state of $p'$. Therefore, it can be shown that eventually $q_{\nu(\idx(a_i))-1} \in F^{p_1}$. Since $\theta_{l-1} = \bot q_{beg}$, $q_{l} = q_{end}$.
  
  \item $\idx(a_i) > 1$, \textit{i.e.,} $n = (a_1 \ldots a_i \ldots a_{\nu(\idx(a_i))} \ldots a_l,  \idx, \nu)$\\
Proof is similar to the above case \ref{pt3}.

  \end{enumerate}
  
  \item Case $p = \textsf{\textbf{match}}~ reg \existschild p_1 ~\textsf{\textbf{then}}\ldots \textsf{\textbf{then}}~p_k$:\\
  Suppose $\mathcal{A}_1$ is the automaton for $reg$ and  $Q^1$ is the set of all its states, $q_{init}^1$ is its initial state, and $F^1$ is the set of final states. $F^{p_i}$ is the set of final states for $\vp{p_i}$ and $q_{beg}^{p_i}$ is its initial state. 
  Existence of the path $\pi(n, a_i) = a_1 \ldots a_i$ implies that configuration state after reading $a_i$ will be a final state of $reg$ \DFA. 
  Suppose $\pi(n, a_i)$ is the first path such that $Calls(\pi(n, a_i)) \in FirstMatch(n, reg)$ and also there exists a set of $k \ge 1$ subtrees $Subtree(n[\idx(a_i), \nu(\idx(a_i))])$ that satisfy sub-policies $p_i$ for all $k$ in the order of the index of their respective first symbols. We need to show that the above conditions are sufficient for the \VPA to accept $n$. Note $Subtree(n[\idx(a_i), \nu(\idx(a_i))])$ cannot be empty.
  
  Let's consider possible cases for the index of symbol $a_i$:
  \begin{enumerate}
  \item  $\idx(a_i) = 1$ and $a_i$ has a subtree, \textit{i.e.,} $n = (a_1 a_2 \ldots a_{l-1} a_l,  \idx, \nu)$:\\
  Here, $q_1 = q_{final}^1 \in F^1$. Consider there are $m$ subtrees in $Subtree(n)$. Let $\{n_1, \ldots, n_k\}$ be the set of $k$ subtrees of $n$ such that $n_j \in \tree{p_j}$. By induction hypothesis of this theorem, $n_j \in \vp{p_j}$. We can show by induction on the position of the subtrees (if ordered by the index of their first symbols) that if subtree $n_j = n[\idx(a_m), \nu(\idx(a_m))]$ satisfies  policy $p_j$ then in the run of $\vp{p}$ on $n$, $q_{\nu(\idx(a_m))} \in F^{p_j}$. Using this, we can show that $q_{l-1} \in \{q_{exists}\} \cup F^{p_k}$, and $q_{l} = q_{end}$.

  \item $\idx(a_i) > 1$, and $a_i$ has subtrees, \textit{i.e.,} $n = (a_1 \ldots a_i \ldots a_{\nu(\idx(a_i))} \ldots a_l,  \idx, \nu)$, where $\idx(a_i) + 1 \neq \nu(\idx(a_i))$:\\
  \rules having a configuration with state $q_{sat}$ is sufficient to be accepted by the \VPA. Requiring $Calls(\pi(n, a_i)) \in FirstMatch(n, reg)$ is  necessary for getting a configuration with state $q_{sat}$. Therefore, $q_{sat}$ cannot occur before  $q_i$.  
  As $Calls(\pi(n, a_i))$ is accepted by \DFA~ for $reg$, the final state will be some $q_{final}^1 \in F^1$. Similar to the above case, $q_{i} = q_{final}^1$. \rules, $q_{\nu(\idx(a_i))} = q_{sat}$, and eventually $q_{l}$ will be $q_{end}$. Therefore, $n$ is accepted by the \VPA.
  \end{enumerate}
  \end{enumerate}
  \end{proof}

\begin{theorem}\label{thm:seq}
Let $s$ be a sequence policy and $\mathcal{L}(\vp{s})$ be the set of \rmnws accepted by its visibly pushdown automaton $\vp{s}$. Then, the set of \rmnws accepted by the policy, \textit{i.e.,} $\Seq{s}$ is equivalent those accepted by the visibly pushdown automaton $\vp{s}$, \textit{i.e.,} $\Seq{s} = \mathcal{L}(\vp{s})$.
\end{theorem}
\begin{proof}
The proof goes by \cref{lem:seq-dir1} and \cref{lem:seq-dir2}.
\end{proof}

\begin{lemma}\label[lemma]{lem:seq-dir1}
Let $s$ be a sequence policy of the form $\callseq ~reg$ and $n$ be a \rmnw. If $n \in \mathcal{L}(\vp{s})$ then $n \in \Seq{s}$, \textit{i.e.,} $\mathcal{L}(\vp{s}) \subseteq \Seq{s}$. Here, $\mathcal{L}(\vp{s})$ is the set of \rmnws accepted by the visibly pushdown automaton $\vp{s}$.
\end{lemma}
\begin{proof}
Given a \rmnw ~$n = (a_1 \dots a_l , \nu)$, let the accepting run of the automaton 
\[\vp{s} = \mathcal{M}_{s} = (Q,~ q_{beg},~ \{q_{end}\},~ \Sigma,~ \Gamma,~ \bot,~ \delta_{c}, ~ \delta_{r})\] 
be 
\[\rho(\mathcal{M}_s, n) = (q_{beg}, \theta_0 = \bot) \dots (q_{end}, \bot).\]

This implies $q_{l-1}$ is some final state of $reg$'s DFA. By induction on $i$, where $1 \le i \le l$, we can show that the state $q_i$ in $\vp{s}$'s run on the prefix $a_1 \ldots a_i$ is same as the final state of the run of $reg$'s DFA on the sequence of just the call symbols between $a_1 \ldots a_i$. Therefore, $q_{l-1}$ being the accepting state of $reg$ DFA implies the sequence of calls in $n$ is accepted by $reg$'s DFA.
%Let us consider any arbitrary sumbol $a_i$ such that $q_i \in F^1$ then path $Calls(\pi(n, a_i))$ would be accepted by DFA of $reg_1$. After reading $a_{\nu(\idx(a_i))}$, the configuration will be some non $q_{rej}$ state only if $q_{\nu(\idx(a_i))-1} \in F^2$. We can prove by induction that the final state after reading some subtree $n_i = n[\idx(a_i), \nu(\idx(a_i))]$ is same as the final state upon running DFA for $reg_2$ on $Calls(a_{i}\ldots a_{\nu(\idx(a_i))}))$. This shows that any path such that it is the first match of $reg_1$ then its subtree call sequence is matched by $reg_2$.
%
%If $\nu(\idx(a_i)) = \idx(a_i) +1$ then it is required that $reg_2$ accepts $\epsilon$.\\

%first symbol reg1. let a new call, and reg1 then 
\end{proof}
\begin{lemma}\label[lemma]{lem:seq-dir2}
Let $s$ be a sequence policy of the form $\callseq ~reg$ and $n$ be a \rmnw. If $n \in \Seq{s}$ then $n \in \mathcal{L}(\vp{s})$, \textit{i.e.,} $\Seq{s} \subseteq \mathcal{L}(\vp{s})$. Here, $\mathcal{L}(\vp{s})$ is the set of \rmnws accepted by the visibly pushdown automaton $\vp{s}$.
\end{lemma}
\begin{proof}
Given a \rmnw ~$n = (a_1 \dots a_l , \nu)$, if the sequence of all call symbols is accepted by $reg$'s DFA then after reading the last call symbol, say $a_i$, the state $q_{i} \in F^1$. Based on the transition function, return symbols preserve the source state. This implies, the state $q_{l-1}$ will be the final state of $reg$'s DFA (meaning the VPA will accept the word).
\end{proof}

\begin{theorem}\label{thm:start}
Let $p$ be a policy of the form $\start S: inner$ and 
$\mathcal{L}(\vp{p})$ be the set of \rmnws accepted by its visibly pushdown automaton $\vp{p}$. Then, 
the set of \rmnws accepted by the policy, \textit{i.e.,} $\Policy{p}$ is equivalent those accepted by the visibly pushdown automaton $\vp{p}$, \textit{i.e.,} 
$\Policy{p} = \mathcal{L}(\vp{p})$.
\end{theorem}
\begin{proof}
    The proof goes by \cref{lem:start-dir1} and \cref{lem:start-dir2}.
\end{proof}

\begin{lemma}\label[lemma]{lem:start-dir1}
Let $p$ be a policy of the form $\start S: inner$ and $n$ be a \rmnw. If $n \in \mathcal{L}(\vp{p})$ then $n \in \Policy{p}$, \textit{i.e.,} $\mathcal{L}(\vp{p}) \subseteq \Policy{p}$. Here, $\mathcal{L}(\vp{p})$ is the set of \rmnws accepted by the visibly pushdown automaton $\vp{p}$.
\end{lemma}
\begin{proof}
By definition of $\vp{p}$, $\mathcal{A}_1$ is the automaton for $(\basealpha^*S)$ and  $Q^1$ is the set of all its states, $q_{init}^1$ is its initial state, and $F^1$ is the set of final states. $F^{p_i}$ is the set of final states for the VPA of the inner policy $p_i$ and $q_{beg}^{i}$ is its initial state.\\
Given a \rmnw ~$n = (a_1 \dots a_l, \nu)$, let the accepting run of 
\[\vp{p} = \mathcal{M}_{p} = (Q,~ q_{beg},~ \{q_{end}\},~ \Sigma,~ \Gamma,~ \bot,~ \delta_{c},~  ~ \delta_{r})\] 
be 
\[\rho(\mathcal{M}_{p}, n) = (q_{beg}, \theta_0 = \bot) \dots (q_{end}, \bot).\]
We prove the above lemma by contradiction. Since $n$ was accepted by $\vp{p}$, no configuration can have state $q_{rej}$. Suppose there exists a path $\pi(n, a_i)$ such that $Calls(\pi(n, a_i)) \in \mathcal{L}(\basealpha^*S)$, but $n[\idx(a_i), \nu(\idx(a_i))] \notin \mathcal{L}(\vp{inner})$. Since $Calls(\pi(n, a_i)) \in FirstMatch(n, \basealpha^*S)$, the state in the last configuration of $\basealpha^*S$'s DFA will be its final state. By induction on the length of the prefix of the path $\pi(n, a_i)$, we can show that the state $q_{i-1}$ in the VPA's run is the same as the state of the second to last configuration in the run of the DFA on $Calls(\pi(n, a_i))$. Therefore, $q_{i-1} \in Q^1$ such that on reading $a_i$ the VPA will transition to a state in $q^{in} \in Q^{in}$, where $(q_{beg}^{in}, \cs, q^{in}, q_{beg}) \in \delta^{in}_{c}$. The state $q_{i-1}$ will be pushed on the stack. In the run $\rho(\mathcal{M}_p, n)$, the call transitions that will be applied on the symbols between $a_i, \ldots, a_{\nu(\idx(a_i))}$ will be the call transitions of $\vp{inner}$ and similarly the return transitions will be the same as the return transitions of $\vp{inner}$. If $n[\idx(a_i), \nu(\idx(a_i))] \notin \mathcal{L}(\vp{inner})$, the state $q'=  q_{\nu(\idx(a_i))-1}$ will not satisfy $(q', q_{i-1}, \rs, q_{i-1})$, and instead $q_{\nu(\idx(a_i))} = q_{rej}$. This contradicts that no configuration has $q_{rej}$ as its state. Therefore, the subtree of every path matching $\basealpha^*S$ will satisfy the $inner$ policy. 
\end{proof}
\begin{lemma}\label[lemma]{lem:start-dir2}
Let $p$ be a policy of the form $\start S: inner$ and $n$ be a \rmnw. If $n \in \Policy{p}$ then $n \in \mathcal{L}(\vp{p})$, \textit{i.e.,} $\Policy{p} \subseteq \mathcal{L}(\vp{p})$. Here, $\mathcal{L}(\vp{p})$ is the set of \rmnws accepted by the visibly pushdown automaton $\vp{p}$.
\end{lemma}
\begin{proof}
We prove this by induction on the number of paths that match $FirstPath(n, \basealpha^*S)$. The key idea is that in any path $\pi(n, a_i)$ in the first match set, the state at $q_{i}$ is some state in $Q^{in}$ which is a  state from the initial state of $\vp{inner}$. The run of the $\vp{p}$ between the symbols of the subtree $n[\idx(a_i), \nu(\idx(a_i))]$ is equivalent to the run of the $\vp{inner}$. Therefore, by induction hypothesis, since the subtree is in the language of $\Policy{inner}$, the subtree is accepted by $\vp{inner}$. Based on the rules, upon reading the symbol $a_{\nu(\idx(a_i))}$, $\vp{p}$ transitions to the state $q_{i-1}$. We can show that eventually on the last symbol $a_l$, either the source state will be certain state in $Q^1$ or $Q^{in}$ that goes to the $q_{end}$ state.
\end{proof}

\begin{theorem} \label{thm:ap-distmon}
Given a nested word $w$ over some call-return augmented alphabet $\Sigma$ and a VPA $\mathcal{M}$ whose run on $w$  is  $\rho(w) = (q_1, \theta_1) \ldots (q_n, \theta_n)$ then:
\begin{enumerate}
\item the run of $\mathcal{M}$'s centralized monitor is $\llbracket \mathcal{M}\rrbracket_{central} (q_1, \theta_1)~ w =  (q_n, \theta_n)$, and 
\item the run of $\mathcal{M}$'s distributed monitor $\llbracket \mathcal{D}\rrbracket_{dist} (q_1, \theta_1)~w =  (q_n, \theta_n)$.
\end{enumerate}
\end{theorem}
\begin{proof}
The proof in case of both the centralized and decentralized monitor goes by induction on the length of the nested word $w$. We will show the more interesting result (2) of the above theorem.
\begin{itemize}
\item Base case $w = (\epsilon, \phi)$ or the empty string:
The run of a VPA on an empty string keeps it in the initial configuration just like the distributed monitor.
\item Let the inductive hypothesis be that for a nested word $w = (a_1 \ldots a_{k}, \nu)$ the following holds:
\[\rho(w) = (q_1, \theta_1) \ldots (q_k, \theta_k)\]
\[\llbracket \mathcal{D}\rrbracket_{dist} (q_1, \theta_1)~w =  (q_k, \theta_k)\]

Now we need to show that the run of VPA and the distributed monitor are consistent for a nested word $w' = (a_1 \ldots a_{k} a_{k+1}, \nu')$, where $\nu' = \nu \cup \{(\idx(a_{k+1}), +\infty)\}$ for some natural number $j$ if $a_{k+1}$ is a call symbol, and if $a_{k+1}$ is a return symbol then $\nu' = \nu \cup \{(i, \idx(a_{k+1}))\}$ for some natural number $i$ or $-\infty$.
We prove the above for the cases when:
\begin{enumerate}
\item Case $a_{k+1} = \langle e$ is a call:
The distributed monitor will pick the local monitor $\mathcal{D}(e)$ and by definition of the single step transition operator, it will go from state $(q_k, \theta_k)$ to the state $(q_{k+1}, \theta_{k+1})$---same as that of the VPA.
\item Case $a_{k+1} = e \rangle$:
The proof in this case is similar to the above case.
\end{enumerate}
\end{itemize}
\end{proof}

\section{Complexity Results} \label{app:complexity}
\begin{definition}[Depth of a policy]
Given a policy, the depth is inductively defined as:
\begin{align*}
& \depth{\start S: inner} = 1 + \depth{inner}&\\
&\depth{\match~\regexone \allpath \regextwo} = 1  &\\
&\depth{\match~ \regex \allchildren p} = 1 + \depth{p} & \\
&\depth{\match ~\regex \existschild p_1 ~\textsf{\textbf{then}}\ldots \textsf{\textbf{then}}~p_k} = 1 + \max \{\depth{p_1}, \ldots, \depth{p_k}\}&\\
&\depth{\callseq~\regex} = 1&
\end{align*}
\end{definition}

\begin{definition}[Maximum fan-out of a policy]
The maximum fan-out for a policy is inductively defined as:
\begin{align*}
& \fanout{\start S: inner} = 1 + \fanout{t_{inner}}&\\
&\fanout{\match~\regexone \allpath \regextwo} = 2  &\\
&\fanout{\match~ \regex \allchildren p} = 1 + \fanout{p} & \\
&\fanout{\match ~\regex \existschild p_1 ~\textsf{\textbf{then}}\ldots \textsf{\textbf{then}}~p_k} = 1+\max \{k, \fanout{p_1}, \dots, \fanout{p_k}\}&\\
&\fanout{\callseq~\regex} = 1&
\end{align*}
\end{definition}
\begin{definition}[Greatest DFA states]
Given a policy, the maximum number of states across all the DFAs of regular expressions present in the policy can be inductively defined as below. Here, $\autsize{r}$ is the number of states in the DFA of $r$.

\begin{align*}
& \lreg{\start S: inner} = \lreg{inner}&\\
&\lreg{\match~\regexone \allpath \regextwo} = \max \{\autsize{\regexone}, \autsize{\regextwo}\} &\\
&\lreg{\match~ \regex \allchildren p} = \max \{\autsize{\regex}, \lreg{p}\} & \\
&\lreg{\match ~\regex \existschild p_1 ~\textsf{\textbf{then}}\ldots \textsf{\textbf{then}}~p_k} = \max \{\autsize{\regex}, \lreg{p_1}, \dots, \lreg{p_k}\}&\\
&\lreg{\callseq ~\regex} = \autsize{\regex}&
\end{align*}
\end{definition}

\begin{theorem}

Suppose the DFA of every regular expression in a policy $p$ has at most $R$ states; each sub-policy has at most $k$ immediate sub-expressions, \textit{i.e.}, fan-out at most $k$; and (nesting) depth $d$. Then the VPA $\mathcal{M}(p)$ such that $\nw{p} = \mathcal{L}(\mathcal{M})$ has $\mathcal{O}((k+1)^d R)$ states. 
\end{theorem}
\begin{proof}\label{proof:complexity}
We show this by proving that the number of states for any policy $p$ is bounded by\\ $(k+1)^d(R+5)$.
This proof goes by structural induction on the AST of the policy $p$: 
\paragraph*{Base Cases:}
\begin{itemize}
\item Linear policy, $\callseq ~\regex$: 
By definition, $d = 1$, $k=1$, and $R = \autsize{reg}$. The number of states in $\mathcal{M}$ is $3+\autsize{reg}$, which is less than $(k+1)^d(\autsize{reg}+5)$.
\item Hierarchical policy, $\match~\regexone \allpath
 \regextwo$: By definition, $d = 1$, $k=2$, $\autsize{reg_1}, \autsize{reg_2} \le R$. The number of states in $\mathcal{M}$ is $\autsize{reg_1} + 2\autsize{reg_2} + 4$, which is less than $(k+1)^d(R+5) = 3(R+ 5)$.
\end{itemize}
\paragraph*{Inductive Cases:}
The proof for all inductive cases go similarly:
\begin{itemize}
\item All child policy, $\match~ \regex \allchildren p$:
Let $\mathcal{M}_p$ be the VPA of $p$ with $m$ states. 
By definition, $d = 1 + \depth{p}$, $k=1+\fanout{p}$, $\autsize{reg} \le R$. The number of states in $\mathcal{M}$ is $\autsize{reg_1} + m + 4$. We can show that the number of states in $\mathcal{M}$ satisfy the theorem statement by using the IH $m \le (\fanout{p} + 1)^{\depth{p}}(\lreg{p} + 5)$.
\item Exists child policy, $\match ~\regex \existschild p_1 ~\textsf{\textbf{then}}\ldots \textsf{\textbf{then}}~p_k$:
The proof goes similarly using the IH about the bound on number of states in the VPA of each sub-policy $p_i$.
\item Start policy, $\start S:~ inner$:
Similar to the above case using the IH on the number of states in $inner$'s VPA.
\end{itemize}
Since the constructed VPAs are total, the number of transitions will be quadratic in the order of number of states because the number of stack symbols is atmost number of states.
\end{proof}